%% file: main.tex
\documentclass[acmsmall]{acmart}
\usepackage{acm-ec-23}
\usepackage{booktabs} 
\usepackage[ruled]{algorithm2e} 
\usepackage{algpseudocode}
\usepackage{graphicx}
\usepackage{subfigure}
\usepackage{amsmath}
\usepackage{amsthm}
\usepackage{bm}

\SetAlFnt{\small}
\SetAlCapFnt{\small}
\SetAlCapNameFnt{\small}
\SetAlCapHSkip{0pt}
\IncMargin{-\parindent}

\input{macro.tex}

\setcitestyle{acmauthoryear}

\title{The Wisdom of Strategic Voting}

\author{Qishen Han}
\email{hnickc2017@gmail.com}
\orcid{0000-0003-0268-6918}
\affiliation{%
  \institution{Rensselaer Polytechnic Institute}\city{Troy}\state{NY}\country{USA}}
\author{Grant Schoenebeck}
\email{schoeneb@umich.edu}
\orcid{0000-0001-6878-0670}
\affiliation{%
  \institution{University of Michigan}\city{Ann Arbor}\state{MI}\country{USA}}
\author{Biaoshuai Tao}
\email{bstao@sjtu.edu.cn}
\orcid{0000-0003-4098-844X}
\affiliation{%
  \institution{Shanghai Jiao Tong University}\city{Beijing}\country{China}}
\author{Lirong Xia}
\email{xialirong@gmail.com}
\orcid{0000-0002-9800-6691}
\affiliation{%
  \institution{Rensselaer Polytechnic Institute}\city{Troy}\state{NY}\country{USA}}

\begin{abstract}
We study the voting game where agents' preferences are endogenous decided by the information they receive, and they can collaborate in a group. We show that strategic voting behaviors have a positive impact on leading to the ``correct'' decision, outperforming the common non-strategic behavior of informative voting and sincere voting. Our results give merit to strategic voting for making good decisions. 

To this end, we investigate a natural model, where voters' preferences between two alternatives depend on a discrete state variable that is not directly observable. Each voter receives a private signal that is correlated with the state variable. 
We reveal a surprising equilibrium between a strategy profile being a strong equilibrium and leading to the decision favored by the majority of agents conditioned on them knowing the ground truth (referred to as the {\em informed majority decision})
: as the size of the vote goes to infinity, every $\varepsilon$-strong Bayes Nash Equilibrium with $\varepsilon$ converging to $0$ formed by strategic agents leads to the informed majority decision with probability converging to $1$. On the other hand, we show that informative voting leads to the informed majority decision only under unbiased instances, and sincere voting leads to the informed majority decision only when it also forms an equilibrium.
\end{abstract}

\ccsdesc[500]{Theory of computation~Algorithmic game theory}
\ccsdesc[500]{Computing methodologies~Distributed artificial intelligence}

\acmYear{2023}\copyrightyear{2023}
\acmConference[EC '23]{Proceedings of the 24th ACM Conference on Economics and Computation}{July 9--12, 2023}{London, United Kingdom}
\acmBooktitle{Proceedings of the 24th ACM Conference on Economics and Computation (EC '23), July 9--12, 2023, London, United Kingdom}
\acmPrice{15.00}
\acmDOI{10.1145/3580507.3597681}
\acmISBN{979-8-4007-0104-7/23/07}

\begin{document}


\maketitle


\section{Introduction}
{\input{Sections/1_intro}}

\section{Models and Preliminaries}
{\input{Sections/2_setting}
}

\section{Equivalence between High Fidelity and Strong Equilibrium}
{\input{Sections/3_result}
}
\section{Probability Analysis on Fidelity}
{\input{Sections/4_result2}
}
\section{Non-binary World States and Non-binary Signals}
{\input{Sections/5_nonbinary.tex}}
\section{Conclusion and Future Work}
{\input{Sections/6_conclusion}}

\section*{Acknowledgments}
We thank all anonymous reviewers for their helpful comments and suggestions. GS acknowledges NSF \#2007256 and NSF \#1452915 for support.
LX acknowledges NSF \#1453542, \#2007994, and \#2106983, and a Google Research Award for support. The research of Biaoshuai Tao was supported by the National Natural Science Foundation of China (Grant No. 62102252).

%
%
%
%
%

\bibliographystyle{ACM-Reference-Format}
\bibliography{references,newref} 

\clearpage
\appendix
\section{Comparison with Feddersen and Pesendorfer's Work}
{\input{Appendix/1_FP_comparison}}
\section{Frequently used notations in the binary setting}
{\input{Appendix/2_notation_list}}
\section{Proof of Proposition~\ref{prop:no_BNE}}
{\input{Appendix/2_Prof1_Proof.tex}}
\section{Non-binary setting}
{\input{Appendix/3_nb_setting}}
\section{Non-binary results}
{\input{Appendix/4_nb_results}
}

\end{document}

%% file: macro.tex
\newtheorem{thm}{Theorem}
\newenvironment{thmnb}[1]
  {\renewcommand{\thethm}{\ref{#1}}%
   \addtocounter{thm}{-1}%
   \begin{thm}}
  {\end{thm}}

\newtheorem{prop}{Proposition}
\newenvironment{propnb}[1]
  {\renewcommand{\theprop}{\ref{#1}}%
   \addtocounter{prop}{-1}%
   \begin{prop}}
  {\end{prop}}

\newtheorem{lem}{Lemma}

\newtheorem{claim}{Claim}
\newtheorem{asm}{Assumption}

\newtheorem{ex}{Example}
\newtheorem*{remark}{Remark}

\newtheorem{coro}{Corollary}

\theoremstyle{definition}
\newtheorem{dfn}{Definition}

\newcommand{\bA}{\mathbf{A}}
\newcommand{\bR}{\mathbf{R}}
\newcommand{\stg}{\sigma}
\newcommand{\stgp}{\Sigma}
\newcommand{\hthd}{f}
\newcommand{\thd}{\hat{f}}
\newcommand{\acc}{A}
\newcommand{\err}{I}
\newcommand{\wos}{k} 
\newcommand{\Wos}{K} 
\newcommand{\Wosset}{\mathcal{W}} 
\newcommand{\Wosrv}{W} 
\newcommand{\sig}{m} 
\newcommand{\Sig}{M} 
\newcommand{\Sigset}{\mathcal{S}} 
\newcommand{\Sigrv}{S} 
\newcommand{\ut}{u} 
\newcommand{\vt}{v} 
\newcommand{\agf}{\alpha_F}
\newcommand{\agc}{\alpha_C}
\newcommand{\agu}{\alpha_U}
\newcommand{\hagf}{\hat{\alpha}_F}
\newcommand{\hagc}{\hat{\alpha}_C}
\newcommand{\hagu}{\hat{\alpha}_U}
\newcommand{\aga}{\alpha^{\bA}}
\newcommand{\agr}{\alpha^{\bR}}
\newcommand{\agset}{N}
\newcommand{\ag}{N}
\newcommand{\sag}{n}
\newcommand{\Thd}{\mu}
\newcommand{\Low}{L}
\newcommand{\High}{H}
\newcommand{\Lowset}{\mathcal{L}}
\newcommand{\Highset}{\mathcal{H}}
\newcommand{\lp}{\lambda}
\newcommand{\bp}{\beta}
\newcommand{\tf}{N_F}
\newcommand{\tc}{N_C}
\newcommand{\tu}{N_U}
\newcommand{\negT}{\mathcal{N}}
\newcommand{\xrv}{X}
\newcommand{\yrv}{Y}
\newcommand{\vart}{s_N}
\newcommand{\inst}{\mathcal{I}}
\newcommand{\exshare}{excess expected vote share}
\newcommand{\exshares}{excess expected vote shares}

%% file: Sections/1_intro.tex
Today,  voting is used to make an array of binary decisions permeating nearly every corner of life including in recall/run-off public elections, adoption of decrees by religious institutions, decisions by corporate boards on whether or not to pursue a new strategy/acquisition/etc, hiring and by-law decisions at university, and public entertainments like talent shows.  In most cases, the voting is attempting to aggregate both the agents' preferences and knowledge. A key aspect of this setting is that agents have preferences over outcomes contingent on some underlying state that they cannot directly observe, and the goal is to make a ``good'' decision that reflects the real preferences of the agents.

\begin{ex}
\label{ex:motive}
Suppose the voters vote to decide the policy towards the COVID-19 pandemic. The two choices are to accept the more-restrictive policy (Accept) and to keep the status quo (Reject). The consequence of the policy depends on the fact that the COVID virus is of high or low risk, and more people tend to accept the policy when COVID is of high risk than when COVID is of low risk. The voters do not know the risk level of the virus directly. Instead, every voter forms a private judgment on the risk level based on his/her own information sources. voters may have different opinions on whether to accept the policy which may or may not depend on the risk level. Can the voters achieve a good decision via the majority vote?
\end{ex}

Three different lines of work aim to address this problem under different models and with different goals. The first line of work is axiomatic social choice~\citep{Arrow1951:social,Plott76:Axiomatic}, where agents' preferences are exogenously given, and the goal is to design voting rules that satisfy desiderata, often called {\em axioms}, especially when agents {\em sincerely} report their preferences. The second line of work is along the extensions of the Condorcet Jury Theorem~\citep{Condorcet1785:Essai}, where agents' preferences are endogenous and depend on the information structure and the signals they receive. The goal is to design mechanisms to reveal the true state of the world, especially when agents vote {\em informatively}, i.e., their votes honestly reflect the private signals they receive. The survey by~\citet{Nitzan17:Collective} provides a comprehensive overview. The third line of work originated from~\citet{Feddersen97:Voting}, where agents' preferences are endogenous as in the second line of work, yet the goal is different. Instead of revealing the true state of the world, the goal is to achieve {\em informed majority decision}, which  is the decision favored by the majority of the agents if the world state were known to them. Our paper is along the third line of work.


\subsection{Strategic Behaviors}


Previous work shows that a good decision can be reached when agents follow sincere or informative behaviors. However, when agents are {\em strategic}, they may have incentives to deviate from sincere or informative voting to achieve a preferred result with a higher probability. 
This is not a problem for axiomatic social choice, as strategic agents will always vote for their preferred alternative in binary voting~\citep{Barbera1991:commitee}. However, when agents have preferences decided by uncertain world states, the surprising result by~\citet{Austen96:Information} shows that even in binary voting, informative voting may fail to form a Nash equilibrium. The key insight is that an agent's vote makes a difference only when all other votes form a tie, which means that when an agent strategically thinks about his/her vote, effectively he/she gains more information about the ground truth (by assuming that other votes are tied). This is illustrated in the following example.

\begin{ex}
   Consider an instance of the COVID policy problem, where the utility of the agents and signal distribution of different risk levels are shown in the tables below. 

\vspace{0.3cm}
\begin{minipage}[c]{0.5\textwidth}
\centering
\begin{tabular}{@{}ccc@{}}
\toprule
State    & High Signal & Low Signal \\ \midrule
High Risk & 0.9         & 0.1        \\
Low Risk  & 0.4         & 0.6        \\ \bottomrule
\end{tabular}
\captionof{table}{Signal distributions.}
\end{minipage}
\begin{minipage}[c]{0.5\textwidth}
\centering
\begin{tabular}{@{}ccc@{}}
\toprule
State    & Accept & Reject  \\ \midrule
High Risk & 1         & 0       \\
Low Risk  & 0         & 1       \\ \bottomrule
\end{tabular}
\captionof{table}{Agents' utilities.}
\end{minipage} 
  
    
    Suppose all but one agents are informative and the remaining agent is strategic. Informative agents vote for Accept when they receive a high signal and Reject when they receive a low signal. The strategic agent only cares about the pivotal case where exactly half of the informative agents vote for Accept. However, given that agents receive high signals with a probability of 0.9 given the risk level being high, the pivotal case implies a high probability that the risk level is low, and the strategic agent will vote for Reject even after receiving a high signal. 
\end{ex}

The above deviation of strategic binary voting with preferences endogenously affected by the unobservable world state sharply contrasts with the axiomatic social choice where preferences are exogenously given. In the latter case, the majority rule is strategy-proof while the former case attracts a large literature to study the binary voting problem under game theoretical contexts, studying the impact of strategic behavior. For the truth-revealing goal, \citet{Wit98:Rational} and \citet{Myerson98:Extended} show that a selected equilibrium with mixed strategy reveals the world state with high probability. \citet{Feddersen1998:innocent} show the existence of such equilibrium in any 
non-unanimous voting, while in unanimous voting strategic voting has a constant probability to make a mistake. And for the informed majority decision, \citet{Feddersen97:Voting} adopt a model with continuous world states and an asymptotically large number of agents whose preferences are drawn from a distribution with full support on a continuum and show that the equilibrium is unique and always leads to the informed majority decision with high probability. \citet{Schoenebeck21:wisdom} proposes a mechanism incentivizing informative voting from agents and leading to the informed majority decision with high probability. 

Nevertheless, there are two aspects not addressed by previous works. 
Firstly, previous works (except for \citet{Schoenebeck21:wisdom}) focus on Nash equilibrium which allows only individual manipulation. In real-world scenarios, on the other hand, such strategic manipulation often occurs in a {\em coalition} of agents. Coalitional manipulation is more powerful than individual manipulation as it allows multiple agents to coordinate and deviate at the same time.
The following example shows that a Nash equilibrium is still prone to a group of manipulators in binary voting. 

\begin{ex}
\label{ex:motive3}
    Consider an instance with three agents, whose utility is shown as follows. 
    
\begin{minipage}{\linewidth}
\vspace{0.3cm}
\centering
\begin{tabular}{@{}ccccc@{}}
\toprule
Agent & (High, Accept) & (High, Reject) & (Low, Accept) & (Low, Reject) \\\midrule
1     & 1            & 0            & 1           & 0           \\
2     & 1            & 0            & 0           & 1           \\
3     & 0            & 1            & 0           & 1 \\\bottomrule
\end{tabular}
\captionof{table}{The utility of three agents under different states and decisions.}
\end{minipage}

The following strategy profile is a Nash equilibrium: agent 1 always votes for Accept, agent 2 votes informatively, and agent 3 always votes for Reject. Here, agent 1 and 3 play their dominant strategies, and, consequently, informative voting is the best strategy for agent 2. However, this strategy profile is dominated by the profile where all agents vote informatively. Under informative voting, the decision in accord with the state (Accept in High state, and Reject in Low state) is selected with a larger probability, and the utility of agent 2 increases. On the other hand, the overall probability of choosing to Accept or Reject does not change, so agent 1 and 3's utilities remain the same. 
\end{ex}

Secondly, previous works focus on the existence of certain equilibria that achieves the goal (revealing the world state or reaching the informed majority decision). However, the existence of multiple equilibria~\citep{Wit98:Rational}, including ``bad equilibria'' that do not lead to the goal, makes the behavior of strategic agents unpredictable, as it is uncertain which equilibrium agents will play. 
One response to multiple equilibria is to select an equilibrium that is more ``natural'' or ``reasonable'' than others, named {\em equilibrium selection}. However, equilibrium selection cannot guarantee that agents will play the selected equilibrium, as it is unclear which equilibrium is more ``natural'' or ``reasonable'' in many scenarios, and agents may not agree on a "more natural" equilibrium even if it exists.

As a consequence, the following research question remains unanswered: {\bf does binary voting always lead to the informed majority decision with coalitional strategic agents?}

\subsection{Our contribution}
We give a surprising confirmative answer to this question under mild conditions. We show that coalitional strategic behaviors positively impact achieving the informed majority decision and outperform non-strategic voting. We show that every equilibrium leads to the informed majority decision, and every voting profile that leads to the informed majority decision is an equilibrium. 
On the contrary, non-strategic behaviors lead to the informed majority decision only under certain conditions. Our results give merit to strategic behaviors and extend Feddersen and Pesendorfer's results to settings with coalitional strategic agents.  

We study the solution concept of {\em $\varepsilon$-strong Bayes Nash Equilibrium}, which precludes groups of agents from reaching higher expected utilities by coordinating. We show the equivalence of a strategy profile being ``good'' (leading to the informed majority decision with high probability, or, equivalently, of high {\em fidelity}) and being an $\varepsilon$-strong Bayes Nash Equilibrium with $\varepsilon = o(1)$ (Theorem~\ref{thm:acc}). We also guarantee the existence of an $\varepsilon$-strong Bayes Nash with $\varepsilon = o(1)$ in any instance (Theorem~\ref{thm:deviate}). 

On the other hand, we characterize the conditions where strategy profiles succeed and fail to achieve the informed majority decision. Applying these results, we study two common non-strategic behavior -- \emph{informative voting}, where agents honestly reflect their private information in their votes, and \emph{sincere voting}, where agents vote as if they are the only decision-maker. We show that (1) informative voting leads to the informed majority decision only when the majority vote threshold is unbiased compared with the signal distribution (Corollary~\ref{coro:informative}), and (2) sincere voting leads to the informed majority decision only when it is also an equilibrium (Corollary~\ref{coro:sincere}). These observations indicate that strategic behavior ``prevails'' over non-strategic behaviors in binary voting!

The technical key for the probability analysis is to compute the \emph{\exshare}, i.e., the amount of expected vote share an alternative attracts that exceeds the threshold, and to upper (or lower) bound the fidelity given different cases of \exshares. A strategy profile has high fidelity if and only if its \exshare{} is strictly positive (Theorem~\ref{thm:arbitrary}). 



We follow the setting in \citet{schoenebeck21wisdom}, which is an extension of the setting in \citet{Austen96:Information}, and consider agents with preferences contingent on underlying world states in a single framework. 
Also, as in Example~\ref{ex:motive} and previous work, we assume that various constraints in the real world prevent discussion after agents see their signals. Such constraints can be of a time aspect (a quick decision must be made and there is no time for discussion), a procedural aspect (a formal conference that prohibits participants from discussing privately before voting), and/or a societal aspect (it is socially unsuitable to discuss some preferences), etc.  Therefore, we consider an {\em ex-ante} setting where the expected utilities are computed before agents receive their signals.

\subsection{Related Work}
The famous Jury Theorem from \citet{Condorcet1785:Essai} has ``formed the basis
for the development of social choice and collective decision-making as modern research
fields''~\citep{Nitzan17:Collective}. The theorem states that a group of decision-makers could reveal the correct world state with a higher probability than any individual in the group, and such probability converges to 1 as the number of group members increases. A large literature on collective decision-making has followed Condorcet's path trying to extend the result into more general models \citep{Miller86:Information, Grofman83:Thirteen, Owen89:Proving, Boland89:Modelling}. 

The game-theoretical study of the Condorcet Jury Theorem starts from \citet{Austen96:Information}. Austen-Smith and Banks study a collaborative voting game where each agent shares the same preference and receives a binary signal correlated with an unknown binary state of the world. However, even in this case, they showed that sincere voting and informative voting do not always form a Nash Equilibrium. As a consequence, the following works focus on the effect of strategic behavior in the majority vote and propose equilibria that reveal the ground truth~\citep{ Myerson98:Extended, Wit98:Rational, Duggan01:Bayesian,meirowitz2002informative,Feddersen1998:innocent}. \citet{Feddersen97:Voting} adopt a similar information structure with the game theoretical study of Condorcet Jury Theorem but aim to achieve a different goal of informed majority decision. We distinguish our work from Feddersen and Pesendorfer's in Appendix~\ref{apx:fp}.
Other generalizations of the Condorcet Jury Theorem include dependent agents~\citep{Nitzan84:Significance, Shapley84:Optimizing, Kaniovski10:Aggregation}, agents with different competencies~\citep{Nitzan80:Investment,Gradstein87:Organizational, Ben11:Condorcet}, and voting with more than two alternatives~\citep{Young88:Condorcet, Goertz14:Condorcet}. 

Another line of work related to collective decision-making focuses on designing mechanisms that lead to the correct decision. Recent work shows the reliability of the ``surprisingly popular'' answer when agents are sincere~\citep{prelec2017solution,Hosseini2021suprisingly} and strategic~\citep{schoenebeck21wisdom}. In particular, \citet{schoenebeck21wisdom} adopt the ``surprisingly popular'' technique into a social choice context with strategic agents, and propose a truthful mechanism to aggregate information. They show that even in a setting where agents have subjective preferences contingent on an objective underlying state, their mechanism reveals the informed majority decision with high probability and is an (ex-ante) $\varepsilon$-strong Bayes Nash Equilibrium with $\varepsilon$ converging to 0 at an exponential rate. Our work follows the setting in Schoenebeck and Tao's work, but our work is different in that the aggregation happens implicitly because agents are acting strategically rather than because a mechanism explicitly selects a surprisingly popular answer. 

Our work is also related to {\em information elicitation}, which aims to collect truthful and high-quality information from agents under a noisy information structure. Information elicitation is well developed with multiple lines of research focusing on different aspects of the problem, including scoring rules~\citep{bickel2007some,Gneiting2007:Scoring}, peer prediction mechanisms~\citep{Miller05:Eliciting,schoenebeck2021:learning,schoenebeck2021information}, Bayesian Truth Serum~\citep{Prelec04:BTS, Witkowski12:Peer}, and prediction markets~\citep{Miller88:Markets,Wolfers2004market}. Unfortunately, information elicitation is incompatible with the voting scenario in our paper for two reasons. Firstly, information elicitation requires agents to be indifferent to the outcome, while agents are incentivized by the outcome of the vote. Secondly, information elicitation uses payments to reward the agents, while voting does not have monetary rewards. 


%% file: Sections/2_setting.tex
\label{sec:setting}
We first present our model and results with binary world states and binary private signals, which convey the main ideas of this work while also hiding much of the complexity. The general extension into the non-binary setting is in Section~\ref{sec:nb}. We follow the setting in \citet{schoenebeck21wisdom} and consider agents with subjective preferences contingent on an objective underlying state in one framework.  

\paragraph{Alternatives and World States.} $\ag$ agents vote for two alternatives $\bA$ (standing for ``accept'') and $\bR$ (standing for ``reject''). There are $\Wos = 2$ possible world states $\Wosset = \{L, H\}$ (standing for ``low risk'' and ``high risk'' respectively), where $\bA$ is more preferred in $H$, and $\bR$ is more preferred in $L$. We use $\wos$ to denote a generic world state. The world state is not directly observable by the agents. Let $P_{H} = \Pr[\Wosrv = H]$ and $P_{L} = \Pr[\Wosrv = L]$ be the common prior of the world states. We assume $P_H > 0$ and $P_L > 0$. 

\paragraph{Private Signals.} Every agent receives a signal in $\Sigset = \{l, h\}$. We use $\sig$ to denote a generic signal, and $\Sigrv_\sag$ to denote the random variable representing the signal that agent $\sag$ receives. We assume the signals agents receive are independent and have identical distributions conditioned on the world state. Let $P_{\sig\wos} = \Pr[\Sigrv_{\sag} = \sig \mid \Wosrv = \wos]$ be the probability that an agent receives signal $\sig$ under world state $\wos$. The signal distributions $((P_{hH}, P_{lH}), (P_{hL}, P_{lL}))$ are also common knowledge.  We assume that the signals are positively correlated to the world states. Specifically, we have $P_{hH} > P_{hL}$ and $P_{lH} < P_{lL}$. On the other hand, we allow biased signals and DO NOT assume $P_{hH} > P_{lH}$ or $P_{hL} < P_{lL}$. 

\paragraph{Majority Vote.} This paper considers the majority vote with threshold $\Thd$. Each  agent $\sag$  votes for $\bA$ or $\bR$. If at least $\Thd\cdot \ag$ agents vote for $\bA$, $\bA$ is announced to be the winner; otherwise, $\bR$ is announced to be the winner.

\paragraph{Utility and Types of Agents.} Each agent $\sag$ has a utility which is a function of the true world state and the outcome of the vote. Formally, we have  $\vt_{\sag}: \Wosset \times \{\bA, \bR\}\to \{0,1,\ldots, B\}$, where $B$ is the positive integer upper bound. We assume that $\bA$ is more preferable in $H$ than in $L$, and $\bR$ is the opposite: for every agent $\sag$, $\vt_{\sag}(H, \bA) > \vt_{\sag}(L, \bA)$ and $\vt_{\sag}(H, \bR) < \vt_{\sag}(L, \bR)$. 

The different endogenous preferences of agents are reflected by different utility functions. {\em Predetermined} agents always prefer the same alternative, and {\em contingent} agents have preferences depending on the world state. Predetermined agents can be further divided into {\em friendly} and {\em unfriendly} agents based on the alternative they prefer. For an agent $\sag$, if $\sag$ is a friendly agent, $\vt_{\sag}(H,\bA) > \vt_{\sag}(L, \bA) > \vt_{\sag}(L, \bR) > \vt_{\sag}(H, \bR)$; if $\sag$ is an unfriendly agent, $\vt_{\sag}(L, \bR) > \vt_{\sag}(H, \bR) > \vt_{\sag}(H,\bA) > \vt_{\sag}(L, \bA)$; and if $\sag$ is a contingent agent, $\vt_{\sag}(H, \bA) > \vt_{\sag}(H, \bR)$ and $\vt_{\sag}(L,\bR) > \vt_{\sag}(L, \bA)$. 

Let $\agf, \agu$, and $\agc$ be the approximated fraction of each type of agent. Formally, given $\ag$ agents, $\tf = \lfloor \agf\cdot \ag \rfloor$ is the number of friendly agents, $\tu = \lfloor \agu\cdot \ag \rfloor$ is the number of unfriendly agents, and $\tc = \ag - \tf - \tu$ is the number of contingent agents.  $\agf, \agu$, and $\agc$ are common knowledge and do not depend on $\ag$.



\paragraph{Informed Majority Decision} The goal of the voting is to output the informed majority decision, which is the alternative favored by the majority of the agents if the world state were known. The informed majority decision shares the same threshold $\Thd$ as the majority vote threshold. If $\bA$ is preferred by at least $\mu\cdot \ag$ agents, then $\bA$ is the informed majority decision; otherwise, $\bR$ is the informed majority decision. 

In this paper, we assume that neither friendly agents nor unfriendly agents can dominate the vote. Otherwise, the informed majority decision does not depend on the state and one coalition can always enact it via a dominant strategy.
As a result, $\bA$ is the informed majority decision when the world state is $H$, and $\bR$ is the informed majority decision when the world state is $L$.

\begin{ex}
\label{ex:setting} Consider the COVID policy-making scenario. $\ag=20$ voters decide whether to accept (denoted as $\bA$) or reject (denote as $\bR$) the more-restrictive policy. The world state $\{L, H\}$ describes the real risk level of the virus. $\Wosrv = H$ means high risk level, and $\Wosrv = L$ means low risk level. The voters' beliefs form a common prior based on some preliminary reports. Suppose $P_H = 0.4$ and $P_L = 0.6$, which means the risk level has a prior probability of 0.4 to be high. 

Every voter receives a private signal $l$ or $h$ from his/her information sources. The signals somehow reflect the risk level but are noisy. Suppose this is a biased scenario (for example, there has been a boost of positive cases in the past week), and members are always more likely to receive the high signal. For example, $P_{hH} =0.8$ and $P_{hL} = 0.6$, i.e., a voter will receive an $h$ signal with probability 0.8 if the risk level is high and receive an $h$ signal with probability 0.6 if the risk level is low. 

\begin{minipage}{\linewidth}
\vspace{0.3cm}
\centering
\begin{tabular}{@{}ccccc@{}}
\toprule
(Winner, World State)           & $(\bA, H)$         & $(\bA, L)$         & $(\bR, H)$         & $(\bR, L)$         \\ \midrule
Friendly agent   & 8           & 6           & 2           & 4           \\
Unfriendly agent & 3           & 1           & 5           & 8           \\
Contingent agent & 3           & 2           & 1           & 8           \\ \bottomrule
\end{tabular}
\captionof{table}{Utility of agents in Example~\ref{ex:setting}.\label{tab:util}}
\end{minipage}

The majority vote threshold is $\Thd = 0.6$. Therefore, $\bA$ is the winner if and only if at least 12 voters vote for it. There are 4 friendly voters, 6 unfriendly voters, and 10 contingent voters. The informed majority decision depends on the world state: ``accept'' is the informed majority decision if the world state is $H$, and ``reject'' is the informed majority decision if the world state is $L$. 

We assume that agents of the same type share the same utility function (which may not be true in general) shown in Table~\ref{tab:util}.

\end{ex}

\paragraph{Strategy.} A (mixed) strategy is a mapping from the agent's private signal to a distribution on $\{\bA, \bR\}$. For a set $S$, let $\Delta(S)$ be the set of all possible distributions on $S$. Formally, an agent $\sag$'s strategy $\stg_{\sag}: \Sigset \to \Delta(\{\bA, \bR\})$. A strategy can be represented as a vector $\stg  = (\beta_l, \beta_h)$, where $\beta_{\sig}$ is the probability that the agent votes for $\bA$ when receiving signal $\sig$. A strategy profile is the vector of strategies of all agents. $\stgp = (\stg_1, \stg_2, \ldots, \stg_{\ag})$. We call a strategy profile $\stgp$ {\em a symmetric strategy profile induced by strategy $\stg$} if all agents play the same strategy $\stg$ in $\stgp$. 

\begin{dfn}
    \label{dfn: informative}
    An informative strategy is $\stg = (0,1)$, i.e. voting for $\bA$ when receiving $h$ and voting for $\bR$ when receiving $l$. A strategy profile is informative when every agent votes informatively. 
\end{dfn}

In this paper, we focus on {\em regular} strategy profiles. 
\begin{dfn}
\label{dfn:reg}
A strategy profile $\stgp$ is {\em regular} if all friendly agents always vote for $\bA$, and all unfriendly agents always vote for $\bR$ in $\stgp$. 
\end{dfn}
We believe this restriction is mild and natural since ``always vote for $\bA$'' is the dominant strategy for a friendly agent, and ``always vote for $\bR$'' is the dominant strategy for an unfriendly agent in the majority vote.

\paragraph{Fidelity and Expected Utility.} 
Given a strategy profile $\stgp$, let $\lp_{\wos}^{\bA}(\stgp)$ ($\lp_{\wos}^{\bR}(\stgp)$, respectively) be the (ex-ante, before agents receiving their signals) probability that $\bA$ ($\bR$, respectively) becomes the winner when the world state is $\wos$. 

\begin{dfn}[{\bf Fidelity}{}] 
Fidelity is the likelihood that the informed majority decision is reached. In our setting, the fidelity when agents play strategy profile $\stgp$ is 
\begin{align*}
    \acc(\stgp) = & P_{L}\cdot \lp_{L}^{\bR}(\stgp) + P_{H} \cdot \lp_{H}^{\bA}(\stgp). 
\end{align*}
\end{dfn}
We use the word {\em fidelity} to distinguish the notion from {\em accuracy}, which usually denotes the likelihood that the correct world state is revealed. 

The (ex-ante) expected utility of an agent $\sag$ exclusively depends on $\lp_{\wos}^{\bA}(\stgp)$ and $\lp_{\wos}^{\bR}(\stgp)$:
\begin{equation*}
    \ut_{\sag}(\stgp) =P_{L}(\lp_{L}^{\bA}(\stgp)\cdot\vt_{\sag}(L, \bA) + \lp_{L}^{\bR}(\stgp)\cdot \vt_{\sag}(L, \bR)) +  P_{H}(\lp_{H}^{\bA}(\stgp)\cdot\vt_{\sag}(H, \bA) + \lp_{H}^{\bR}(\stgp)\cdot \vt_{\sag}(H, \bR)). 
\end{equation*}

\paragraph{Instance and Sequence of Strategy Profiles.} We define an instance $\inst$ of a voting game on the agent number $\ag$, the majority vote threshold $\Thd$, the world state prior distribution $(P_L, P_H)$, the signal distributions $((P_{hH}, P_{lH}), (P_{hL}, P_{lL}))$, the utility functions of all the agents $\{\vt_\sag\}_{\sag=1}^\ag$, and the approximated fraction of each type $(\agf, \agu, \agc)$. Let $\{\inst_\ag\}_{\ag=1}^{\infty}$ (or $\{\inst_\ag\}$ for short) be a sequence of instances, where each $\inst_\ag$ is an instance of $\ag$ agents.
The instances in a sequence share the same parameters $\{\mu, (P_L, P_H), ((P_{hH}, P_{lH}), (P_{hL}, P_{lL})), (\agf, \agu, \agc)\}$. We do not regard agents in different instances as related and have no additional assumption on the utility functions of agents.

We define a sequence of strategy profiles $\{\stgp_\ag\}_{\ag=1}^{\infty}$ on an instance sequence $\{\inst_\ag\}$. Similarly, we do not have additional assumptions about the agents. Therefore, for different instances in the sequence, the strategies and utility functions of agents can be drastically different. A strategy profile sequence $\{\stgp_{\ag}\}$ is symmetric and induced by strategy $\stg$ if every strategy profile $\stgp_\ag$ in the sequence is a symmetric strategy profile induced by $\stg$. A sequence of strategy profiles is regular if every strategy profile in the sequence is a regular profile. 

\paragraph{$\varepsilon$-strong Bayes Nash Equilibrium} In this paper, we use the solution concept of $\varepsilon$-strong Bayes Nash Equilibrium, an approximation of strong Bayes Nash Equilibrium where no group of agents can increase their utilities by more than $\varepsilon$ through deviation.  A strategy profile $\stgp = (\stg_1, \stg_2,\cdots, \stg_\ag)$ is an {\em $\varepsilon$-strong Bayes Nash Equilibrium} ($\varepsilon$-strong BNE) if there does not exist a subset of agents $D$ and a strategy profile $\stgp' = (\stg_1', \stg_2',\cdots, \stg_\ag')$ such that
\begin{enumerate}
    \item $\stg_\sag = \stg_\sag'$ for all $\sag\not\in D$; 
    \item $\ut_{\sag}(\stgp') \ge \ut_{\sag}(\stgp)$ for all $\sag\in D$; and 
    \item there exists $\sag\in D$ such that $\ut_{\sag}(\stgp') > \ut_{\sag}(\stgp) +\varepsilon$. 
\end{enumerate}

By definition, when $\varepsilon=0$, the equilibrium is a strong Bayes Nash Equilibrium where no group of agents can strictly increase their utilities through deviation. Unfortunately, a strong BNE does not always exist, as shown in the following theorem. Therefore, we seek $\varepsilon$-strong BNE as an approximation. 

\begin{thm}
\label{prop:no_BNE}
For any $N_0 \in \mathbb{N}$, there exists an instance of $N > N_0$ agents, in which a strong Bayes Nash Equilibrium does not exist. 
\end{thm}

\begin{proof}[Proof Sketch]
For any $N_0\in \mathbb{N}$, we construct an instance of $N = 2N_0+3$ agents. The agents consist of three parts: $F$ is a set of $N_0+1$ friendly agents. $C$ is a set of two contingent agents. And $U$ is a set of $N_0$ unfriendly agents. Agents in the same set share the same utility, which is shown in Table~\ref{tbl:noBNE_vt}. The threshold is $\Thd =0.5$. The prior distribution is $P_L = P_H = 0.5$. The signal distribution is $P_{hH} = P_{lL}=0.8$ and $P_{lH} = P_{hL} = 0.2$.

\vspace{0.3cm}
\begin{minipage}{0.5\linewidth}
\centering
\begin{tabular}{@{}ccccc@{}}
\toprule
Agents & $\vt(H, \bA)$ & $\vt(L, \bA)$ & $\vt(L, \bR)$ & $\vt(H, \bR)$ \\ \midrule
$F$    & 100           & 99            & 1             & 0            \\
$C$ & 90            & 0             & 100           & 0            \\ 
$U$ & 1 & 0 & 100 & 99\\
\bottomrule
\end{tabular}
\captionof{table}{Utility of three groups\label{tbl:noBNE_vt}}
\end{minipage}
\begin{minipage}{0.5\linewidth}
\centering
    \begin{tabular}{@{}cccc@{}}
\toprule
Agents & $\stgp_1$ & $\stgp_2$ & $\stgp_3$ \\ \midrule
$F$    & 50.396    & 66.14     & 50.3      \\
$C$ & 85.12     & 75.2      & 76        \\ 
$U$ & 50.396 & 34.46 & 50.3\\
\bottomrule
\end{tabular}
\captionof{table}{Expected utility under three profiles\label{tbl:noBNE_ut}}
\end{minipage}

Consider the following three strategy profiles, under which the expected utility of each group is shown in Table~\ref{tbl:noBNE_ut}. 
\begin{itemize}
    \item $\stgp_1$: $N_0$ agents in $F$ always vote for $\bA$, and one agent votes informatively. $C$ vote informatively. $U$ always vote for $\bR$. 
    \item $\stgp_2$: $F$ always vote for $\bA$.  $C$ vote informatively. $U$ always vote for $\bR$. 
    \item $\stgp_3$: $F$ always vote for $\bA$.  One $C$ agent votes informatively, and the other always votes for $\bR$. $U$ always vote for $\bR$. 
\end{itemize}
These three strategy profiles form a cycle $\stgp_1\to \stgp_2\to \stgp_3\to \stgp_1$ of deviation, where a group of agents has incentives to deviate to the next profile. 

For any other strategy profile $\stgp$, there exists a group of agents with incentives to deviate to one of the three profiles. Firstly, $F$ agents and $U$ agents would like to deviate from their dominant strategy of always voting for $\bA$ ($\bR$, respectively) whenever it can increase the probability that their preferred candidate wins. Given $F$ and $U$ agents play dominant strategies, the best strategy for two $C$ agents is to play the strategy in $\stgp_3$ (one agent votes informatively, the other votes for $\bR$). Then we know that an $F$ agent and two $C$ agents have incentives to deviate from $\stgp_3$ to $\stgp_1$. Therefore, there does not exist a strong Bayes Nash in this instance. The full proof is in Appendix~\ref{apx:no_BNE}.
\end{proof}

%% file: Sections/3_result.tex
In this section, we show that strategic behaviors indeed have a positive impact on leading to the informed majority decision. 
Theorem~\ref{thm:acc} states that if the fidelity of a regular (Definition~\ref{dfn:reg}) strategy profile sequence $\{\stgp_\ag\}_{\ag=1}^{\infty}$, i.e.,   $\acc(\stgp_{\ag})$, converges to 1 as $\ag$ goes to infinity, every $\stgp_\ag$ in the sequence will be an $\varepsilon$-strong Bayes Nash Equilibrium where $\varepsilon$ converges to 0. On the other hand, if $\acc(\stgp_{\ag})$ does not converge to 1, then we can find infinitely many $\stgp_{\ag}$ that are not $\varepsilon$-strong BNE with a constant $\varepsilon$. Moreover, Theorem~\ref{thm:deviate} guarantees that there always exists a regular strategy profile whose fidelity converges to 1, which leads to an $\varepsilon$-strong BNE with $\varepsilon = o(1)$. The two theorems together indicate that strategic voting leads to the informed majority decision in any sequence of instances. 

\begin{thm}
\label{thm:acc}  Given an arbitrary sequence of instances and an arbitrary regular strategy profile sequence $\{\stgp_{\ag}\}_{\ag=1}^{\infty}$, let $\{\acc(\stgp_{\ag})\}_{\ag=1}^{\infty}$ be the sequence of the fidelities of $\stgp_\ag$. 
\begin{itemize}
    \item If $\lim_{\ag\to \infty} \acc(\stgp_{\ag}) = 1$, then for every $\ag$, $\stgp_{\ag}$ is an $\varepsilon$-strong BNE with $\varepsilon = o(1)$. 
    \item If $\lim_{\ag\to \infty} \acc(\stgp_{\ag}) = 1$ does not hold, then there exist infinitely many $\ag$ such that $\stgp_{\ag}$ is NOT an $\varepsilon$-strong BNE for some constant $\varepsilon$. 
    
\end{itemize}
\end{thm}

\begin{thm}
    \label{thm:deviate} Given any arbitrary sequence of instances, there always exists a sequence of regular strategy profiles $\{\stgp_{\ag}'\}_{\ag=1}^{\infty}$ such that $\acc(\stgp_{\ag}')$ converges to 1. 
\end{thm}
We first give a concrete example to illustrate Theorem~\ref{thm:acc}, in which we show an instance for each case in the theorem. 

\begin{ex}
\label{ex:acc}
 We follow the setting of Example~\ref{ex:setting} except for two differences.  First, there is a series of $\ag = 20, 30, \ldots, 500$. For each $N$, the ratio of friendly, unfriendly, and contingent agents is fixed at $2:3:5$. Second, we consider two different cases of signal distributions that fall into different cases of Theorem~\ref{thm:acc}. They share the same signal distribution in world state $L$: $P_{lL} = 0.8, P_{hL} = 0.2$, but the signal distribution for $H$ is different. In case (1), $P_{lH} = 0.1, P_{hH} = 0.9$; and in case (2), $P_{lH} = 0.25, P_{hH} = 0.75$.

 We focus on regular strategy profiles where all contingent agents vote informatively (Definition~\ref{dfn: informative}). 
  For case (2), we also consider another series of regular strategy profiles $\stgp_{\ag}'$ where contingent agents play $\stg' = (0.48, 0.96)$. In Example~\ref{ex:f} and Theorem~\ref{thm:arbitrary} later, we verify that the fidelity of the regular informative voting converges to 1 in case (1) but does not converge to 1 in case (2). On the other hand, the fidelity of the deviating strategy profile $\stgp_\ag'$ in case (2) converges to 1. Figure~\ref{fig:acc1} illustrates these trends of fidelity. 

\begin{figure}[htbp]
\centering
\subfigure[Fidelity]{
\label{fig:acc1}
\begin{minipage}[t]{0.48\linewidth}
\centering
\includegraphics[width=0.99\linewidth]{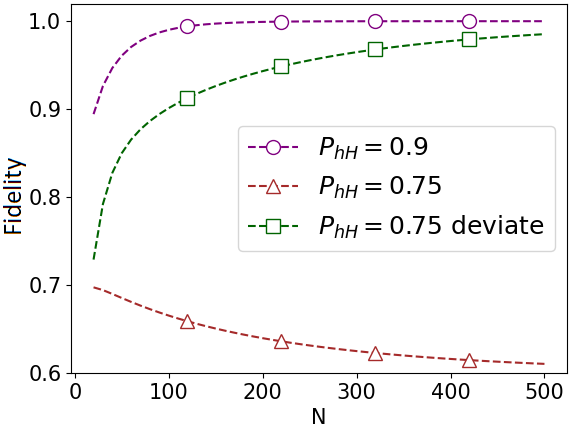}
\end{minipage}%
}%
\subfigure[Expected utility for contingent agents]{
\label{fig:acc2}
\begin{minipage}[t]{0.48\linewidth}
\centering
\includegraphics[width=\linewidth]{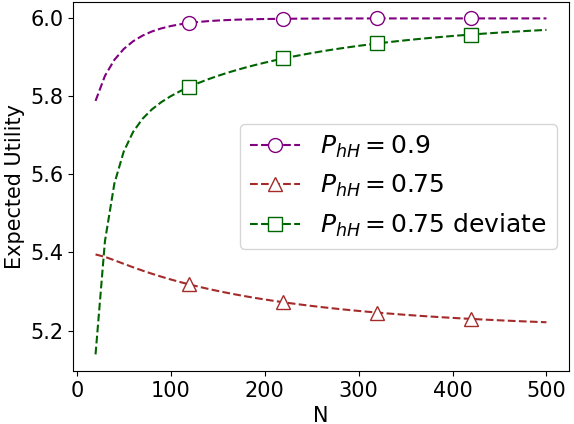}
\end{minipage}%
}%
\caption{Fidelity and expected utilities of informative voting. \label{fig:acc}}
\centering
\end{figure}

The expected utilities of contingent agents in different cases and strategies are shown in Figure~\ref{fig:acc2}. Note the maximum expected utility that a contingent agent can get is $0.4\times 3 + 0.6\times 8 = 6.0$. In accordance with the fidelity, the expected utility of $\stgp_{\ag}$ in case (1) converges to the maximum. In case (2), on the other hand,  $\stgp_{\ag}$ is dominated by $\stgp_{\ag}'$ by a utility gain of at least 0.4. Therefore, the group of contingent agents has no incentive to deviate in case (1) but has an incentive to deviate to $\stgp_{\ag}'$ in case (2). 
\end{ex}

\subsection{ Proof Sketch of Theorem~\ref{thm:acc}}

To show the relationship between $\acc(\stgp_{\ag})$ and $\varepsilon$, we have the following lemma.
\begin{lem}
\label{lem:nash}
For every $\ag$, a regular strategy profile $\stgp_N$ is an $\varepsilon$-strong BNE with $\varepsilon = 2B(B+1)(1 - \acc(\stgp_N))$, where $B$ is the upper bound of utility function $\vt_{\sag}$.
\end{lem}
Lemma~\ref{lem:nash} is an extension of Theorem~3.3 in \citet{schoenebeck21wisdom}. It shows that every $\stgp_{\ag}$ is an $\varepsilon$-strong BNE with $\varepsilon$ proportional to $1 -\acc(\stgp_\ag)$. To prove Lemma~\ref{lem:nash}, we show that, for any other strategy profile $\stgp_\ag'$, a group of agents with incentives to deviate does not exist.

There are two cases of $\stgp_\ag'$. In the first case, the fidelity $\stgp_\ag'$ is bounded by the fidelity of $\stgp_\ag$. More precisely, $(1 - \acc(\stgp_\ag')) < (B+1)\cdot (1 -\acc(\stgp_\ag))$. Then two profiles do not make a big difference, and no agent can gain more than $\varepsilon = 2B(B+1)(1 - \acc(\stgp_N))$ after deviation (Claim 1). 

In the second case, the fidelity $\stgp_\ag'$ is unbounded and much worse than the fidelity of $\stgp_\ag$. Then any contingent agent has no incentives to deviate, as their expected utilities are most positively correlated with fidelity (Claim 2). Next, we show that a deviating group cannot contain both friendly and unfriendly agents, because if the expected utility of one side increases by more than $\varepsilon$, the other side's will decrease (Claim 3). Therefore, a deviating group contains either only friendly agents or only unfriendly agents. Finally, we show that in neither case can the deviation succeed, because pre-determined agents have already played their dominant strategies in a regular profile (Claim 4). 

Now we are ready to propose the proof for Theorem~\ref{thm:acc}. We will actually use Lemma~\ref{lem:nash} and Theorem~\ref{thm:deviate} to prove Theorem~\ref{thm:acc}. We will discuss the two cases separately. 

When $\lim_{\ag\to \infty} \acc(\stgp_{\ag}) = 1$, we apply Lemma~\ref{lem:nash} to each $\stgp_{\ag}$, and get that every $\stgp_\ag$ is an  $\varepsilon$-strong BNE where $\varepsilon = 2 B(B+1)\cdot (1-\acc(\stgp_{\ag}))$. Then $\varepsilon$ will converge to 0 as $\ag\to\infty$. 

When $\lim_{\ag\to \infty} \acc(\stgp_{\ag}) = 1$ does not hold, there are infinitely many $\ag$ with $\stgp_\ag$ being of low fidelity. By Theorem~\ref{thm:deviate} there exists a regular strategy profile sequence $\{\stgp_{\ag}'\}$ with fidelity converging to 1. Because of the difference in fidelity, there are infinitely many $\ag$ such that $\stgp_\ag \neq \stgp_\ag'$. Then we show that, for all sufficiently large $\ag$ where $\stgp_\ag$ is of low fidelity, if all contingent agents turn to play $\stgp_\ag'$ from $\stgp_\ag$, every contingent agent will gain at least a constant amount of extra utility. Therefore, for infinitely many $\ag$, $\stgp_{\ag}$  is NOT an $\varepsilon$-strong BNE for some constant $\varepsilon$. 
The full proof of Theorem~\ref{thm:acc} is in Appendix~\ref{apx:nb_acc}, and the full proof of Lemma~\ref{lem:nash} is in Appendix~\ref{apx:nb_nash}.

\subsection{Proof Sketch of Theorem~\ref{thm:deviate}}



In the proof of Theorem~\ref{thm:deviate}, we construct a strategy $\stg'$ and show that the regular strategy profile sequence $\{\stgp'_N\}$ where all contingent agents play $\stg'$ has fidelity that converges to 1.
It suffices to construct $\stg'$ such that
\begin{enumerate}
    \item if $H$ is the actual world state, the expected fraction of the voters voting for $\bA$ is more than $\mu$ by a constant;
    \item if $L$ is the actual world state, the expected fraction of the voters voting for $\bA$ is less than $\mu$ by a constant.
\end{enumerate}
If this is true, $\acc(\stgp_\ag')$ converges to $1$ due to the Hoeffding Inequality.
It remains to construct $\stg'$ such that (1) and (2) hold.

We first construct $\stg_\mu'$ such that the expected fraction of the voters voting for $\bA$ is exactly $\mu$, where in $\stg_\mu'$ the contingent voter votes for $\bA$ with a probability that is independent to the signal she receives.
This can be done by setting $\stg_\mu'=(\beta^*,\beta^*)$ where $\beta^*$ satisfies $\alpha_F+\alpha_C\cdot\beta^\ast=\mu$ (notice that, given the fraction $\alpha_F$ of the friendly voters who always vote for $\bA$, the fraction $\alpha_U$ of the unfriendly voters who never vote for $\bA$, and the fraction $\alpha_C$ of the contingent voters who vote for $\bA$ with probability $\beta^\ast$, the expected fraction of votes for $\bA$ is $\alpha_F+\alpha_C\cdot\beta^\ast$).

Next, we will adjust $\stg_\mu'$ to $\stg'=(\beta_l,\beta_h)$ that satisfies (1) and (2).
Naturally, we would like to increase the probability for voting $\bA$ if an $h$ signal is received, and we would like to decrease this probability if $l$ is received.
That is, we have $\beta_l=\beta^\ast-\delta_l$ and $\beta_h=\beta^\ast+\delta_h$ for some $\delta_l,\delta_h>0$, and we need to show the existences of $\delta_l$ and $\delta_h$ that make (1) and (2) hold.

When $H$ is the actual world, comparing with $\stg_\mu'$, the probability that each contingent agent votes for $\bA$ is increased by $P_{hH}\cdot\delta_h-P_{lH}\cdot\delta_l$ in $\stg'$.
Thus, the total expected fraction of votes for $\bA$ is increased by
$$\alpha_C\cdot\left(P_{hH}\cdot\delta_h-P_{lH}\cdot\delta_l\right).$$
Similarly, when $L$ is the actual world, similar calculations reveal that the total expected fraction of votes for $\bA$ is increased by
$$\alpha_C\cdot\left(P_{hL}\cdot\delta_h-P_{lL}\cdot\delta_l\right).$$
Since the expected fraction of votes for $\bA$ is exactly $\mu$ for $\stg_\mu'$, we need to choose $\delta_h$ and $\delta_l$ such that
$$\left\{\begin{array}{l}
    \alpha_C\cdot\left(P_{hH}\cdot\delta_h-P_{lH}\cdot\delta_l\right)>0\\
    \alpha_C\cdot\left(P_{hL}\cdot\delta_h-P_{lL}\cdot\delta_l\right)<0
\end{array}\right..$$
This can always be done due to the positive correlation $P_{hH}>P_{hL}$ and $P_{lH}<P_{lL}$.
In particular, if we set $\delta_h=\delta_l\cdot\frac{P_{lH}}{P_{hH}}$, the first inequality would become equality, while the second inequality holds due to the positive correlation.
By slightly increasing $\delta_h$, we can make both inequalities hold.
During these adjustments, we just need to make sure the two constants $\delta_h$ and $\delta_l$ are small enough such that $\beta_h$ and $\beta_l$ are valid probabilities.

\begin{ex}
    \label{ex:thm3}
    In this example, we follow the setting of case (2) in Example~\ref{ex:acc} to illustrate the construction of the strategy $\stg'$. Recall that $\agf = 0.2, \agu = 0.3$, and $\agc = 0.5$. The signal distribution $P_{hH} = 0.75$ and $P_{lL} = 0.8$. The threshold $\Thd = 0.6$.
    
    In the first step, let $\stg_{\mu}' = (0.8, 0.8)$. We could verify that $\agf + \agc\cdot 0.8 = 0.2+ 0.5\times 0.8 = 0.6 = \Thd$.
    
    In the second step, let $\delta_l = 0.3$. Then $\delta_h = \delta_l\cdot \frac{P_{lH}}{P_{hH}}=0.1$. Then $\stg' = (0.5, 0.9)$. Then we have 
    \begin{align*}
        &P_{hH}\cdot \delta_h - P_{lH}\cdot \delta_l = 0.75\times 0.1 -0.25\times 0.3 = 0.\\
        &P_{hL}\cdot \delta_h - P_{lL}\cdot \delta_l = 0.2\times 0.1 - 0.8\times 0.3 =-0.22 < 0. 
    \end{align*}
    
    Finally, we increase $\delta_h$ by 0.06. Then $\delta_l = 0.3$, $\delta_h = 0.16$, and $\stg' = (0.5, 0.96)$. We have 
    \begin{align*}
        &P_{hH}\cdot \delta_h - P_{lH}\cdot \delta_l = 0.75\times 0.16 -0.25\times 0.3 = 0.05 > 0.\\
        &P_{hL}\cdot \delta_h - P_{lL}\cdot \delta_l = 0.2\times 0.16 - 0.8\times 0.3 =-0.208 < 0. 
    \end{align*}
    Therefore, $\stg' = (0.5, 0.96)$ satisfies the condition. 
\end{ex}

%% file: Sections/4_result2.tex
In this section, we analyze the condition that a strategy profile is of high fidelity and apply the analysis to the most common forms of non-strategic voting: informative voting and sincere voting. We show that neither informative nor sincere voting can lead to the informed majority decision in every instance, and we characterize the conditions where they lead to the informed majority decision. Our results give merit to strategic voting. 


In order to characterize the fidelity, we introduce the notion of the {\em \exshare}. 
Given a world state $\wos$, the \exshare{} is the expected vote share the informed majority decision alternative attracts under state $\wos$ minus the threshold of the alternative. 
\begin{dfn}[\bf{Excess expected vote share}]
   Given an instance of $\ag$ agents, and a strategy profile $\stgp$, let random variable $\xrv_{\sag}^{\ag}$ be "agent $\sag$ votes for $\bA$":  $\xrv_{\sag}^{\ag} = 1$ if agent $\sag$ votes for $\bA$, and $\xrv_{\sag}^{\ag} = 0$ if $\sag$ votes for $\bR$. Then the \exshare{} is defined as follows: 
\begin{align}
    \hthd^{\ag}_{H} =& \frac{1}{\ag}\sum_{\sag=1}^{\ag}E[\xrv_{\sag}^{\ag}\mid H] - \Thd.\label{eq:hthdh}
\end{align}
\begin{align}
    \hthd^{\ag}_{L} =&  \frac{1}{\ag}\sum_{\sag=1}^{\ag}E[1-\xrv_{\sag}^{\ag}\mid L] - (1-\Thd)\label{eq:hthdl}. 
\end{align}
Specifically, $\hthd_H^{\ag}$ is the \exshare\ of $\bA$ condition on world state $H$, and $\hthd_L^{\ag}$ is the \exshare\ of $\bR$ condition on world state $L$. For technical convenience, we define $\hthd^{\ag} = \min(\hthd^{\ag}_H, \hthd^{\ag}_L).$ 
\end{dfn}

Our next result shows that we can judge whether the fidelity of a strategy profile sequence converges to $1$ with the tendency of its \exshare\ (or more precisely, the lower limit of $\sqrt{\ag}\cdot \hthd^\ag$). If $\sqrt{\ag}\cdot \hthd^\ag$ has a lower limit of $+\infty$, Then the fidelity of the profiles in the sequence converges to 1. Otherwise, the fidelity is likely not to converge to 1. 

\begin{thm}
\label{thm:arbitrary}
Given an arbitrary sequence of instances and arbitrary sequence of strategy profiles $\{\stgp_{\ag}\}_{\ag=1}^\infty$, let $\hthd^\ag$ be the \exshare{} for each $\stgp_{\ag}$.
\begin{itemize}
     \item If $\liminf_{\ag\to\infty} \sqrt{\ag}\cdot \hthd^{\ag} = +\infty$, the fidelity of $\stgp_{\ag}$ converges to 1 , i.e., $\lim_{\ag\to\infty} \acc(\stgp_{\ag}) = 1$. 
     \item If $\liminf_{\ag\to\infty} \sqrt{\ag}\cdot \hthd^{\ag} < 0$ (including $-\infty$), $\acc(\stgp_{\ag})$ does NOT converge to 1. 
     \item If $\liminf_{\ag\to\infty} \sqrt{\ag}\cdot \hthd^{\ag} \ge 0$ (not including $+\infty$), and the variance of $\sum_{\sag=1}^{\ag} \xrv_{\sag}^{\ag}$ is at least proportional to $\ag$,  $\acc(\stgp_{\ag})$ does NOT converge to 1. 
\end{itemize}
\end{thm}

\begin{remark}
Although Theorem~\ref{thm:arbitrary} does not cover the case when  $\liminf_{\ag\to\infty} \sqrt{\ag}\cdot \hthd^{\ag} \ge 0$ and the variance of $\sum_{\sag=1}^{\ag} \xrv_{\sag}^{\ag}$ is not large enough, we argue that this case is very special and rare. In this case, $(\inf \hthd^{\ag})$  converges to $0$ at the rate of  $O(\frac{1}{\sqrt{\ag}})$, which means the expected vote share of an alternative is almost equal to the threshold. Moreover, the strategies of the agents have low randomness in total. Therefore, we believe that Theorem~\ref{thm:arbitrary} covers the most interesting cases of a sequence of strategy profiles. 
\end{remark}

\begin{proof}[Proof Sketch]
Recall that $\acc(\stgp) = P_{L}\cdot \lp_{L}^{\bR}(\stgp) + P_{H} \cdot \lp_{H}^{\bA}(\stgp).$ Note that $\lp_{H}^{\bA}(\stgp)$ is the probability that the total vote share on $\bA$ exceeds the threshold $\Thd$ when the world state is $H$. Therefore, we can write $\lp_{H}^{\bA}(\stgp)$ using the following formula. $\lp_{L}^{\bR}(\stgp)$ can be written using a similar formula. 
\begin{align*}
    \lp_{H}^{\bA}(\stgp_{\ag}) =&\ \Pr\left[\sum_{\sag=1}^{\ag}\xrv_{\sag}^{\ag} \ge \Thd\cdot \ag \mid H\right]
    =\Pr\left[\sum_{\sag=1}^{\ag}\xrv_{\sag}^{\ag} - \sum_{\sag=1}^{\ag}E\left[\xrv_{\sag}^{\ag}\mid H\right] \ge -\hthd_H^{\ag}\cdot \ag \mid H\right]
\end{align*}
For the first and the second case, we apply the Hoeffding Inequality. For the first case, we show that both $\lp_{H}^{\bA}(\stgp)$ and $\lp_{L}^{\bR}(\stgp)$ are lower bounded by a function of $\ag$ that converges to 1. For the second case, we show that either $\lp_{H}^{\bA}(\stgp)$ and $\lp_{L}^{\bR}(\stgp)$ is upper bounded by a constant smaller than 1.

For the third case, we apply the Berry-Esseen Theorem~\citep{Berry41Accuracy, Esseen42Liapunoff}, which bounds the difference between the distribution of the sum of independent random variables and the normal distribution. Therefore, for some constant $\delta$ and infinitely many $\ag$, $\lp_{H}^{\bA}(\stgp_{\ag})$ (or $\lp_{L}^{\bR}(\stgp_{\ag})$) will not deviate from $1 - \Phi(\delta)$ too much and is bounded away from 1 by a constant. $\Phi$ is the CDF of the standard normal distribution. The requirements for the variance in the third case are from the Berry-Esseen Theorem.

The full proof of Theorem~\ref{thm:arbitrary} is in Appendix~\ref{apx:nb_arbitrary}, and the detailed definition of the Berry-Esseen Theorem is in Appendix~\ref{apx:berry}.
\end{proof}

Theorem~\ref{thm:arbitrary} provides a criterion for judging whether a strategy profile sequence is of high fidelity. If we apply Theorem~\ref{thm:acc} to each case of Theorem~\ref{thm:arbitrary}, we directly get a criterion for judging whether a regular strategy profile sequence is an $\varepsilon$-strong equilibrium.

\begin{coro}
\label{coro:arbeq}
Given an arbitrary sequence of instances and an arbitrary regular sequence of strategy profiles $\{\stgp_{\ag}\}_{\ag=1}^\infty$, let $\hthd^\ag$ be defined for each $\stgp_{\ag}$.
\begin{itemize}
     \item If $\liminf_{\ag\to\infty} \sqrt{\ag}\cdot \hthd^{\ag} = +\infty$,  then for every $\ag$, $\stgp_{\ag}$ is an $\varepsilon$-strong BNE with $\varepsilon = o(1)$.
     \item If $\liminf_{\ag\to\infty} \sqrt{\ag}\cdot \hthd^{\ag} < 0$ (including $-\infty$), there are infinitely many $\ag$ such that $\stgp_{\ag}$ is NOT an $\varepsilon$-strong BNE with constant $\varepsilon$.
     \item If $\liminf_{\ag\to\infty} \sqrt{\ag}\cdot \hthd^{\ag} \ge 0$ (not including $+\infty$), and the variance of $\sum_{\sag=1}^{\ag} \xrv_{\sag}^{\ag}$ is at least proportional to $\ag$, there are infinitely many $\ag$ such that $\stgp_{\ag}$ is NOT an $\varepsilon$-strong BNE with constant $\varepsilon$.
\end{itemize}
\end{coro}

\begin{ex}
\label{ex:f}
In this example, we use Theorem~\ref{thm:arbitrary} to bound the fidelity of different cases in Example~\ref{ex:acc}. Note that in each $\ag$ the ratio of friendly, unfriendly, and contingent agents is fixed to be $2:3:5$, and agents of the same type play the same strategy for different $\ag$. Therefore, the \exshare{} of profile $\stgp_\ag$ and $\stgp_{\ag}'$ is independent of $\ag$. 
\paragraph{Case 1: $P_{lH} = 0.1, P_{hH} = 0.9$.} In both world states, we have $E[\xrv_{\sag}^{\ag}\mid H] = 1$ for friendly agents and $E[\xrv_{\sag}^{\ag}\mid H] = 0$ for unfriendly agents. Contingent agents vote for $\bA$ with probability $P_{hH} = 0.9$ in $H$ state and for $\bR$ with probability $P_{lL} = 0.8$ in $L$ state. Therefore, $\hthd^{\ag} > 0$, and $\sqrt{\ag}\cdot \hthd^{\ag}$ goes to $+\infty$. 
\begin{align*}
    \hthd^{\ag}_H = 0.2 + 0.5 \times 0.9 - 0.6 = 0.05 \qquad
    \hthd^{\ag}_L =  0.3 + 0.5 \times 0.8 - 0.4 = 0.3.
\end{align*}

\paragraph{Case 2: $P_{lH} = 0.25, P_{hH} = 0.75$.} For $\stgp_{\ag}$, we have $\hthd^{\ag} < 0$, and $\sqrt{\ag}\cdot \hthd^{\ag}$ goes to $-\infty$. 
\begin{align*}
    \hthd^{\ag}_H = 0.2 + 0.5 \times 0.75 - 0.6 = - 0.025 \qquad
    \hthd^{\ag}_L =  0.3 + 0.5 \times 0.8 - 0.4 = 0.3.
\end{align*}

And for the the deviating strategy profile $\stgp_{\ag}'$, we have $\hthd^{\ag} > 0$,  and $\sqrt{\ag}\cdot \hthd^{\ag}$ goes to $+\infty$.  
\begin{align*}
    \hthd^{\ag}_H =& 0.2 + 0.5 \cdot (0.75\times 0.96 + 0.25\times 0.48) - 0.6 = 0.02 \\
    \hthd^{\ag}_L = & 0.3 + 0.5 \cdot (0.8\times 0.52+0.2\times 0.04) - 0.4 = 0.112.
\end{align*}

In Case 1, the regular strategy profile $\stgp_\ag$ lies in the first case of Theorem~\ref{thm:arbitrary}, has an fidelity converging to 1, and is an $\varepsilon$-strong BNE with $\varepsilon = o(1)$.
In Case 2, $\stgp_\ag$ lies in the second case of Theorem~\ref{thm:arbitrary} and is dominated by the deviating strategy profile $\stgp_\ag'$. This is in accordance with our observation in Figure~\ref{fig:acc} and Example~\ref{ex:acc}.

\end{ex}

Although Theorem~\ref{thm:arbitrary} (and Corollary~\ref{coro:arbeq}) do not cover all the strategy profile sequences, the following result provides a dichotomy for symmetric profile sequences to judge fidelity. Given a symmetric strategy profile $\stgp_{\ag}$ induced by strategy $\stg = (\bp_l, \bp_h)$, we can compute the \exshare{} of $\stgp_{\ag}$. Recall the definition of \exshare{}:

\begin{equation*}
    \hthd^{\ag}_{H} = \frac{1}{\ag}\sum_{\sag=1}^{\ag}E[\xrv_{\sag}^{\ag}\mid H] - \Thd.
\end{equation*}

In $H$ state, an agent with signal $h$ votes for $\bA$ with probability $\bp_h$, and an agent with signal $l$ votes for $\bA$ with probability $\bp_l$. Therefore, $E[\xrv_{\sag}^{\ag}\mid H] = P_{hH}\cdot \bp_h + P_{lH}\cdot \bp_l.$ Then we have 
\begin{align*}
    \hthd^{\ag}_{H} 
    = &\ \frac{1}{\ag}\sum_{\sag=1}^{\ag} (P_{hH}\cdot \bp_h + P_{lH}\cdot \bp_l) -\Thd
    =  P_{hH}\cdot \bp_h + P_{lH}\cdot \bp_l -\Thd.
\end{align*}

With similar reasoning, we can compute $\hthd^{\ag}_{L}$ and $\hthd^{\ag}$:
\begin{align*}
    \hthd^{\ag}_{L}
    =\ &(P_{hL}\cdot (1-\bp_h) + P_{lL}\cdot (1-\bp_l)) -(1-\Thd), \quad
    \hthd^{\ag} =  \min \left( \hthd^{\ag}_{H}, \hthd^{\ag}_{L}\right).\textbf{}
\end{align*}

An interesting observation for the symmetric strategy profiles is that its \exshare{} is independent of the number of agents $\ag$. This is because when every agent plays the same strategy, the expectation of every $\xrv_{\sag}^{\ag}$ is the same. For simplicity, given a sequence of symmetric strategy profiles $\{\stgp_\ag\}$, we write its \exshare{} as $\hthd_H, \hthd_L,$ and $\hthd$. 

\begin{coro}
    \label{coro:sym}
    For an arbitrary strategy $\stg$ and an arbitrary sequence of instances, let $\{\stgp_{\ag}\}$ be the sequence of symmetric strategy profile $\stgp_{\ag}$ induced by $\stg$, and $\hthd$ be the \exshare{} of $\{\stgp_{\ag}\}$.
\begin{itemize}
    \item If $\hthd > 0$, $\acc(\stgp_{\ag})$ converges to 1. 
    \item If $\hthd \le 0$, $\acc(\stgp_{\ag})$ does not converge to 1. 
\end{itemize}
\end{coro}

\begin{proof}[Proof Sketch]
The proof of Corollary~\ref{coro:sym} works by showing that each case of a symmetric strategy profile falls into some case of Theorem~\ref{thm:arbitrary}. When $\hthd > 0$, $\sqrt{\ag}\cdot \hthd^{\ag}\to+\infty$. When $\hthd< 0$, $\sqrt{\ag}\cdot \hthd^{\ag}\to-\infty$.  And when $\hthd = 0$,and $\sqrt{\ag}\cdot \hthd^{\ag} = 0$. The variance requirement is also satisfied. (Otherwise, the strategy $\stg$ must be always voting for the same candidate. This directly implies $\hthd < 0$, which is a contradiction.)  The full proof of the Corollary~\ref{coro:sym} is in Appendix~\ref{apx:nb_sym}
\end{proof}

\subsection{Case Study: Informative Voting and Sincere Voting}

In this section, we study the two most common non-strategic voting schemes -- informative voting and sincere voting under our information structure. We show that both voting schemes lead to the informed majority decision if and only if certain conditions are satisfied. 

In informative voting, all agents play the strategy $\stg = (0, 1)$.  
    When the world state is $H$, an agent receives signal $h$ and votes for $\bA$ with probability $P_{hH}$. Therefore, the \exshare{} in the $H$ state is $\hthd_H = P_{hH} - \Thd$. Similarly, the \exshare{} in the $L$ state is $\hthd_L = P_{lL} - (1 - \Thd) = \Thd - P_{hL}$. Applying Corollary~\ref{coro:sym}, we get the following statement. 

    \begin{coro}
        \label{coro:informative}
        For an arbitrary sequence of instances, let $\{\stgp_{\ag}\}$ be the sequence of informative voting profile. Then, the fidelity $\acc(\stgp_N)$ converges to 1 if and only if $P_{hH} > \Thd > P_{hL}$. 
    \end{coro}
    \begin{figure}[htbp]
        \centering        \includegraphics[width=0.6\linewidth]{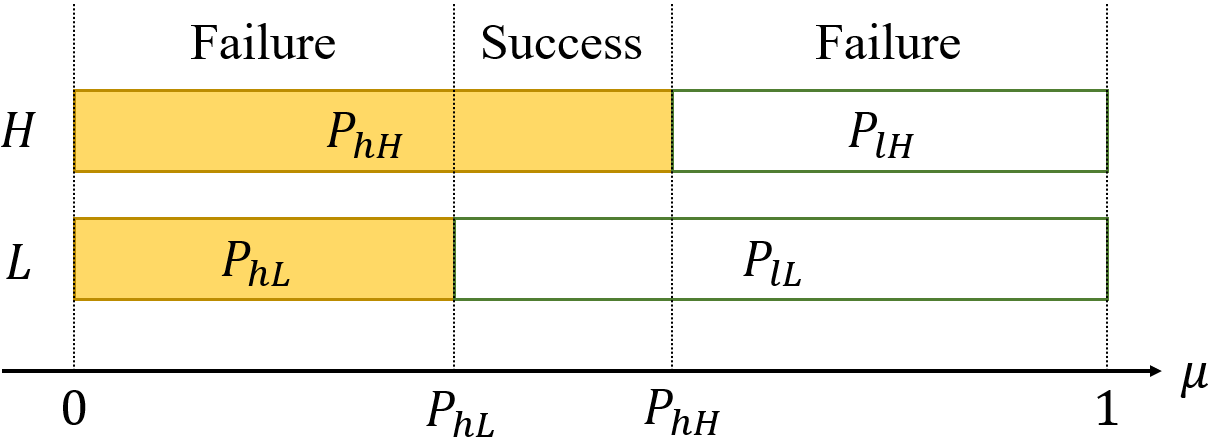}
        \caption{Illustration of Corollary~\ref{coro:informative}.\label{fig:informative}}
        
    \end{figure}

Corollary~\ref{coro:informative} forms a comparison with Theorem~\ref{thm:acc}. Strategic behavior always leads to the informed majority decision, while non-strategic informative voting achieves the same only when the majority vote threshold is ``unbiased'' compared with the signal distribution. 

In sincere voting, an agent votes as if she is making the decision individually. A sincere agent chooses the alternative that maximizes the expected utility conditioned on the signal. The expected utility of an agent making an individual decision conditioned on signal $\sig$ is  
\begin{align*}
    \ut_{\sag}(\bA\mid \sig) = & \Pr[\Wosrv = L \mid \sig] \cdot \vt_{\sag}(\bA, L) + \Pr[\Wosrv = H \mid \sig] \cdot \vt_{\sag}(\bA, H). \\
    \ut_{\sag}(\bR\mid \sig) = & \Pr[\Wosrv = L \mid \sig] \cdot \vt_{\sag}(\bR, L) + \Pr[\Wosrv = H \mid \sig] \cdot \vt_{\sag}(\bR, H). 
\end{align*}

\begin{dfn}
    A strategy profile $\stgp$ is \emph{sincere} if for any agent $i$, conditioned that $i$ receives signal $\sig$, $i$ votes for $\bA$ if $\ut_{\sag}(\bA\mid \sig) > \ut_{\sag}(\bR\mid \sig)$ and votes for $\bR$ otherwise. 
\end{dfn}
A sincere strategy profile is not always symmetric, because the sincere behavior of agents not only depends on his/her signal but also on his/her utility. Therefore, sincere agents with different utility functions may play different strategies.
Given the assumption of $P_{hH} > P_{hL}$, we have $\Pr[\Wosrv = L \mid l] > \Pr[\Wosrv = L \mid h]$ and $\Pr[\Wosrv = H \mid l] < \Pr[\Wosrv = H \mid h]$. As a result, $\ut_{\sag}(\bA\mid l) < \ut_{\sag}(\bA\mid h)$ and $\ut_{\sag}(\bR\mid l) > \ut_{\sag}(\bR\mid h)$. Therefore, a sincere voter would play one of the five strategies below based on her utility function $\vt_\sag$.  
\begin{enumerate}
    \item If $\ut_{\sag}(\bA\mid l) < \ut_{\sag}(\bR\mid l)$, and $\ut_{\sag}(\bA\mid h) < \ut_{\sag}(\bR\mid h)$, an agent always votes for $\bR$. 
    \item If $\ut_{\sag}(\bA\mid l) < \ut_{\sag}(\bR\mid l)$, and $\ut_{\sag}(\bA\mid h) = \ut_{\sag}(\bR\mid h)$, an agent votes for $\bR$ under signal $l$ and votes arbitrarily under signal $h$. 
    \item If $\ut_{\sag}(\bA\mid l) < \ut_{\sag}(\bR\mid l)$, and $\ut_{\sag}(\bA\mid h) > \ut_{\sag}(\bR\mid h)$, an agent votes informatively. 
    \item If $\ut_{\sag}(\bA\mid l) = \ut_{\sag}(\bR\mid l)$, and $\ut_{\sag}(\bA\mid h) > \ut_{\sag}(\bR\mid h)$, an agent votes arbitrarily under signal $l$, and vote for $\bA$ under signal $h$. 
    \item If $\ut_{\sag}(\bA\mid l) > \ut_{\sag}(\bR\mid l)$, and $\ut_{\sag}(\bA\mid h) > \ut_{\sag}(\bR\mid h)$, an agent always votes for $\bA$. 
\end{enumerate}

A sincere profile is also a regular profile, as friendly agents always vote for $\bA$, and unfriendly agents always vote for $\bR$ in their individual decisions. Therefore, applying Theorem~\ref{thm:acc}, we have the following statement. 

\begin{coro}
    \label{coro:sincere}
    For an arbitrary sequence of instances, let $\{\stgp_\ag\}$ be the sequence of sincere strategy profiles. Then the fidelity $\acc(\stgp_\ag)$ converges to 1 if and only if $\stgp_\ag$ is an $\varepsilon$-strong Bayes Nash with $\varepsilon = o(1)$. 
\end{coro}

Corollary~\ref{coro:sincere} tells us that sincere voting performs as well as strategic voting if and only if itself is also strategic. The following example illustrates different behaviors of sincere voters by ``manipulating'' their utility functions under the same world state and signal distribution and gives examples where sincere voting succeeds and fails.

\begin{ex}
    \label{ex:sincere} 
    Consider the following scenario. The world state prior $P_{H} = P_{L} = 0.5$. The signal distribution $P_{hH} = P_{lL} = 0.8$, and $P_{lH} = P_{hL} = 0.2$. By the Bayes Theorem, we compute the probability of a world state conditioned on a private signal as follows. 

    \begin{align*}
        \Pr[\Wosrv = L \mid l] = 0.8, &\qquad \Pr[\Wosrv = H \mid l] = 0.2\\
        \Pr[\Wosrv = L \mid h] = 0.2, &\qquad \Pr[\Wosrv = H \mid h] = 0.8.
    \end{align*}

We assume all agents are contingent and share the same utility function, and consider three different cases as shown in Table~\ref{tbl:ex_sincere}. Suppose $\stgp$ is a strategy profile where all agents vote sincerely. 

\begin{minipage}{\linewidth}
\vspace{0.3cm}
\centering
\begin{tabular}{@{}ccccc@{}}
\toprule
(Winner, World State) & $(\bA, H)$         & $(\bA, L)$         & $(\bR, H)$         & $(\bR, L)$         \\ \midrule
Case 1      & 1           & 0           & 0           & 1           \\
Case 2      & 5           & 1           & 0           & 2           \\
Case 3      & 4           & 1           & 0           & 2           \\ \bottomrule
\end{tabular}
\captionof{table}{Utility functions for three cases.\label{tbl:ex_sincere}}
\end{minipage}


In case 1, $\ut_{\sag}(\bA\mid h) = \ut_{\sag}(\bR\mid l) = 0.8$, and $\ut_{\sag}(\bA\mid l) = \ut_{\sag}(\bR\mid h) = 0.2$. Therefore, every sincere voter votes informatively, and $\stgp$ is also informative voting. By corollary~\ref{coro:informative}, $\stgp$ leads to the informed majority decision if and only if the threshold $0.2 < \Thd < 0.8$.  

In case 2,  $\ut_{\sag}(\bA\mid l) = 1.8, \ut_{\sag}(\bR\mid l) = 1.6$, $\ut_{\sag}(\bA\mid h) = 4.2$, and $\ut_{\sag}(\bR\mid h) = 0.4$. Therefore, all sincere agents always vote for $\bA$ even if they are contingent. In this case, $\bA$ is always the winner, and $\acc(\stgp)$ does not converge to 1. 

In case 3, $\ut_{\sag}(\bA\mid l) = \ut_{\sag}(\bR\mid l) = 1.6$, $\ut_{\sag}(\bA\mid h) = 3.4$, and $\ut_{\sag}(\bR\mid h) = 0.4$. In this case, a sincere agent votes $\bA$ under signal $h$, and votes arbitrarily under signal $l$. Then for any $0.2 < \Thd < 1$, $\stgp$ induced by strategy $\stg = (\bp_l, 1)$ leads to the informed majority decision with high probability, where $\bp_l \ge 0$ satisfies $\bp_l > 5\Thd - 4$ and $\bp_l < 1.25\Thd - 0.25$. These conditions guarantee the \exshare{} of $\stgp$ to be strictly positive. 
    
\end{ex}

As shown in Example~\ref{ex:sincere}, sincere agents can have drastically different behaviors in different scenarios. Nevertheless, once we know the strategy of each agent, we can apply Theorem~\ref{thm:arbitrary} to analyze the probability of a sequence of sincere voting leading to the informed majority decision. 

%% file: Sections/5_nonbinary.tex
\label{sec:nb}
In this section, we discuss how we extend our model and results to a setting with non-binary world states and non-binary signals. We follow the setting of~\citet{schoenebeck21wisdom}. The largest difference in the non-binary setting is that the preferences of agents form a spectrum along multiple world states. Different agents have different thresholds of world states in which they switch the preferred alternatives. 

\paragraph{World State.}There are $\Wos$ possible world states $\Wosset = \{1,2,\cdots, \Wos\}$. The higher the world state is, the more $\bA$ is preferred to $\bR$. We use $\wos$ to denote a generic world state. Let $P_{\wos} = \Pr[\Wosrv = \wos]$ be the common prior of the world state. We assume $P_{\wos} > 0$ for every $\wos \in \Wosset$. 

\paragraph{Signal.} Every agent receives a signal from $\Sigset = \{1,2,\cdots, \Sig\}$. We use $\sig$ to denote a generic signal. Signals are i.i.d conditioned on the world state. Let $P_{\sig\wos} = \Pr[\Sigrv_{\sag} = \sig \mid \Wosrv = \wos]$ be the probability that an agent receives signal $\sig$ given world state is $\wos$. The assumption of $P_{hH} >P_{hL}$ in the binary state is extended to the following assumption of stochastic dominance, which requires signals to be positively correlated to world states.

\begin{asm}[Stochastic Dominance]
For any agent $\sag$ and any world states $\wos_1 > \wos_2$,
\begin{equation*}
    \Pr[\Sigrv_{\sag} \ge \sig \mid \Wosrv =\wos_1] > \Pr[\Sigrv_{\sag} \ge \sig \mid \Wosrv =\wos_2].
    \end{equation*}
\end{asm}

\paragraph{Utility} Every agent $\sag$ has a utility function $\vt_{\sag}: \Wosset \times \{\bA, \bR\}\to \{0,1,\cdots, B\}$, where $B$ is the positive integer upper bound. We assume that $\bA$ is more preferable in a higher world state than in a lower state, and $\bR$ is the opposite: for any $\wos_1$ and $\wos_2$ with $\wos_1 > \wos_2$, $\vt_{\sag}(\wos_1, \bA) > \vt_{\sag} (\wos_2, \bA)$, and $\vt_{\sag}(\wos_1, \bR) < \vt_{\sag} (\wos_2, \bR)$. We also assume $\vt_{\sag}(\wos, \bA)\neq \vt_{\sag}(\wos,\bR)$ for any $\wos$ and any $\sag$. 

\paragraph{Fraction of agents} Given a world state $\wos$, let $\aga_\wos$ and  $\agr_\wos$ be the {\em approximated} fraction of agents prefer $\bA$ ($\bR$, respectively) in world state $\wos$. We assume $\aga_\wos$ and  $\agr_\wos$ are independent from $\ag$. Formally, given $\ag$,  let $\agset(\wos,\bA) =\{ \sag\mid \vt_{\sag}(\wos, \bA) > \vt_{\sag}(\wos, \bR)\}$ be the set of agents preferring $\bA$ in $\wos$ ($\agset(\wos,\bR) $ defined similarly). We have $|\agset(\wos,\bR)| = \lfloor \agr_\wos \cdot \ag\rfloor$ and $|\agset(\wos,\bA)| = \ag - |\agset(\wos,\bR)| $. Naturally, $\aga_\wos$ increases, and $\agr_\wos$ decreases as $\wos$ increases. We assume that $\aga_\wos$ and $\agr_\wos$ are common knowledge.

\paragraph{Majority vote and informed majority decision} We study the majority vote with threshold $\Thd$. If at least $\Thd\cdot \ag$ agents vote for $\bA$, $\bA$ is announced to be the winner; otherwise, $\bR$ is announced to be the winner.
The informed majority decision is defined on each world state $\wos$. 
Given a world state $\wos$, if $\aga_\wos > \Thd$, we say $\bA$ is the informed majority decision; otherwise, we say $\bR$ is the informed majority decision. We assume that $\aga_{\wos}\neq \Thd$ for all $\wos\in\Wosset$, and the rounding between $\aga_{\wos}$ and $|\agset(\wos,\bA)|$ ($\agr_{\wos}$ and $|\agset(\wos,\bR)|$, respectively) does not flip the informed majority decision. 

\paragraph{Types of agents}
Let $\Lowset = \{\wos\in\Wosset \mid \aga_\wos < \Thd\}$ and $\Highset = \{\wos\in\Wosset \mid \aga_\wos > \Thd\}$ be the sets of world states where $\bR$ ($\bA$, respectively) is the informed majority decision. We only consider the case where both $\Lowset$ and $\Highset$ are non-empty. (Otherwise, an alternative is unanimously the informed majority decision, and there is no uncertainty.) There is a threshold partitioning $\Wosset$ into two sets.  Let $\Low = \max \{\wos\in \Lowset\}$ to be the largest world where $\bR$ is the informed majority decision, and $\High = \min\{\wos\in \Highset\}$ to be the smallest world where $\bA$ is the informed majority decision. We have $\High = \Low +1$.

Similarly, for an agent $\sag$, let $\Lowset_{\sag} = \{\wos\in\Wosset \mid \vt_{\sag}(\wos, \bR) > \vt_{\sag}(\wos, \bA)\}$ and $\Highset_{\sag} = \{\wos\in\Wosset \mid \vt_{\sag}(\wos, \bR) < \vt_{\sag}(\wos, \bA)\}$. $\Lowset_{\sag}$ ($\Highset_{\sag}$, respectively) is the set of world states where $\bR$ ($\bA$, respectively) is preferred by $\sag$. For a agent $\sag$, let $\Low_{\sag} = \max \{\wos\in \Lowset_{\sag}\}$ be the largest world where $\sag$ prefers $\bR$, and $\High_{\sag} = \min\{\wos\in \Highset_{\sag}\}$ to be the smallest world where $\sag$ prefers $\bA$. Specifically, let $\Low_{\sag} = 0$ if $\Lowset_{\sag} = \emptyset$ and $\High_{\sag} = \Wos+1$ if $\Highset_{\sag} = \emptyset$.  We have $\High_{\sag} = \Low_{\sag} +1$. 

\begin{enumerate}
    \item We say an agent $\sag$ is {\em (candidate) friendly} if $\Lowset\cap \Highset_{\sag} \neq \emptyset$. This says that there exists a world state $\wos$ where $\sag$ prefer $\bA$ while the informed majority decision is $\bR$. 
    \item Similarly, an agent $\sag$ is {\em (candidate) unfriendly} if $\Highset \cap \Lowset_{\sag} \neq \emptyset$. This says that there exists a world state $\wos$ where $\sag$ prefers $\bR$ while the informed majority decision is $\bA$.
    \item Finally, an agent $n$ is {\em contingent} if $\Lowset_{\sag} = \Lowset$. or equivalently, $\Highset_{\sag} = \Highset$. 
\end{enumerate}
Unlike the binary setting, friendly/unfriendly agents in the non-binary setting do not always prefer one alternative. They just have thresholds above or below the majority.

\begin{ex}
\label{ex:nb}
We extend the COVID policy-making scenario in Example~\ref{ex:setting} to the non-binary setting. $N=20$ voters decide whether to accept or reject the more-restrictive policy. The world states $\Wosset = \{1, 2, 3\}$ describe the risk level of the virus, whereas a larger state represents a higher risk. Suppose $P_1 =P_2 = 0.3$, and $P_3 = 0.4$. 

Every voter receives a private signal from $\Sigset = \{1, 2, 3, 4\}$. The larger the signal is, the higher the risk is likely to be. Table~\ref{tbl:nb_ex_sig} is the signal distribution given the world state. 

\begin{minipage}{\linewidth}
\vspace{0.3cm}
\centering
\begin{tabular}{@{}ccccc@{}}
\toprule
World State & Signal 1   &  Signal 2   & Signal 3   &  Signal 4   \\ \midrule
1           & 0.6 & 0.2 & 0.1 & 0.1 \\
2           & 0.4 & 0.2 & 0.2 & 0.2 \\
3           & 0.1 & 0.2 & 0.3 & 0.4 \\ \bottomrule
\end{tabular}
\captionof{table}{Signal distribution.\label{tbl:nb_ex_sig}}
\end{minipage}

The majority threshold is $\Thd = 0.6$. Therefore, $\bA$ is the informed majority decision if and only if at least 12 agents prefer $\bA$ to $\bR$. The voters are categorized into four different groups. Each group has five voters, and voters in the same group share the same utility shown in Table~\ref{tbl:nb_ex_vt}. The larger the group index is, the more voters prefer $\bA$ to $\bR$. 

Table~\ref{tbl:nb_ex_pref} shows the preferences of each group and the informed majority decision under each world state. The preference of each group comes from the comparison of utilities in Table~\ref{tbl:nb_ex_vt}. The informed majority decision is the aggregation of group preferences. Since each group has five voters, and the majority threshold is 12 agents, $\bA$ needs to be preferred by at least three groups to become the informed majority decision. Therefore, $\bA$ is the informed majority decision only in state 3.
By comparing the preferences of each group and the informed majority decision, we know that group 1 voters are unfriendly agents, group 2 voters are contingent agents, and group 3 and 4 voters are friendly agents. 

\begin{minipage}{.4\linewidth}
\vspace{0.25cm}
    \centering
\begin{tabular}{@{}ccccccc@{}}
\toprule
Winner      & \multicolumn{3}{c}{$\bA$} & \multicolumn{3}{c}{$\bR$} \\ \midrule
World State & 1       & 2      & 3      & 1       & 2      & 3      \\ \midrule
Group 1     & 1       & 2      & 3      & 8       & 6      & 4      \\
Group 2     & 2       & 3      & 4      & 6       & 4      & 2      \\
Group 3     & 2       & 5      & 8      & 4       & 3      & 2      \\
Group 4     & 4       & 6      & 9      & 3       & 2      & 1      \\ \bottomrule
\end{tabular}
\captionof{table}{Utility function of each group\label{tbl:nb_ex_vt}}
\end{minipage}
\begin{minipage}{0.6\linewidth}
\centering
\vspace{0.45cm}
\begin{tabular}{@{}cccc@{}}
\toprule
World State & 1     & 2     & 3     \\ \midrule
Group 1     & $\bR$ & $\bR$ & $\bR$ \\
Group 2     & $\bR$ & $\bR$ & $\bA$ \\
Group 3     & $\bR$ & $\bA$ & $\bA$ \\
Group 4     & $\bA$ & $\bA$ & $\bA$ \\
Informed Majority    & $\bR$ & $\bR$ & $\bA$ \\ \bottomrule
\end{tabular}
\captionof{table}{Preference of each group and the majority.\label{tbl:nb_ex_pref}}
\end{minipage}

\end{ex}

\paragraph{Strategy.} In the non-binary setting, a strategy can be represented as a vector $\stg  = (\beta_1, \beta_2, \cdots, \beta_{\Sig})$, where $\beta_{\sig}$ is the probability that the agent votes for $\bA$ when receiving signal $\sig$. 
A strategy profile is the vector of strategies of all agents. $\stgp = (\stg_1, \stg_2, \cdots, \stg_{\ag})$. 

The definition of the regular strategy profile remains the same as in the binary setting: friendly agents always vote for $\bA$, and unfriendly agents always vote for $\bR$. 

\paragraph{Fidelity and Expected Utility} Given a strategy profile $\stgp$, let $\lp_{\wos}^{\bA}(\stgp)$ ($\lp_{\wos}^{\bR}(\stgp)$, respectively) be the probability that $\bA$ ($\bR$, respectively) becomes the winner when world state is $\wos$. We can define the fidelity and the expected utility in the same manner as in the binary setting. 
\begin{align*}
    \acc(\stgp) = & \sum_{\wos\in \Lowset}P_{\wos}\cdot \lp_{\wos}^{\bR}(\stgp) + \sum_{\wos\in \Highset}P_{\wos} \cdot \lp_{\wos}^{\bA}(\stgp).\\
    \ut_{\sag}(\stgp) =& \sum_{\wos = 1}^{\Wos} P_{\wos}(\lp_{\wos}^{\bA}(\stgp)\cdot\vt_{\sag}(\wos, \bA) + \lp_{\wos}^{\bR}(\stgp)\cdot \vt_{\sag}(\wos, \bR)). 
\end{align*}


\paragraph{Excess Expected Vote Share} Similarly, the \exshare{} is the expected vote share that an alternative attracts under state $\wos$ minus the threshold of the alternative. In different instances, the informed majority decision may change in different world states. Therefore, we define the \exshare{} for both $\bA$ and $\bR$ in every world state.
\begin{align*}
    \hthd^\ag_{\wos\bA} =&  \frac{1}{\ag}\sum_{\sag = 1}^{\ag}E[\xrv_{\sag}^{\ag}\mid \wos] - \Thd.\\
    \hthd^\ag_{\wos\bR} =& \frac{1}{\ag}\sum_{\sag = 1}^{\ag}E[1-\xrv_{\sag}^{\ag}\mid \wos] - (1-\Thd). 
\end{align*}

For world states $\wos \in \Highset$ where $\bA$ is the informed majority decision, we care about $\hthd^\ag_{\wos\bA}$; and for states $\wos \in \Lowset$, we care about $\hthd^\ag_{\wos\bR}$. Therefore, we define $\hthd^{\ag}$ to be the smallest \exshare{} among those we care about. 
$$\hthd^{\ag} = \min\left( \min_{\wos\in\Highset} (\hthd^\ag_{\wos\bA}), \min_{\wos\in\Lowset} (\hthd^\ag_{\wos\bR})\right).  $$

For symmetric profile sequences where \exshare{} is independent of $\ag$, we use $\hthd_{\wos\bA}, \hthd_{\wos\bR}$, and $\hthd$ to denote them. 

\paragraph{Instance and Sequence of Strategy Profiles.}  Let $\{\inst_\ag\}_{\ag=1}^{\infty}$ (or $\{\inst_\ag\}$ for short) be a sequence of instances, where each $\inst_\ag$ is an instance of $\ag$ agents. The instances in a sequence share the same majority threshold $\Thd$, world state prior distribution $\{P_{\wos}\}$, signal prior distribution $\{P_{\sig\wos}\}$, and approximated type fractions ($\aga_\wos, \agr_\wos$). Same to the binary setting, we do not regard agents in different instances as related and have no additional assumption on the utility functions of agents.
We define a sequence of strategy profile $\{\stgp_\ag\}_{\ag=1}^{\infty}$ based on the instance sequence, where for each $\ag$, $\stgp_\ag$ is a strategy profile in $\inst_\ag$.


Our positive results on strategic voting can be extended to the non-binary setting. Theorem~\ref{thm:nb_acc} states the equivalence of fidelity converging to 1 and the $\varepsilon$-strong Bayes Nash with $\varepsilon = o(1)$. Theorem~\ref{thm:nb_deviate} guarantees the existence of the regular profile sequence with fidelity converging to 1. 

\begin{thm}
\label{thm:nb_acc}  Given an arbitrary sequence of instance and an arbitrary regular strategy profile sequence $\{\stgp_{\ag}\}_{\ag=1}^{\infty}$, let $\{\acc(\stgp_{\ag})\}_{\ag=1}^{\infty}$ be the sequence of the fidelities of $\stgp_\ag$. 

\begin{itemize}
    \item If $\lim_{\ag\to \infty} \acc(\stgp_{\ag}) = 1$, then for every $\ag$, $\stgp_{\ag}$ is an $\varepsilon$-strong BNE with $\varepsilon = o(1)$. 
    \item If $\lim_{\ag\to \infty} \acc(\stgp_{\ag}) = 1$ does not hold, then there exist infinitely many $\ag$ such that $\stgp_{\ag}$ is NOT an $\varepsilon$-strong BNE for some constant $\varepsilon$. 
\end{itemize}
\end{thm}

\begin{thm}
    \label{thm:nb_deviate} 
    Given any arbitrary sequence of instances, there always exists a sequence of regular strategy profiles $\{\stgp_{\ag}'\}_{\ag=1}^{\infty}$ such that $\acc(\stgp_{\ag}')$ converges to 1. 
\end{thm}

Theorem~\ref{thm:nb_acc} and Theorem~\ref{thm:nb_deviate} preserves the same reasoning as Theorem~\ref{thm:acc} and Theorem~\ref{thm:deviate}. More details as well as the characterization of the fidelity in the non-binary setting are in Appendix~\ref{apx:nb_results}.

%% file: Sections/6_conclusion.tex
\label{sec:conclusion}
We study the binary voting game where agents can coordinate in groups. We show that strategic voting always leads to the informed majority decision, while non-strategic behaviors sometimes fail. 
In particular, we show that a strategy profile is an $\varepsilon$-strong Bayes Nash Equilibrium with small $\varepsilon$ if and only if it leads to the ``correct'' decision with high probability. Moreover, we analyze the fidelity of the strategy profile and provide criteria for judging whether a strategy profile is an equilibrium based on \exshare. Applying the analysis to non-strategic voting, we characterize the conditions that informative and sincere voting lead to the informed majority decision. Our results stand on the framework where agents have endogenous preferences over outcomes contingent on some underlying state.

One limitation of our work is that our results are restricted to a setting with two alternatives. An interesting yet challenging future direction is to study the impact of strategic behavior in a setting with more than two alternatives. We expect more complicated results as \citet{Goertz14:Condorcet} show that informative voting may be an equilibrium but leads to the wrong alternative in a model of three alternatives.

Another interesting direction is to explore strategic iterative voting with information uncertainty. We expect iterative voting to be more powerful in aggregating information and able to simulate some sophisticated mechanisms. For example, the mechanism in~\citet{schoenebeck21wisdom} can be regarded as a two-round voting game where only the second round counts, and every agent votes informatively in the first round and plays a surprisingly popular strategy in the second round. \citet{Kavner2021Bliss} show a surprising result that strategic behaviors increase the social welfare of agents in iterative voting on average. Nevertheless, the behavior of strategic agents is even more complicated in iterative voting, and the analysis of equilibria will be highly challenging.


%% file: Appendix/1_FP_comparison.tex
\label{apx:fp}
\citet{Feddersen97:Voting} consider a two-alternative setting and present a (unique) class of equilibria that leads to the decision favored by the majority. We argue that Fedderson and Pesendorfer's work is drastically different from our work in many aspects. 

\paragraph{Setting} The most important and fundamental difference between Fedderson and Pesendorfer's work and ours is the setting. \citet{Feddersen97:Voting} adopt a model with continuous world states and an asymptotically large number of voters whose preferences are drawn from a distribution with full support on a continuum. Most of their results (the uniqueness of their equilibrium, for example) require continuity. We, on the other hand, consider discrete and finite world states and private signals. The solution concepts in their work is also different from ours. Fedderson and Pesendorfer consider symmetric Nash Equilibrium with no weakly dominated strategies. We, on the other hand, consider $\varepsilon$-strong Bayes Nash Equilibrium which precludes agents from coordinating with each other.


\paragraph{Equilibrium} Due to different settings and different solution concepts, the equilibria are also different from each other. In Fedderson and Pesendorfer's work, a Nash equilibrium consists of only three types of strategies: always vote for $\bA$, always vote for $\bR$, and vote informatively. In our work, on the other hand, a voting instance may have multiple distinct equilibria. Example~\ref{ex:motive3} shows scenarios where a Bayes Nash Equilibrium fails to be a strong Bayes Nash.
Moreover, Fedderson and Pesendorfer only show that the equilibrium will lead to the informed majority decision, while we also show that non-equilibrium strategy profiles will not lead to the informed majority decision. 

\paragraph{Types of agents} Although agents in Fedderson and Pesendorfer's work can also be classified into three groups based on their strategy, we argue, as \citet{schoenebeck21wisdom} argued, that there is ``fundamental difference in the motivation behind this classification''. As we follow the setting of Reference~ \citep{schoenebeck21wisdom}, our agent types reflect the preferences of agents among two alternatives. In Fedderson and Pesendorfer's work, however, agents choose their type by playing different strategies in some specific scheme so that the mechanism outputs the ``correct'' alternative. Therefore, the motivation of their classifying agents is to aggregate information and make a good decision, but not to reflect the preferences of agents. 


%% file: Appendix/2_notation_list.tex
\label{apx:notations}
Table~\ref{tab:notation} provides a list of frequently used notation in the main paper (binary setting).
\begin{table}[htbp]
\centering
\renewcommand{\arraystretch}{1.2}
\begin{tabular}{@{}ll@{}}
\toprule
Notation                                           & Meaning                                                                          \\ \midrule
$\ag$                                              & the total number of agents                                                       \\
$\bA, \bR$                                         & alternatives                                                                     \\
$\Wosset = \{L, H\}$                               & the set of world states                                                      \\
$\Sigset = \{l, h\}$                               & the set of private signals                                                   \\
$P_\wos$                                           & the prior belief on the probability of world state being $\wos$                  \\
$P_{\sig \wos}$   & the probability of receiving signal $\sig$ under world state $\wos$    \\
$\Thd$                                             & threshold of the majority vote                                                   \\
$\vt_\sag(\bA, \wos), \vt_\sag (\bR, \wos)$        & the utility of agent $\sag$                                                       \\
$B$                                                & the upper bound of the utility function                                                              \\
$\agf, \agu, \agc$                                 & the approximated fraction of three types of agents                               \\
$\stg, \stgp$                                      & strategy and strategy profile                                                    \\
$\bp_h, \bp_l$                                     & the probability of voting for $\bA$ when receiving signal $h$ ($l$)               \\
$\{\stgp_\ag\}_{\ag = 1}^{\infty}$                 & a sequence of strategy profiles                                                  \\
$\lp_{\wos}^{\bA}(\stgp), \lp_{\wos}^{\bR}(\stgp)$ & the probability that $\bA$ ($\bR$) becomes the winner under world state $\wos$   \\
$\acc(\stgp)$                                      & fidelity: probability that $\stgp$ leads to the informed majority decision \\
$\ut_\sag(\stgp)$                                  & the (ex-ante) expected utility of agent $\sag$                                   \\
$\hthd_H^\ag, \hthd_L^\ag, \hthd^\ag$              & the \exshare                                                      \\ \bottomrule
\end{tabular}
\caption{Frequently used notation in the binary setting\label{tab:notation}}
\end{table}

%% file: Appendix/2_Prof1_Proof.tex
\label{apx:no_BNE}
\begin{propnb}{prop:no_BNE}
For any $N_0 \in \mathbb{N}$, there exists an instance of $N > N_0$ agents, which does not exist a strong Bayes Nash equilibrium. 
\end{propnb}

\begin{proof}
For any $N_0$, we construct an instance of $N = 2N_0+3$ agents. The agents consist of three parts: $F$ is a set of $N_0+1$ friendly agents. $C$ is a set of two contingent agents. And $U$ is a set of $N_0$ unfriendly agents. The valuation function of each set is shown in the table. The threshold $\Thd =0.5$. The prior distribution $P_L = P_H = 0.5$. The signal distribution $P_{hH} = P_{lL}=0.8$, and $P_{lH} = P_{hL} = 0.2$.
    \begin{table}[htbp]
\centering
\begin{tabular}{@{}ccccc@{}}
\toprule
agents & $\vt(H, \bA)$ & $\vt(L, \bA)$ & $\vt(L, \bR)$ & $\vt(H, \bR)$ \\ \midrule
$F$    & 100           & 99            & 1             & 0            \\
$C$ & 90            & 0             & 100           & 0            \\ 
$U$ & 1 & 0 & 100 & 99\\
\bottomrule
\end{tabular}
\end{table}

Still, consider three strategy profiles: 
\begin{itemize}
    \item $\stgp_1$: In $F$, $N_0$ agents always vote for $\bA$, and one agent votes informatively. Two agents of $C$ vote informatively. $N_0$ agents of $U$ always vote for $\bR$. 
    \item $\stgp_2$: $N_0+1$ agents of $F$ always vote for $\bA$.  Two agents of $C$ vote informatively. $N_0$ agents of $U$ always vote for $\bR$. 
    \item $\stgp_3$: $N_0+1$ agents of $F$ always vote for $\bA$.  One $C$ agent votes informatively, and the other agent always votes for $\bR$. $N_0$ agents of $U$ always vote for $\bR$. 
\end{itemize}
The expected utility of each type of agents in three profiles is in the table as follows. 
    \begin{table}[htbp]
\begin{tabular}{@{}cccc@{}}
\toprule
agents & $\stgp_1$ & $\stgp_2$ & $\stgp_3$ \\ \midrule
$F$    & 50.396    & 66.14     & 50.3      \\
$C$ & 85.12     & 75.2      & 76        \\ 
$U$ & 50.396 & 34.46 & 50.3\\
\bottomrule
\end{tabular}
\end{table}
The three profiles form a deviating cycle of $\stgp_1 \to \stgp_2\to \stgp_3\to \stgp_1$. In $\stgp_1$, the $F$ agent who votes informatively has the incentive to always vote for $\bA$, and the profile becomes $\stgp_2$. In $\stgp_2$, a $C$ agent has the incentive to always vote for $\bR$, and the profile becomes $\stgp_3$. And in $\stgp_3$, the group of a $F$ agent and two $C$ agents have the incentives to turn to informative voting, and the profile becomes $\stgp_1$. 

Now we show that for any other strategy profile $\stgp$, there exists a group of agents with incentives to deviate. 

{\bf Case 1:} $\bA$ is always the winner. In this case, at least $N_0+2$ agents always vote for $\bA$, the expected utility of $U$ agents is 0.5, and the expected utility of $C$ agents is 45. If there exists a group of $U$ agents such that $\bA$ is not always the winner after they deviate to always votes for $\bR$, their expected utility will increase. Otherwise, there are still at least $N_0 + 2$ agents who always vote for $\bA$ even if all $U$ agents always vote for $\bR$. Therefore, at least one $C$ agent always votes for $\bA$. Now consider the deviating group $C \cup U$, where all agents in the group turn to always voting for $\bR$. Then $\bR$ will always be the winner. The expected utility of $U$ agents increases to 99.5, and the expected utility of $C$ agents increases to 50. Therefore, in this case, $\stgp$ is not a strong BNE. 

{\bf Case 2:} $\bR$ is always the winner. In this case, at least $N_0+2$ agents always vote for $\bR$, the expected utility of $F$ agents is 0.5, and the expected utility of $C$ agents is 50. Similarly, if all agents in $F$ voting for $\bA$ can reverse the decision, their expected utility will increases. Otherwise, all agents in $U$ and $C$ always vote for $\bR$. Then, the group $F\cup C$ has incentives to deviate, where $F$ agents always vote for $\bA$, and $C$ agents vote informatively. The profile after the deviation is $\stgp_2$, in which both $F$ and $C$ agents have high expected utilities. Therefore, in this case, $\stgp$ is not a strong BNE. 

{\bf Case 3:} Neither of the alternatives is always the winner. In this case, any $F$ agent not always voting for $\bA$ and $U$ agent not always voting for $\bR$ has incentives to deviate to get a higher probability that their preferred alternative to be selected. Now suppose all $F$ agents always vote for $\bA$, and $U$ agents always vote for $\bR$. In this case, we show that the expected utility of $C$ agents is uniquely maximized by one agent voting informatively and the other agent always voting for $\bR$, just as $\stgp_3$. Therefore, for any $\stgp\neq \stgp_3$, $C$ agents have incentives to deviate to $\stgp_3$. And $\stgp_3$ itself is dominated by $\stgp_1$. 

Now without loss of generality, we index the two agents in $C$ as agent 1 and agent 2. Suppose the strategy of agent 1 is $\stg_1 = (\bp_l^1, \bp_h^1)$, and the strategy of agent 2 is $\stg_2 = (\bp_l^2, \bp_h^2)$. Since there are $N_0+1$ votes for $\bA$ and $N_0$ votes for $\bR$, $\bR$ is the winner only when both two agents vote for $\bR$. Let $X_1$ and $X_2$ be the random variable that denotes the vote of agent 1 and agent 2 respectively. $X_i =1 $ stands for ``agent $i$ votes for $\bA$'', and $X_i = 0$ stands for ``agent $i$ votes for $\bR$''. 
Therefore, the probability of $\bA$ being winner under $H$ state is 
\begin{align*}
    \lp_H^{\bA}(\stgp) =& 1 - \Pr[X_1 =0\wedge X_2 = 0\mid \Wosrv = H]\\
    = &1 - \Pr[X_1 =0\mid \Wosrv = H]\cdot \Pr[X_2 =0\mid \Wosrv = H]\\
    = & 1 - (0.8\cdot (1-\bp_h^1) + 0.2 \cdot (1-\bp_l^1))(0.8\cdot (1-\bp_h^2) + 0.2 \cdot (1-\bp_l^2)). 
\end{align*}
Similarly, 
\begin{align*}
    \lp_L^{\bR}(\stgp) = & \Pr[X_1 =0\wedge X_2 = 0\mid \Wosrv = L]\\
    =&(0.8\cdot (1-\bp_l^1) + 0.2 \cdot (1-\bp_h^1))(0.8\cdot (1-\bp_l^2) + 0.2 \cdot (1-\bp_h^2)). 
\end{align*}
And 
\begin{equation*}
    \ut_1(\stgp) = \ut_2(\stgp) = 0.5 \cdot (90 \cdot\lp_H^{\bA}(\stgp) + 100 \cdot  \lp_L^{\bR}(\stgp)). 
\end{equation*}
Let $x_1 = 1 - \bp_h^1, y_1 = 1-\bp_l^1, x_2 = 1 - \bp_h^2$, and $ y_2 = 1-\bp_l^2$, and expand the expected utility, we get 
\begin{equation*}
    \ut_1(\stgp)  = -26.8x_1x_2 + 0.8x_1y_2 + 0.8x_2y_1 + 30.2 y_1y_2 +45.  
\end{equation*}
We consider the variables to maximize the expected utility. 
Note that given $x_1$ and $y_1$, if at least one of $x_1$ and $y_1$ does not equal to zero, $y_2 =1$ is a necessary condition for $\ut_2(\stgp)$ to be maximized, and vice versa. In the special case where four variables are all 0, $\stgp$ stands for both $C$ agents always vote for $\bA$. In this case, $\bA$ is always the winner, and $C$ agents have incentives to deviate to always vote for $\bR$ as shown in Case 1. Now, given at least one of four values is non-zero, we must have $y_1 = y_2 = 1$ to maximize the expected utility. Then the expected utility turns to
\begin{equation*}
    \ut_1(\stgp) = -26.8x_1x_2 + 0.8x_1 + 0.8x_2 + 75.2.  
\end{equation*}
Then it's not hard to see that $\ut_1(\stgp)$ is uniquely maximized when one of $x_1$ and $x_2$ takes 0 and the other takes 1. Then $\stgp$ is exactly $\stgp_3$. Therefore, the expected utility of $C$ agents is maximized on $\stgp_3$. 

Consequently, we show that in every case a strategy profile is not a strong BNE. Therefore, a strong BNE does no exist in the instance. 
\end{proof}

%% file: Appendix/3_nb_setting.tex
\subsection{Additional Setting}
\label{apx:setting}
In this section, we propose notions in non-binary settings that are used in technical proofs and not mentioned in the main paper. 

\paragraph{Fraction of Agents}
Given a specific $\ag$, we use $\tf, \tc,$ and $\tu$ to denote the number of each type of agents, and use $\hagf = \frac{\tf}{\ag}, \hagc = \frac{\tc}{\ag}$ and $\hagu = \frac{\tu}{\ag}$ to denote the fraction of each type of agents. Note that $\hagf, \hagc$, and $\hagu$ is dependent to $\ag$.

In the non-binary setting, we have $\hagf < \Thd$ and $\hagu <  \Thd$ for all sufficiently large $\ag$, which means that neither friendly nor unfriendly agents can dominate the vote. Note that this is guaranteed by the assumption that both $\Lowset$ and $\Highset$ are non-empty. 

\begin{prop}
\label{prop:middle-case}  There exists a constant $\ag_\mu> 0$, s.t. for all $\ag> \ag_\mu$, $\hagf < \Thd$ and $\hagu < \Thd$. 
\end{prop}

\begin{proof}
First, we show that $\hagf < \Thd$.  For a candidate-friendly agent, we have $\Lowset_{\sag} \subsetneq \Lowset$. Therefore, $L_{\sag} <L$, and $\sag$ prefer $\bA$ if the world state is $L$. On the other hand, an agent vote for $\bA$ in $L$ is a friendly agent according to the definition. Therefore, the set of friendly agents is exactly the set of agents who will vote $\bA$ in world state $L$. Therefore we have
\begin{align*}
    \tf = & |\agset(L, \bA)|\\
    =& \ag - \lfloor \agr_L \cdot T\rfloor\\
    \le & (1 - \agr_L )\cdot \ag +1\\
    = & \aga_L \ag +1 \\
    =& \ag(\aga_L +\frac1T). 
\end{align*}
According to the definition of $L$, we know that $\aga_L <\Thd$ and $\aga_L$ is independent from $\ag$. Therefore, there must exist a $\ag_\mu > 0$ s.t. for all  $\ag > \ag_\mu$, $\aga_L + \frac{1}{\ag} < \Thd$, and therefore $\hagf < \Thd$. 

Then we show that $\hagu < 1 - \Thd$. Similarly, the set of unfriendly agents is exactly the set of agents who will vote $\bR$ in world state $H$. Therefore, 
\begin{align*}
    \tu = & |\agset(H, \bR)|\\
    =&  \lfloor \agr_H \cdot \ag\rfloor\\
    \le & \agr_H \cdot \ag\\
    < & (1-\Thd)\cdot \ag. 
\end{align*}
The last inequality is due to the definition of $H$. Therefore, we have $\hagu < 1-\Thd$. 
\end{proof}
For the non-binary results we all assume that $\ag >\ag_\mu$, therefore $\hagf < \Thd$ and $\hagu < \Thd$. 

Then we define the approximate fraction of each type of agents $\agf, \agc,$ and $\agu$ that are independent of $\ag$. As we have shown in Proposition~\ref{prop:middle-case}, the set of friendly agents is exactly the set of agents who will vote $\bA$ in $\Low$, and the set of unfriendly agents is exactly the set of agents who will vote $\bR$ in $\High$. Therefore, we define $\agf, \agc,$ and $\agu$ as follows:
\begin{enumerate}
    \item $\agf = \aga_\Low$.
    \item $\agu = \agr_\High$.
    \item $\agc = 1 - \agf - \agu$. 
\end{enumerate}
According to the definition we have $\agf < \Thd$ and $\agu < 1- \Thd$. 
And it's not hard to verify that $\tf = \ag - \lfloor (1-\agf) \cdot \ag\rfloor$, $\tu = \lfloor \agu \cdot \ag\rfloor$, and $\tc = \lfloor (1-\agf) \cdot \ag\rfloor - \lfloor \agu \cdot \ag\rfloor$.

\paragraph{Error Rate} Error rate is complimentary of the fidelity. Given a strategy profile $\stgp$, let $\lp_{\wos}^{\bA}(\stgp)$ ($\lp_{\wos}^{\bR}(\stgp)$, respectively) be the probability that $\bA$ ($\bR$, respectively) becomes the winner when world state is $\wos$. We can define the error rate $\err$ of the mechanism:
\begin{align}
    \err(\stgp) = & \sum_{\wos\in \Lowset}P_{\wos}\cdot \lp_{\wos}^{\bA}(\stgp) + \sum_{\wos\in \Highset}P_{\wos} \cdot \lp_{\wos}^{\bR}(\stgp).
\end{align}
 Note that $\acc(\stgp) + \err(\stgp) = 1$. 

 \subsection{Berry-Esseen Theorem}
\label{apx:berry}
In this section, we recall the technical theorem which will be used in the proof of our result. Berry-Esseen Theorem~\citep{Berry41Accuracy, Esseen42Liapunoff} bounds the difference  between the distribution of the sum of independent variables and the normal distribution.
\begin{dfn}[Berry-Esseen Theorem]
Let $X_1, X_2, \cdots, X_n$ be $n$ independent variables with $E[X_i] = 0$, $Var(X_i) = \sigma_i^2 > 0$, and $E[|X_i|^3] = \rho_i < \infty$. Let
$$S_n = \frac{X_1+X_2+\cdots+ X_n}{\sqrt{\sigma_1^2+\sigma_2^2+\cdots+\sigma_n^2}}.$$ Denote $F_n$ to be the CDF of $S_n$, and $\Phi$ to be the CDF of the standard normal distribution. Then there exists an absolute constant $C_0$ s.t. for all $n$, 
\begin{equation*}
    \sup_{s\in \mathbb{R}} |F_n(x) - \Phi(x)| \le C_0 \cdot \frac{\rho_1 + \rho_2 + \cdots + \rho_n}{(\sigma_1^2+\sigma_2^2+\cdots+\sigma_n^2)^{3/2}}.
\end{equation*}
\end{dfn}

The upper bound of $C_0$ is estimated to be 0.5600 by Reference \citep{Shevtsova10Improvement}.

%% file: Appendix/4_nb_results.tex
\label{apx:nb_results}
In this section, we present our theoretical results in the non-binary setting.  All the theorems are natural extensions from the binary setting to the non-binary setting, and the proofs also hold for the binary setting. We'll give remarks to show how the convert the proofs into the binary setting. 




\subsection{Lemma~\ref{lem:nb_nash} (Lemma~\ref{lem:nash} for binary setting)}
\label{apx:nb_nash}
First, we show that a strategy profile with high fidelity will lead to an $\varepsilon$-strong BNE with a sufficiently small $\varepsilon$.

\begin{lem}
\label{lem:nb_nash}
Let $e(\ag)$ be a function of $\ag$. For any $\ag > \ag_\mu$ and any regular strategy profile $\stgp^*$ with $\ag$ agents, if $\stgp^*$ satisfies $\acc(\stgp^*) \ge 1 - e(\ag)$ (equivalently, $\err(\stgp^*) \le e(\ag)$), then $\stgp^*$ is an $\varepsilon$-strong Bayes Nash where $\varepsilon = \Wos B((\Wos-1)B+1)\cdot e(\ag)$. 
\end{lem}

Lemma~\ref{lem:nb_nash} is a natural extension of Theorem B.4 in \citet{schoenebeck21wisdom}. 

\begin{proof}
Since we have $\hagf < \Thd$ and $\hagu< 1 - \Thd$, we do not need to consider the friendly/unfriendly agent dominating case. Consider a deviating strategy profile $\stgp'$, there are two possible cases:
\begin{enumerate}
    \item  $I(\stgp') <((\Wos-1)B+1)\cdot e(\ag)$
    \item $I(\stgp') \ge((\Wos-1)B+1)\cdot e(\ag)$
\end{enumerate}

\noindent{\bf Case 1.} In the first case, the error rate of $\stgp'$ is small. Since both strategy profile has high fidelity, we show that agents cannot gain large utility from deviating. 
\begin{claim}
If $I(\stgp') <((\Wos-1)B+1)\cdot e(\ag)$, then $\ut_{\sag}(\stgp') - \ut_{\sag}(\stgp^*) \le \varepsilon$ for all agents. 
\end{claim}

We know that $\err(\stgp) = \sum_{\wos\in \Lowset}P_{\wos}\cdot \lp_{\wos}^{\bA}(\stgp) + \sum_{\wos\in \Highset}P_{\wos} \cdot \lp_{\wos}^{\bR}(\stgp).$ Therefore, we have $0\le \lp_{\wos}^{\bA}(\stgp')\le \frac{((\Wos-1)B+1)\cdot e(\ag)}{P_{\wos}}$ for all $\wos\in \Lowset$ and $0\le \lp_{\wos}^{\bR}(\stgp')\le \frac{((\Wos-1)B+1)\cdot e(\ag)}{P_{\wos}}$ for all $\wos \in \Highset$. At the same time, since $I(\stgp^*)\le e(\ag)$, we have $0\le \lp_{\wos}^{\bA}(\stgp^*)\le \frac{e(\ag)}{P_{\wos}}$ for all $\wos\in \Lowset$ and $0\le \lp_{\wos}^{\bR}(\stgp^*)\le \frac{e(\ag)}{P_{\wos}}$ for all $\wos \in \Highset$. Therefore, for every $\wos\in \Lowset$, we have 
\begin{align*}
   | \lp_{\wos}^{\bR}(\stgp') - \lp_{\wos}^{\bR}(\stgp^*)| = | \lp_{\wos}^{\bA}(\stgp') - \lp_{\wos}^{\bA}(\stgp^*)| \le &  \max(\lp_{\wos}^{\bA}(\stgp'), \lp_{\wos}^{\bA}(\stgp^*) )\le \frac{((\Wos-1)B+1)\cdot e(\ag)}{P_{\wos}}. 
\end{align*}
And for every $\wos\in \Highset$, we have 
\begin{align*}
   | \lp_{\wos}^{\bA}(\stgp') - \lp_{\wos}^{\bA}(\stgp^*)| = | \lp_{\wos}^{\bR}(\stgp') - \lp_{\wos}^{\bR}(\stgp^*)| \le &  \max(\lp_{\wos}^{\bR}(\stgp'), \lp_{\wos}^{\bR}(\stgp^*) )\le \frac{((\Wos-1)B+1)\cdot e(\ag)}{P_{\wos}}. 
\end{align*}
Therefore, for any $\wos\in\Wosset$, $| \lp_{\wos}^{\bA}(\stgp') - \lp_{\wos}^{\bA}(\stgp^*)|\le \frac{((\Wos-1)B+1)\cdot e(\ag)}{P_{\wos}}. $
Now for an arbitrary agent $\sag$, we consider her gain in deviation:
\begin{align*}
    &\ut_{\sag}(\stgp') - \ut_{\sag}(\stgp')\\
    = & \sum_{\wos = 1}^{\Wos} P_{\wos}((\lp_{\wos}^{\bA}(\stgp') - \lp_{\wos}^{\bA}(\stgp^*))\cdot\vt_{\sag}(\wos, \bA) + (\lp_{\wos}^{\bR}(\stgp') - \lp_{\wos}^{\bR}(\stgp^*))\cdot \vt_{\sag}(\wos, \bR))\\
    = & \sum_{\wos\in \Lowset_{\sag}} P_{\wos}(\vt_{\sag}(\wos, \bR) - \vt_{\sag}(\wos, \bA))(\lp_{\wos}^{\bA}(\stgp^*) - \lp_{\wos}^{\bA}(\stgp')) \\
    & + \sum_{\wos\in \Highset_{\sag}} P_{\wos}(\vt_{\sag}(\wos, \bA) - \vt_{\sag}(\wos, \bR))(\lp_{\wos}^{\bA}(\stgp') - \lp_{\wos}^{\bA}(\stgp^*))
\end{align*}
\begin{align*}
    \le &\sum_{\wos\in \Lowset_{\sag}} P_{\wos}(\vt_{\sag}(\wos, \bR) - \vt_{\sag}(\wos, \bA))\frac{((\Wos-1)B+1)\cdot e(\ag)}{P_{\wos}}\\
    & + \sum_{\wos\in \Highset_{\sag}} P_{\wos}(\vt_{\sag}(\wos, \bA) - \vt_{\sag}(\wos, \bR))\frac{((\Wos-1)B+1)\cdot e(\ag)}{P_{\wos}}\\
    \le & \sum_{\wos\in \Lowset_{\sag}} P_{\wos}\cdot B\cdot \frac{((\Wos-1)B+1)\cdot e(\ag)}{P_{\wos}} + \sum_{\wos\in \Highset_{\sag}} P_{\wos}\cdot B\cdot \frac{((\Wos-1)B+1)\cdot e(\ag)}{P_{\wos}}\\
    = & \Wos B((\Wos-1)B+1)\cdot e(\ag)\\
    = & \varepsilon.
\end{align*}

\noindent{\bf Case 2.} 
In the second case, we show that $\stgp'$ cannot be a successful deviating strategy in the following steps:
\begin{enumerate}
    \item Contingent agents get strictly less utility in $\stgp'$ than in $\stgp^*$ and thus have no incentive to deviate.
    \item For a friendly agent $t_1$ and an unfriendly agent $t_2$, if one of them gain more than $\Delta  = B\cdot e(\ag)$ from deviation, the other will get strictly less utility. Therefore, $D$ contains either only friendly agents or only unfriendly agents.
    \item $D$ with only one type of pre-determined agents cannot gain utility more than $\varepsilon$. 
\end{enumerate}
\begin{claim}
If  $I(\stgp') \ge((\Wos-1)B+1)\cdot e(\ag)$, then $\ut_{\sag}(\stgp') - \ut_{\sag}(\stgp^*) < 0$ for contingent agents. 
\end{claim}
Firstly, we have $\acc(\stgp^*) \ge \acc(\stgp') + (\Wos-1)B\cdot e(\ag)$. This is because $I(\stgp') \ge((\Wos-1)B+1)\cdot e(\ag)$ and $I(\stgp^*) \le e(\ag)$. 
Then we consider the utility difference of an arbitrary contingent agent $\sag$:
\vspace{-0.05cm}
\begin{align*}
     \ut_{\sag}(\stgp') - \ut_{\sag}(\stgp^*)
    =  & \sum_{\wos\in \Lowset} P_{\wos}(\vt_{\sag}(\wos, \bR) - \vt_{\sag}(\wos, \bA))(\lp_{\wos}^{\bR}(\stgp') - \lp_{\wos}^{\bR}(\stgp^*))\\
    & + \sum_{\wos\in \Highset} P_{\wos}(\vt_{\sag}(\wos, \bA) - \vt_{\sag}(\wos, \bR))(\lp_{\wos}^{\bA}(\stgp') - \lp_{\wos}^{\bA}(\stgp^*))
\end{align*}
\vspace{-0.05cm}
Recall that $\vt_{\sag}(\wos, \bR) - \vt_{\sag}(\wos, \bA) > 0$ for all $\wos\in \Lowset$ and $\vt_{\sag}(\wos, \bA) - \vt_{\sag}(\wos, \bR) > 0$ for all $\wos\in\Highset$. We consider two cases:

\noindent{\bf(1) } For all world state $\wos$, $\stgp^*$ leads to the informed majority decision with a higher probability than $\stgp'$. That is, for all $\wos\in \Lowset$,  $\lp_{\wos}^{\bR}(\stgp') - \lp_{\wos}^{\bR}(\stgp^*) \le 0$, and for all $\wos\in \Highset$, $\lp_{\wos}^{\bA}(\stgp') - \lp_{\wos}^{\bA}(\stgp^*)\le 0$. In this case, there is at least one $\wos$ such that the inequality is strict. Therefore we have $\ut_{\sag}(\stgp') - \ut_{\sag}(\stgp') < 0$. 

\noindent{\bf (2)} There exists a world state $\wos$ in which $\stgp'$ leads to the informed majority decision with a higher probability than $\stgp^*$. That is, there exists $\Lowset'\subseteq \Lowset$ and $\Highset'\subseteq\Highset$ s.t. $\Lowset'\cup\Highset' \neq \emptyset$, and for $\wos\in \Lowset'$,  $\lp_{\wos}^{\bR}(\stgp') - \lp_{\wos}^{\bR}(\stgp^*) > 0$; for $\wos\in \Highset'$, $\lp_{\wos}^{\bA}(\stgp') - \lp_{\wos}^{\bA}(\stgp^*)>0$. Note that $\Lowset' \cup \Highset' \neq \Wosset$ always holds, otherwise we will have $\acc(\stgp') >\acc(\stgp^*)$ which is a contradiction. Then we have 
\vspace{-0.05cm}
\begin{align*}
    &\sum_{\wos\in \Lowset\setminus\Lowset'} P_{\wos} (\lp_{\wos}^{\bR}(\stgp^*) - \lp_{\wos}^{\bR}(\stgp')) + \sum_{\wos\in \Highset\setminus\Highset'} P_{\wos} (\lp_{\wos}^{\bA}(\stgp^*) - \lp_{\wos}^{\bA}(\stgp'))\\
    = & \acc(\stgp^*) - \acc(\stgp') - \sum_{\wos\in \Lowset'} P_{\wos} (\lp_{\wos}^{\bR}(\stgp^*) - \lp_{\wos}^{\bR}(\stgp'))- \sum_{\wos\in \Highset'} P_{\wos} (\lp_{\wos}^{\bA}(\stgp^*) - \lp_{\wos}^{\bA}(\stgp')\\
    > & \acc(\stgp^*) - \acc(\stgp') \\
    \ge & (\Wos-1)B\cdot e(\ag). 
\end{align*}
\vspace{-0.05cm}
At the same time, for all $\wos\in\Lowset$, we have $\lp_{\wos}^{\bR}(\stgp') - \lp_{\wos}^{\bR}(\stgp^*)\le \frac{e(\ag)}{P_{\wos}}$, and for all $\wos\in \Highset'$, $\lp_{\wos}^{\bA}(\stgp') - \lp_{\wos}^{\bA}(\stgp^*)\le \frac{e(\ag)}{P_{\wos}}$. Therefore, the difference of utility will have 
\vspace{-0.05cm}
\begin{align*}
    &\ut_{\sag}(\stgp') - \ut_{\sag}(\stgp')\\
    < & -(\Wos - 1)B\cdot e(\ag) + \sum_{\wos\in \Lowset'\cup\Highset'} P_\wos \cdot \frac{e(\ag)}{P_{\wos}}\cdot B\\
    \le  &  -(\Wos - 1)B\cdot e(\ag) + (\Wos-1)B\cdot e(\ag)\\
    =&0.
\end{align*}
\vspace{-0.05cm}
In the second line, the first term comes from the sum of all terms for $\wos\in \Wosset\setminus(\Lowset'\cup\Highset')$:
\begin{align*}
    &\sum_{\wos\in \Lowset\setminus\Lowset'} P_{\wos} (\lp_{\wos}^{\bR}(\stgp') - \lp_{\wos}^{\bR}(\stgp^*))(\vt_{\sag}(\wos, \bR) - \vt_{\sag}(\wos, \bA)) \\
    & + \sum_{\wos\in \Highset\setminus\Highset'} P_{\wos} (\lp_{\wos}^{\bA}(\stgp') - \lp_{\wos}^{\bA}(\stgp^*))(\vt_{\sag}(\wos, \bA) - \vt_{\sag}(\wos, \bR))\\ 
    = & - \sum_{\wos\in \Lowset\setminus\Lowset'} P_{\wos} (\lp_{\wos}^{\bR}(\stgp^*) - \lp_{\wos}^{\bR}(\stgp'))(\vt_{\sag}(\wos, \bR) - \vt_{\sag}(\wos, \bA))\\
    & - \sum_{\wos\in \Highset\setminus\Highset'} P_{\wos} (\lp_{\wos}^{\bA}(\stgp^*) - \lp_{\wos}^{\bA}(\stgp'))(\vt_{\sag}(\wos, \bA) - \vt_{\sag}(\wos, \bR))\\
    \le& - \sum_{\wos\in \Lowset\setminus\Lowset'} P_{\wos} (\lp_{\wos}^{\bR}(\stgp^*) - \lp_{\wos}^{\bR}(\stgp'))\cdot 1 - \sum_{\wos\in \Highset\setminus\Highset'} P_{\wos} (\lp_{\wos}^{\bA}(\stgp^*) - \lp_{\wos}^{\bA}(\stgp'))\cdot 1\\
    < & -(\Wos-1)B\cdot e(\ag).
\end{align*}
And the second term comes from the sum of all terms for $\wos \in \Lowset'\cup\Highset'$. Note that since $\Lowset'\cup\Highset'\subsetneq \Wosset$, we have $|\Lowset'\cup\Highset'| \le \Wos-1.$

After excluding contingent agents from the deviating group, we show that any deviating group $D$ cannot contain both friendly and unfriendly agents. 
\begin{claim}
Suppose $\sag_1$ is an arbitrary friendly agent, and $\sag_2$ is an arbitrary unfriendly agent. For $\Delta  = \Wos\cdot B(B+1)\cdot e(\ag)$, we have 
\begin{enumerate}
    \item If $ \ut_{\sag_1}(\stgp') - \ut_{\sag_1}(\stgp^*)> \Delta$, $ \ut_{\sag_2}(\stgp') - \ut_{\sag_2}(\stgp^*)< 0$.
    \item If $ \ut_{\sag_2}(\stgp') - \ut_{\sag_2}(\stgp^*)> \Delta$, $ \ut_{\sag_1}(\stgp') - \ut_{\sag_1}(\stgp^*)< 0$.
\end{enumerate}
\end{claim}
We will only prove (1) since the reasoning of (1) and (2) are similar. Without loss of generality, let $\sag_1=1$ and $\sag_2 = 2$. Suppose $\ut_1(\stgp') - \ut_1(\stgp^*)> \Delta$, we'll show that $\ut_2(\stgp') - \ut_2(\stgp^*)< 0.$ According to the definition of friendly and unfriendly agents, we have $\Lowset_1\subsetneq\Lowset\subsetneq \Lowset_2$, and $\Highset_2\subsetneq\Highset\subsetneq\Highset_1$.

We'll mainly use three facts in the proof: 
\begin{enumerate}
    \item for all $\wos\in \Lowset$, $\lp_{\wos}^{\bA}(\stgp^*) - \lp_{\wos}^{\bA}(\stgp')\le \frac{e(\ag)}{P_\wos} -\lp_{\wos}^{\bA}(\stgp')  \le \frac{e(\ag)}{P_\wos}$; for all $\wos\in \Highset$, $\lp_{\wos}^{\bA}(\stgp') - \lp_{\wos}^{\bA}(\stgp^*)\le \lp_{\wos}^{\bA}(\stgp') - (1 - \frac{e(\ag)}{P_\wos}) \le \frac{e(\ag)}{P_\wos}$.
    \item for both $\sag = 1$ and $\sag=2$, $\vt_{\sag}(\wos, \bA) -\vt_{\sag}(\wos, \bR)$ is negative for all $\wos\in \Lowset_1$, and positive for all $\wos\in \Highset_2$. 
    \item for both $\sag = 1$ and $\sag=2$, $\vt_{\sag}(\wos, \bA) - \vt_{\sag}(\wos, \bR)$ is increasing in $\wos$. 
\end{enumerate}
We first consider the difference of 1's expected utility:
\begin{align*}
    \Delta < \ut_1(\stgp') - \ut_1(\stgp^*)\notag
    = \sum_{\wos=1}^{\Wos} P_{\wos}(\vt_1(\wos, \bA) - \vt_1(\wos, \bR))(\lp_{\wos}^{\bA}(\stgp') - \lp_{\wos}^{\bA}(\stgp^*))\notag\\
\end{align*}
We can categorize $\wos$ into three parts:
\begin{itemize}
    \item $\wos\in\Lowset_1$. For these $\wos$, $\vt_{\sag}(\wos, \bA) -\vt_{\sag}(\wos, \bR) < 0$ and $\lp_{\wos}^{\bA}(\stgp^*) - \lp_{\wos}^{\bA}(\stgp')\le\frac{e(\ag)}{P_\wos}$. Therefore, we have 
    \begin{align*}
        &\sum_{\wos\in\Lowset_1} P_{\wos}(\vt_1(\wos, \bA) - \vt_1(\wos, \bR))(\lp_{\wos}^{\bA}(\stgp') - \lp_{\wos}^{\bA}(\stgp^*))\\
        = &\sum_{\wos\in\Lowset_1} P_{\wos}(\vt_1(\wos, \bR) - \vt_1(\wos, \bA))(\lp_{\wos}^{\bA}(\stgp^*) - \lp_{\wos}^{\bA}(\stgp'))\\
        \le &\sum_{\wos\in\Lowset_1}P_{\wos}\cdot B\cdot \frac{e(\ag)}{P_\wos}\\
        =& \sum_{\wos\in\Lowset_1}B\cdot e(\ag).
    \end{align*}
    \item $\wos\in \Highset_2$. For these $\wos$, $\vt_{\sag}(\wos, \bA) -\vt_{\sag}(\wos, \bR) >0$ and $\lp_{\wos}^{\bA}(\stgp') - \lp_{\wos}^{\bA}(\stgp^*)\le\frac{e(\ag)}{P_\wos}$. Therefore, we have
    \begin{align*}
        &\sum_{\wos\in \Highset_2} P_{\wos}(\vt_1(\wos, \bA) - \vt_1(\wos, \bR))(\lp_{\wos}^{\bA}(\stgp') - \lp_{\wos}^{\bA}(\stgp^*))\\
        \le &\sum_{\wos\in \Highset_2}P_{\wos}\cdot B\cdot \frac{e(\ag)}{P_\wos}\\
        =& \sum_{\wos\in \Highset_2}B\cdot e(\ag).
    \end{align*}
    \item $\wos\in \Wosset\setminus(\Lowset_1 \cup\Highset_2)$. Note that this part can be divided into two parts: $\wos\in \Lowset\cap\Highset_1$ and $\Lowset_2\cap\Highset$. We'll deal with this part later. 
\end{itemize}
Therefore, 
\begin{align*}
    \Delta < & \ut_1(\stgp^*) - \ut_1(\stgp')\notag\\
    = & \sum_{\wos=1}^{\Wos} P_{\wos}(\vt_1(\wos, \bA) - \vt_1(\wos, \bR))(\lp_{\wos}^{\bA}(\stgp') - \lp_{\wos}^{\bA}(\stgp^*))\notag\\
    \le &\sum_{\wos \in \Wosset\setminus(\Lowset_1 \cup \Highset_2)} P_{\wos}(\vt_1(\wos, \bA) - \vt_1(\wos, \bR))(\lp_{\wos}^{\bA}(\stgp') - \lp_{\wos}^{\bA}(\stgp^*)) + \sum_{\wos \in \Lowset_1 \cup \Highset_2} B\cdot e(\ag)\\
    = & \sum_{\wos \in \Lowset \cap \Highset_1} P_{\wos}(\vt_1(\wos, \bA) - \vt_1(\wos, \bR))(\lp_{\wos}^{\bA}(\stgp') - \lp_{\wos}^{\bA}(\stgp^*))\notag\\
    & - \sum_{\wos \in \Lowset_2 \cap \Highset} P_{\wos}(\vt_1(\wos, \bA) - \vt_1(\wos, \bR))(\lp_{\wos}^{\bA}(\stgp^*) - \lp_{\wos}^{\bA}(\stgp'))+
    \sum_{\wos \in \Lowset_1 \cup \Highset_2} e(\ag)\cdot B\notag.
\end{align*}
Now we start deal with the first two terms with fact 1 and 3. 
\begin{enumerate}
    \item For the first term $\sum_{\wos \in \Lowset \cap \Highset_1} P_{\wos}(\vt_1(\wos, \bA) - \vt_1(\wos, \bR))(\lp_{\wos}^{\bA}(\stgp') - \lp_{\wos}^{\bA}(\stgp^*))$, we have $(\vt_1(\wos, \bA) - \vt_1(\wos, \bR)) > 0$ and $(\lp_{\wos}^{\bA}(\stgp') - \lp_{\wos}^{\bA}(\stgp^*)) \ge -\frac{e(\ag)}{P_{\wos}}.$ we want to replace the $\wos$ in $\vt_1$ into $L$, which will make $(\vt_1(\wos, \bA) - \vt_1(\wos, \bR))$ larger. At the same time, we want the whole term to become larger as well. Therefore, we add some term into $(\lp_{\wos}^{\bA}(\stgp') - \lp_{\wos}^{\bA}(\stgp^*))$ as follows:
    \begin{align*}
        &\sum_{\wos \in \Lowset \cap \Highset_1} P_{\wos}(\vt_1(\wos, \bA) - \vt_1(\wos, \bR))(\lp_{\wos}^{\bA}(\stgp') - \lp_{\wos}^{\bA}(\stgp^*))\\
        = & \sum_{\wos \in \Lowset \cap \Highset_1} P_{\wos}(\vt_1(\wos, \bA) - \vt_1(\wos, \bR))(\lp_{\wos}^{\bA}(\stgp') - \lp_{\wos}^{\bA}(\stgp^*) + \frac{e(\ag)}{P_\wos})\\
        & - \sum_{\wos \in \Lowset \cap \Highset_1} P_{\wos}(\vt_1(\wos, \bA) - \vt_1(\wos, \bR))\cdot \frac{e(\ag)}{P_\wos}\\
        \le & \sum_{\wos \in \Lowset \cup \Highset_1} P_{\wos}(\vt_1(\Low, \bA) - \vt_1(\Low, \bR))(\lp_{\wos}^{\bA}(\stgp') - \lp_{\wos}^{\bA}(\stgp^*) + \frac{e(\ag)}{P_\wos})\\
        &- \sum_{\wos \in \Lowset \cap \Highset_1} (\vt_1(\wos, \bA) - \vt_1(\wos, \bR))\cdot e(\ag)\\
        = & \sum_{\wos \in \Lowset \cap \Highset_1} P_{\wos}(\vt_1(\Low, \bA) - \vt_1(\Low, \bR))(\lp_{\wos}^{\bA}(\stgp') - \lp_{\wos}^{\bA}(\stgp^*))\\
        &+\sum_{\wos \in \Lowset \cap \Highset_1}(((\vt_1(\Low, \bA) - \vt_1(\Low, \bR)) - (\vt_1(\wos, \bA) - \vt_1(\wos, \bR)))\cdot e(\ag)\\
        \le & \sum_{\wos \in \Lowset \cap \Highset_1} P_{\wos}(\vt_1(\Low, \bA) - \vt_1(\Low, \bR))(\lp_{\wos}^{\bA}(\stgp') - \lp_{\wos}^{\bA}(\stgp^*)) + \sum_{\wos \in \Lowset \cap \Highset_1} B\cdot e(\ag). 
    \end{align*}
    \item For the second term $-\sum_{\wos \in \Lowset_2 \cap \Highset} P_{\wos}(\vt_1(\wos, \bA) - \vt_1(\wos, \bR))(\lp_{\wos}^{\bA}(\stgp^*) - \lp_{\wos}^{\bA}(\stgp'))$, we use similar technique to replace $\wos$ into $\Low$:
    \begin{align*}
        &-\sum_{\wos \in \Lowset_2 \cap \Highset} P_{\wos}(\vt_1(\wos, \bA) - \vt_1(\wos, \bR))(\lp_{\wos}^{\bA}(\stgp^*) - \lp_{\wos}^{\bA}(\stgp'))\\
        = & -\sum_{\wos \in \Lowset_2 \cap \Highset} P_{\wos}(\vt_1(\wos, \bA) - \vt_1(\wos, \bR))(\lp_{\wos}^{\bA}(\stgp^*) - \lp_{\wos}^{\bA}(\stgp') + \frac{e(\ag)}{P_\wos})\\
        &+ \sum_{\wos \in \Lowset_2 \cap \Highset} (\vt_1(\wos, \bA) - \vt_1(\wos, \bR))\cdot e(\ag)\\
        \le  & - \sum_{\wos \in \Lowset_2 \cap \Highset} P_{\wos}(\vt_1(\Low, \bA) - \vt_1(\Low, \bR))(\lp_{\wos}^{\bA}(\stgp^*) - \lp_{\wos}^{\bA}(\stgp') + \frac{e(\ag)}{P_\wos})\\
        &+ \sum_{\wos \in \Lowset_2 \cap \Highset} (\vt_1(\wos, \bA) - \vt_1(\wos, \bR))\cdot e(\ag)\\
        = & - \sum_{\wos \in \Lowset_2 \cap \Highset} P_{\wos}(\vt_1(\Low, \bA) - \vt_1(\Low, \bR))(\lp_{\wos}^{\bA}(\stgp^*) - \lp_{\wos}^{\bA}(\stgp'))\\
        & + \sum_{\wos \in \Lowset_2 \cap \Highset}(((\vt_1(\wos, \bA) - \vt_1(\wos, \bR)) - (\vt_1(\Low, \bA) - \vt_1(\Low, \bR)))\cdot e(\ag)\\
        \le  & - \sum_{\wos \in \Lowset_2 \cap \Highset} P_{\wos}(\vt_1(\Low, \bA) - \vt_1(\Low, \bR))(\lp_{\wos}^{\bA}(\stgp^*) - \lp_{\wos}^{\bA}(\stgp')) + \sum_{\wos \in \Lowset_2 \cap \Highset}B\cdot e(\ag). 
    \end{align*}
    Note that the first inequality comes from fact 3 that $0< \vt_1(\Low, \bA) - \vt_1(\Low, \bR) \le \vt_1(\wos, \bA) - \vt_1(\wos, \bR)$ for all $\wos\in \Lowset_2\cap\Highset$. 
\end{enumerate}
Now we put everything together and get 
\begin{align*}
     \Delta < & \ut_1(\stgp') - \ut_1(\stgp^*)\\
      \le  & \sum_{\wos \in \Lowset \cap \Highset_1} P_{\wos}(\vt_1(\wos, \bA) - \vt_1(\wos, \bR))(\lp_{\wos}^{\bA}(\stgp') - \lp_{\wos}^{\bA}(\stgp^*))\notag\\
    & - \sum_{\wos \in \Lowset_2 \cap \Highset} P_{\wos}(\vt_1(\wos, \bA) - \vt_1(\wos, \bR))(\lp_{\wos}^{\bA}(\stgp^*) - \lp_{\wos}^{\bA}(\stgp'))+
    \sum_{\wos \in \Lowset_1 \cup \Highset_2}  B\cdot e(\ag)\\
    \le & \sum_{\wos \in \Lowset \cap \Highset_1} P_{\wos}(\vt_1(\Low, \bA) - \vt_1(\Low, \bR))(\lp_{\wos}^{\bA}(\stgp') - \lp_{\wos}^{\bA}(\stgp^*))\\
    & -  \sum_{\wos \in \Lowset_2 \cap \Highset} P_{\wos}(\vt_1(\Low, \bA) - \vt_1(\Low, \bR))(\lp_{\wos}^{\bA}(\stgp^*) - \lp_{\wos}^{\bA}(\stgp'))\\
    & + \sum_{\wos \in \Lowset_1 \cup \Highset_2} B\cdot e(\ag) + \sum_{\wos \in \Lowset \cup \Highset_1} B\cdot e(\ag) + \sum_{\wos \in \Lowset_2 \cup \Highset}B\cdot e(\ag)\\
    = & \sum_{\wos \in \Lowset \cap \Highset_1} P_{\wos}(\vt_1(\Low, \bA) - \vt_1(\Low, \bR))(\lp_{\wos}^{\bA}(\stgp') - \lp_{\wos}^{\bA}(\stgp^*))\\
    & -  \sum_{\wos \in \Lowset_2 \cap \Highset} P_{\wos}(\vt_1(\Low, \bA) - \vt_1(\Low, \bR))(\lp_{\wos}^{\bA}(\stgp^*) - \lp_{\wos}^{\bA}(\stgp'))+\Wos\cdot B\cdot e(\ag). 
\end{align*}
Note that $B\ge \vt_1(\Low, \bA) - \vt_1(\Low, \bR) > 0$. Therefore, we have 
\begin{align*}
    & \sum_{\wos \in \Lowset \cap \Highset_1} P_{\wos}(\lp_{\wos}^{\bA}(\stgp') - \lp_{\wos}^{\bA}(\stgp^*))
    -  \sum_{\wos \in \Lowset_2 \cap \Highset} P_{\wos}(\lp_{\wos}^{\bA}(\stgp^*) - \lp_{\wos}^{\bA}(\stgp'))\\
    > & \frac{\Delta - \Wos\cdot B\cdot e(\ag)}{B} \\
    =&  \frac{ \Wos\cdot B(B+1)\cdot e(\ag) - \Wos\cdot B\cdot e(\ag)}{B}\\
    = & \Wos\cdot B\cdot e(\ag)
\end{align*}

Now we consider agent 2's side. We show that $\ut_2(\stgp') - \ut_2(\stgp^*) < 0$. Most calculations will be the same as those of agent 1. Note that one main difference is for agent 2, $\vt_2(\wos, \bA) - \vt_2(\wos, \bR) < 0$ for all $\wos \in \Wosset\setminus(\Lowset_1 \cup \Highset_2)$.

\begin{align*}
    &\ut_2(\stgp') - \ut_2(\stgp^*)\notag\\
    = & \sum_{\wos=1}^{\Wos} P_{\wos}(\vt_2(\wos, \bA) - \vt_2(\wos, \bR))(\lp_{\wos}^{\bA}(\stgp') - \lp_{\wos}^{\bA}(\stgp^*))\\
    \le &\sum_{\wos \in \Wosset\setminus(\Lowset_1 \cup \Highset_2)} P_{\wos}(\vt_2(\wos, \bA) - \vt_2(\wos, \bR))(\lp_{\wos}^{\bA}(\stgp') - \lp_{\wos}^{\bA}(\stgp^*)) + \sum_{\wos \in \Lowset_1 \cup \Highset_2} B\cdot e(\ag)\\
    = & \sum_{\wos \in \Lowset \cap \Highset_1} P_{\wos}(\vt_2(\wos, \bA) - \vt_2(\wos, \bR))(\lp_{\wos}^{\bA}(\stgp') - \lp_{\wos}^{\bA}(\stgp^*))\notag\\
    & - \sum_{\wos \in \Lowset_2 \cap \Highset} P_{\wos}(\vt_2(\wos, \bA) - \vt_2(\wos, \bR))(\lp_{\wos}^{\bA}(\stgp^*) - \lp_{\wos}^{\bA}(\stgp'))+
    \sum_{\wos \in \Lowset_1 \cup \Highset_2} B\cdot e(\ag)\notag\\
    \le & \sum_{\wos \in \Lowset \cap \Highset_1} P_{\wos}(\vt_2(\Low, \bA) - \vt_2(\Low, \bR))(\lp_{\wos}^{\bA}(\stgp') - \lp_{\wos}^{\bA}(\stgp^*))\\
    & -  \sum_{\wos \in \Lowset_2 \cap \Highset} P_{\wos}(\vt_2(\Low, \bA) - \vt_2(\Low, \bR))(\lp_{\wos}^{\bA}(\stgp^*) - \lp_{\wos}^{\bA}(\stgp'))+\Wos\cdot B\cdot e(\ag)\\
    < & (\vt_2(\Low, \bA) - \vt_2(\Low, \bR)) \cdot \Wos\cdot B\cdot e(\ag)+\Wos\cdot B\cdot e(\ag)\\
    \le& -1\cdot \Wos\cdot B\cdot e(\ag) +\Wos\cdot B\cdot e(\ag)\\
    =&  0. 
\end{align*}
Therefore, we show that $\ut_2(\stgp') - \ut_2(\stgp') < 0$, which implies the claim. 

\begin{remark}
Note that in the binary setting, this part of proof can be largely simplified. Note that $\Lowset_1 = \Highset_2 =\emptyset$, $\Lowset\cup\Highset_1 = \{L\}$, and $\Highset\cup \Lowset_2 = \{H\}$.
\end{remark}

Now we come to the final step of Case 2. We'll show that a deviating group $D$ that contains only friendly or only unfriendly agents cannot gain utility more than $\varepsilon$.

\begin{claim}
If $D$ contains only friendly agents (or only unfriendly agents), then for every $\sag\in D$, $\ut_{\sag}(\stgp') - \ut_{\sag}(\stgp^*) < \varepsilon$. 
\end{claim}

With the loss of generality, suppose $D$ contains only friendly agents. The calculation for $D$ contains only unfriendly agents will be the same.  Note that in $\stgp^*$ every friendly agent always votes for $\bA$. Therefore, when friendly agents in $D$ deviate, the probability of $\bA$ being the winner will decrease. Formally, for any $\wos=1,2,\cdots, \Wos$, we have $\lp_{\wos}^{\bA} (\stgp') \le \lp_{\wos}^{\bA} (\stgp^*)$. Now we consider the difference in the expected utility of an arbitrary friendly agent in $D$: 
\begin{align*}
    \ut_{\sag}(\stgp') - \ut_{\sag}(\stgp^*)
    =& \sum_{\wos\in \Lowset_{\sag}} P_{\wos}(\vt_{\sag}(\wos, \bR) - \vt_{\sag}(\wos, \bA))(\lp_{\wos}^{\bA}(\stgp^*) - \lp_{\wos}^{\bA}(\stgp')) \\
    & + \sum_{\wos\in \Highset_{\sag}} P_{\wos}(\vt_{\sag}(\wos, \bA) - \vt_{\sag}(\wos, \bR))(\lp_{\wos}^{\bA}(\stgp') - \lp_{\wos}^{\bA}(\stgp^*)). 
\end{align*}
Note that for all $\wos\in\Highset_{\sag}$, we have $\vt_{\sag}(\wos, \bA) - \vt_{\sag}(\wos, \bR) > 0$. Therefore, the second sum is non-positive. At the same time, since $\Lowset_{\sag} \subsetneq \Lowset$, we have $\lp_{\wos}^{\bA}(\stgp^*) - \lp_{\wos}^{\bA}(\stgp') \le \lp_{\wos}^{\bA}(\stgp^*)\le \frac{e(\ag)}{P_{\wos}}$ and $\vt_{\sag}(\wos, \bR) - \vt_{\sag}(\wos, \bA) > 0$. Therefore,  
\begin{align*}
    \ut_{\sag}(\stgp') - \ut_{\sag}(\stgp^*)
    \le & \sum_{\wos\in \Lowset_{\sag}} P_{\wos}\cdot B\cdot \frac{e(\ag)}{P_{\wos}}< \Wos\cdot B\cdot e(\ag) < \varepsilon. 
\end{align*}
Therefore, we show that $\stgp^*$ is an $\varepsilon$-strong BNE where $\varepsilon = \Wos B((\Wos-1)B+1)\cdot e(\ag)$. 
\end{proof}

\subsection{Theorem~\ref{thm:nb_deviate} (Thorem~\ref{thm:deviate} for binary setting)}
\label{apx:nb_deviate}
Then we show that there always exists a sequence of regular strategy profiles with high fidelity. 

\begin{thmnb}{thm:nb_deviate}
For an arbitrary sequence of instances, there exists a series of regular strategy profiles $\{\stgp_{\ag}'\}_{\ag=1}^{\infty}$ and constants $\ag_0>0$, $\phi > 0$ such that for all $\ag > \ag_0$, $\acc(\stgp_{\ag}') \ge 1 - 2\exp(-2\phi^2\ag)$. 
\end{thmnb}

\begin{proof}
Let $B(n, p)$ denote a random variable of binomial distribution with $n$ experiments and probability $p$. Given $\ag$ and a regular strategy profile $\stgp_{\ag}$ where all contingent agents play strategy $\stg = (\bp_1, \bp_2, \cdots, \bp_\sig)$, we can write up $\lp_{\wos}^{\bA}$ and $\lp_{\wos}^{\bR}$ for each $\wos$. 

\begin{align*}
    \lp_{\wos}^{\bA}(\stgp_{\ag}) = & \Pr[\text{\#number of votes for }\bA \ge \Thd \ag\mid \Wosrv = \wos]\\
    = & \Pr[\tf + B(\tc, \sum_{\sig =1}^{\Sig} P_{\sig \wos}\cdot \bp_{\sig}) \ge \Thd \ag]\\
    =& \Pr[B(\tc, \sum_{\sig =1}^{\Sig} P_{\sig \wos}\cdot \bp_{\sig}) - \tc \sum_{\sig =1}^{\Sig} P_{\sig \wos}\cdot \bp_{\sig} \ge - \hthd^\ag_{\wos\bA} (\stg)\cdot \ag], 
\end{align*}
where we define
\begin{align*}
    \hthd^\ag_{\wos\bA} (\stg) =& \frac{\tf}{\ag} + \frac{\tc}{\ag} \sum_{\sig =1}^{\Sig} P_{\sig \wos}\cdot \bp_{\sig} -\Thd\\
    = & \hagf + \hagc \sum_{\sig =1}^{\Sig} P_{\sig \wos}\cdot \bp_{\sig} -\Thd
\end{align*}

Similarly, 
\begin{align*}
    \lp_{\wos}^{\bR}(\stgp_{\ag}) = \Pr[B(\tc, \sum_{\sig =1}^{\Sig} P_{\sig \wos}\cdot (1-\bp_{\sig})) - \tc \sum_{\sig =1}^{\Sig} P_{\sig \wos}\cdot (1-\bp_{\sig}) \ge - \hthd^\ag_{\wos\bR} (\stg)\cdot \ag], 
\end{align*}
where 
\begin{align*}
    \hthd^\ag_{\wos\bR} (\stg)=& \frac{\tu}{\ag} + \frac{\tc}{\ag} \sum_{\sig =1}^{\Sig} P_{\sig \wos}\cdot (1-\bp_{\sig}) -(1-\Thd)\\
    = &\hagu + \hagc\sum_{\sig =1}^{\Sig} P_{\sig \wos}\cdot (1-\bp_{\sig}) -(1-\Thd)
\end{align*}
If we can find a $\stg'$, and constant $\ag_0 > 0$, $\phi > 0$ s.t. for every $\ag > \ag_0$, for every $\wos\in \Highset$ we have $ \hthd^\ag_{\wos\bA} (\stg') > \phi$ and for every $\wos\in\Lowset$ we have $ \hthd^\ag_{\wos\bR} (\stg') > \phi$, then we can directly apply the Hoeffding Inequality and show that for every $\ag > \ag_0$, $\stgp_{\ag}'$ has high fidelity.  

Now we start constructing strategy $\stg'$. First we define the "approximated version" of $\hthd^\ag_{\wos\bA}$ and $\hthd^\ag_{\wos\bR}$ which are independent from $\ag$. Given a strategy $\stg$, let 
\begin{align}
    \thd_{\wos\bA} (\stg) 
    = & \agf + \agc \sum_{\sig =1}^{\Sig} P_{\sig \wos}\cdot \bp_{\sig} -\Thd\\
    \thd_{\wos\bR} (\stg)
    = &\agu + \agc\sum_{\sig =1}^{\Sig} P_{\sig \wos}\cdot (1-\bp_{\sig}) -(1-\Thd). 
\end{align}
The rest of the proof will go as follows:
\begin{enumerate}
    \item First, we construct a strategy $\stg'$ s.t. for every $\wos\in \Highset$ we have $ \thd_{\wos\bA} (\stg') > 0$ and for every $\wos\in\Lowset$ we have $ \thd_{\wos\bR} (\stg') > 0$.
    \item Then we find out $\ag_0$ and show that for every $\ag >\ag_0$, $\hthd^\ag_{\wos\bA} (\stg')$ will not deviate from $ \thd_{\wos\bA} (\stg')$ too much, thus will larger than some constant $\phi$. Similarly we show that  for every $\ag >\ag_0$, $\hthd^\ag_{\wos\bA} (\stg') > \phi$. 
    \item Finally, we show that for all $\ag > \ag_0$, $\acc(\stgp_{\ag}') \ge 1 - 2\exp(-2\phi^2 T)$ by applying the Hoeffding inequality. 
\end{enumerate}

\noindent{\bf Part 1. Constructing $\stg'$.} The construction of $\stg$ has three steps. We start from a simple strategy and then modify it to the target step by step. 

\noindent{\bf Step 1.} Consider a simple strategy $\stg_0 = (\bp^*, \bp^*, \cdots, \bp^*)$. That is, the agent always votes for $\bA$ with probability $\bp^*$ and votes for $\bR$ with probability $1 - \bp^*$ no matter what signal she receives. $\bp^*$ satisfies $\agf + \agc\cdot \bp^* - \Thd =0$. Since $\agf < \Thd$ and $\agu < 1- \Thd$, we have $\bp^* \in (0,1)$. Moreover, we have $\thd_{\wos\bA}(\stg_0) = 0$ and $\thd_{\wos\bR}(\stg_0) = 0$ for all $\wos\in\Wosset$. 

\noindent{\bf Step 2.} Let $\sig^* = \lfloor \frac {\Sig}{2} \rfloor$. $\sig^*$ will serve as a threshold, as we regard signals not larger than $\sig^*$ as "low signals" and those exceeding $\sig^*$ as "high signals". Specifically, for a world state $\wos$, define $P_{l\wos} = \sum_{\sig=1}^{\sig^*} P_{\sig \wos}$, and $P_{h\wos} = \sum_{\sig = \sig^* +1}^{\Sig} P_{\sig \wos}$. According to the stochastic dominance assumption, we have $0\le P_{l\wos_1} < P_{l\wos_2}$ and $P_{h\wos_1} > P_{h\wos_2}\ge 0$ for all $\wos_1 > \wos_2$. Let $\bp_l = \bp^* -\delta_l$, and $\bp_h = \bp^* + \delta_h$, where $\delta_l > 0$ and $\delta_h > 0$ are constant satisfying $\delta_h  = \delta_l \cdot \frac{P_{lH}}{P_{hH}}$. We can carefully select and fix $\delta_h$ and $\delta_l$ to make $\bp_h$ and $\bp_l$ inside $(0, 1)$. Then we construct strategy $\stg_1$: if the agent receives a low signal $\sig \le \sig^*$, she votes for $\bA$ with probability $\bp_l$ and $\bR$ with probability $1 - \bp_l$; if she receives a high signal $m > \sig^*$, she votes for $\bA$ with probability $\bp_h$ and $\bR$ with probability $1 - \bp_h$. Then we show that 
\begin{itemize}
    \item for all $\wos\in \Highset$, $\thd_{\wos\bA}(\stg_1) \ge 0$, and 
    \item for all $\wos\in \Lowset$, $\thd_{\wos\bR}(\stg_1) > 0$. 
\end{itemize}
For the $\Highset$ side, we'll first show that $\thd_{\High\bA}(\stg_1) = 0$. Then we show that for all $\wos\in \Highset$, $\thd_{\wos\bA}(\stg_1) \ge \thd_{\High\bA}(\stg_1)$. For $\wos = \High$, we have 
\begin{align*}
    \thd_{\High\bA}(\stg_1) = & \agf + \agc (P_{l\High}\cdot \bp_l + P_{h\High} \cdot \bp_h)-\Thd\\
    = & \agf + \agc (P_{l\High}\cdot (\bp^* -\delta_l) + P_{h\High} \cdot (\bp^* + \delta_h))-\Thd\\
    = & \agc\cdot (-P_{l\High}\cdot \delta_l + P_{h\High}\cdot \delta_l \cdot \frac{P_{l\High}}{P_{h\High}})\\
    = & 0.
\end{align*}
For other $\wos\in\Highset$, because $\bp_l < \bp_h$, $P_{l\wos} \le P_{l\High}$,and $P_{h\wos} \ge P_{h\High}$, we have $ \thd_{\wos\bA}(\stg_1) \ge  \thd_{\High\bA}(\stg_1) \ge 0$. 

For the $\Lowset$ side, similarly, we'll first show that $\thd_{\Low\bR}(\stg_1) > 0$. Then we show that for all $\wos\in \Lowset$, $\thd_{\wos\bR}(\stg_1) > \thd_{\Low\bR}(\stg_1)$. For $\wos = \Low$ we have 
\begin{align*}
    \thd_{\Low\bR}(\stg_1) = & \agu + \agc (P_{l\Low}\cdot (1 -\bp_l) + P_{h\Low} \cdot (1 - \bp_h))-(1-\Thd)\\
    = & \agu + \agc (P_{l\Low}\cdot (1 - \bp^* +\delta_l) + P_{h\Low} \cdot ( 1- \bp^* - \delta_h))-(1-\Thd)\\
    = & \agc\cdot (P_{l\Low}\cdot \delta_l - P_{h\Low}\cdot \delta_l \cdot \frac{P_{l\High}}{P_{h\High}})\\
    = & \frac{\agc\cdot \delta_l}{P_{h\High}}(P_{l\Low}\cdot P_{h\High} - P_{h\Low}\cdot P_{l\High})\\
    > & 0. 
\end{align*}
The final inequality comes that  $P_{l\High} < P_{l\Low}$,and $P_{h\High} > P_{h\Low}$ since $\Low < \High$. For other $\wos\in\Lowset$, because $(1 - \bp_l)> (1- \bp_h)$, $P_{l\wos} \ge P_{l\Low}$,and $P_{h\wos} \le P_{h\Low}$, we have $ \thd_{\wos\bA}(\stg_1) \ge  \thd_{\Low\bA}(\stg_1) > 0$. 

\begin{remark}
In the binary setting, we just let $\stg_1 = (\bp_l, \bp_h)$. And we will find that $\thd_{H\bA}(\stg_1) = 0$ and $\thd_{L\bR}(\stg_1) > 0$. 
\end{remark}

\noindent{\bf Step 3}. In step 3, we modify $\bp_h$ to make both sides strictly positive. Let $\bp_l' = \bp_l$, and $\bp_h' = \bp_h + \delta_h'$, where $\delta_h' > 0$ is a constant satisfying $\agc\cdot \delta_h' \le \frac{ \thd_{\Low\bR}(\stg_1)}{2}$ and $\bp_h + \delta_h' \le 1$. Therefore, $0\le \bp_l'<\bp_h'\le 1$.  Now we are ready to construct our final strategy $\stg'$: if the agent receives a low signal $\sig \le \sig^*$, she votes for $\bA$ with probability $\bp_m' = \bp_l'$ and $\bR$ with probability $1 - \bp_l'$; if she receives a high signal $m > \sig^*$, she votes for $\bA$ with probability $\bp_m' =\bp_h'$ and $\bR$ with probability $1 - \bp_h'$. Then we show that 
\begin{itemize}
    \item for all $\wos\in \Highset$, $\thd_{\wos\bA}(\stg') > 0$, and 
    \item for all $\wos\in \Lowset$, $\thd_{\wos\bR}(\stg') > 0$. 
\end{itemize}

For the $\Highset$ side, we still consider $\High$ first: 
\begin{align*}
    \thd_{\High\bA}(\stg') = & \agf + \agc (P_{l\High}\cdot \bp_l' + P_{h\High} \cdot \bp_h')-\Thd\\
    = & \agf + \agc (P_{l\High}\cdot \bp_l + P_{h\High} \cdot (\bp_h + \delta_h'))-\Thd\\
    = & \thd_{\High\bA}(\stg_1) + \agc\cdot P_{h\High}\cdot \delta_h'\\
    = & \agc\cdot P_{h\High}\cdot \delta_h'\\
    > & 0.
\end{align*}
For other $\wos\in\Highset$, because $\bp_l' < \bp_h'$, $P_{l\wos} \le P_{l\High}$,and $P_{h\wos} \ge P_{h\High}$, we have $ \thd_{\wos\bA}(\stg') \ge  \thd_{\High\bA}(\stg') >0$. 

For the $\Lowset$ side, we consider $\Low$ first: 
\begin{align*}
    \thd_{\Low\bR}(\stg') = & \agu + \agc (P_{l\Low}\cdot (1 -\bp_l') + P_{h\Low} \cdot (1 - \bp_h'))-(1-\Thd)\\
    & \agu + \agc (P_{l\Low}\cdot (1 -\bp_l) + P_{h\Low} \cdot (1 - \bp_h - \delta_h'))-(1-\Thd)\\
    = &\thd_{\Low\bR}(\stg') -\agc\cdot P_{h\Low}\cdot \delta_h'\\
    \ge& \frac{\thd_{\Low\bR}(\stg')}{2}\\
    > & 0. 
\end{align*}
 For other $\wos\in\Lowset$, because $(1 - \bp_l')> (1- \bp_h')$, $P_{l\wos} \ge P_{l\Low}$,and $P_{h\wos} \le P_{h\Low}$, we have $ \thd_{\wos\bA}(\stg') \ge  \thd_{\Low\bA}(\stg') > 0$. 
 \begin{remark}
 In the binary setting, similarly, let $\stg' = (\bp_l', \bp_h')$. We will find that $\thd_{H\bA}(\stg') > 0$ and $\thd_{L\bR}(\stg') > 0$
 \end{remark}
 
 \noindent{\bf Part 2 Determine $\ag_0$ and $\phi$.} In this part, we'll show that for all $\wos\in\Highset$, ($\wos\in\Lowset$, respectively), $\hthd^\ag_{\wos\bA} (\stg')$ ($\hthd^\ag_{\wos\bR} (\stg')$, respectively) will not deviate from $ \thd_{\wos\bA} (\stg')$ ($\thd_{\wos\bR} (\stg')$, respectively) too much. Since  $\thd_{\wos\bA} (\stg')$ and $\thd_{\wos\bR} (\stg')$ do not depend on $\ag$, we can determine constant $\ag_0$ and $\phi$, s.t. for all $\ag > \ag_0$, $\hthd^\ag_{\wos\bA} (\stg') > \phi$, and $\hthd^\ag_{\wos\bR} (\stg')> \phi$. 
We start from the $\Highset$ side. Recall that 
\begin{align*}
    \hthd^\ag_{\wos\bA} (\stg) =& \frac{\tf}{\ag} + \frac{\tc}{\ag} \sum_{\sig =1}^{\Sig} P_{\sig \wos}\cdot \bp_{\sig} -\Thd. 
\end{align*}
And recall that  $\tf = \ag - \lfloor (1-\agf) \cdot \ag\rfloor$, $\tu = \lfloor \agu \cdot \ag\rfloor$, and $\tc = \lfloor (1-\agf) \cdot \ag\rfloor - \lfloor \agu \cdot \ag\rfloor$. Therefore, for all $\wos\in\Highset$, we have 
\begin{align*}
    \hthd^\ag_{\wos\bA} (\stg') =&
    1 - \frac{\lfloor (1-\agf) \cdot \ag\rfloor}{\ag} + \frac{\lfloor (1-\agf) \cdot \ag\rfloor - \lfloor \agu \cdot \ag\rfloor}{\ag} \sum_{\sig =1}^{\Sig} P_{\sig \wos}\cdot \bp_{\sig}' -\Thd\\
    \ge & 1 - (1-\agf)+ ((1-\agf) -\frac{1}{\ag} - \agu )\sum_{\sig =1}^{\Sig} P_{\sig \wos}\cdot \bp_{\sig}' -\Thd\\
    = & \agf + \agc\cdot \sum_{\sig =1}^{\Sig} P_{\sig \wos}\cdot \bp_{\sig}' -\Thd - \frac{1}{\ag}  \sum_{\sig =1}^{\Sig} P_{\sig \wos}\cdot \bp_{\sig}'\\
    \ge& \thd_{\wos\bA}(\stg') - \frac{1}{\ag}\\
    \ge & \thd_{\High\bA}(\stg') - \frac{1}{\ag}
\end{align*}
Since $\thd_{\High\bA}(\stg')>0$ is independent from $\ag$, there exists a $\ag_\Highset> 0$, s.t. for all $\ag > \ag_{\Highset}$ and all $\wos\in\Highset$, $\hthd^\ag_{\wos\bA} (\stg') \ge \frac{\thd_{\High\bA}(\stg')}{2}$. 

For the $\Lowset$ side, similarly for all $\wos\in\Lowset$,  we have 
\begin{align*}
    \hthd^\ag_{\wos\bR} (\stg')=& \frac{\tu}{\ag} + \frac{\tc}{\ag} \sum_{\sig =1}^{\Sig} P_{\sig \wos}\cdot (1-\bp_{\sig}') -(1-\Thd)\\
    = &\frac{\lfloor \agu \cdot \ag\rfloor}{\ag} + \frac{\lfloor (1-\agf) \cdot \ag\rfloor - \lfloor \agu \cdot \ag\rfloor}{\ag}\sum_{\sig =1}^{\Sig} P_{\sig \wos}\cdot (1-\bp_{\sig}') -(1-\Thd)\\
    \ge & \agu - \frac{1}{\ag} + ((1-\agf) - \frac{1}{\ag} - \agu )\sum_{\sig =1}^{\Sig} P_{\sig \wos}\cdot (1-\bp_{\sig}') -(1-\Thd)\\
    = & \agu + \agc\cdot \sum_{\sig =1}^{\Sig} P_{\sig \wos}\cdot (1-\bp_{\sig}') -(1-\Thd) -\frac{1}{\ag} (1 +\sum_{\sig =1}^{\Sig} P_{\sig \wos}\cdot (1-\bp_{\sig}'))\\
    \ge & \thd_{\wos\bR}(\stg') - \frac2\ag\\
    \ge &  \thd_{\Low\bR}(\stg') - \frac2\ag.
\end{align*}
Since $\thd_{\Low\bR}(\stg')>0$ is independent from $\ag$, there exists a $\ag_\Lowset> 0$, s.t. for all $\ag > \ag_{\Lowset}$ and all $\wos\in\Lowset$, $\hthd^\ag_{\wos\bR} (\stg') \ge \frac{ \thd_{\Low\bR}(\stg')}{2}$. 

\begin{remark}
This part is different in the binary setting due to the different rounding method. However, we can still show that $\hthd^\ag_{H\bA} (\stg')  \ge \thd_{\High\bA}(\stg') - \frac2\ag$ and $\hthd^\ag_{L\bR} (\stg')\ge \thd_{\Low\bR}(\stg') - \frac2\ag$ by the same reasoning. 
\end{remark}

Therefore, putting $\Highset$ and $\Lowset$ sides together, let $\phi =\min(\frac{\thd_{\High\bA}(\stg')}{2}, \frac{ \thd_{\Low\bR}(\stg')}{2})$, and $\ag_0 = \max(\ag_\Highset, \ag_\Lowset)$, we have for all $\ag > \ag_0$, for all $\wos\in\Highset$, $\hthd^\ag_{\wos\bA} (\stg') \ge \phi$, and  all $\wos\in\Lowset$, $\hthd^\ag_{\wos\bR} (\stg') \ge\phi$. 

\noindent{\bf Part 3. Bound the fidelity.} Now we come back to the fidelity and bound it by the Hoeffding Inequality. For all $\ag > \ag_0$, for the $\Highset$ side we have 
\begin{align*}
    \lp_{\wos}^{\bA}(\stgp_{\ag}') 
    =& \Pr[B(\tc, \sum_{\sig =1}^{\Sig} P_{\sig \wos}\cdot \bp_{\sig}') - \tc \sum_{\sig =1}^{\Sig} P_{\sig \wos}\cdot \bp_{\sig}' \ge - \hthd^\ag_{\wos\bA} (\stg)\cdot \ag]\\
    \ge &\Pr[B(\tc, \sum_{\sig =1}^{\Sig} P_{\sig \wos}\cdot \bp_{\sig}') - \tc \sum_{\sig =1}^{\Sig} P_{\sig \wos}\cdot \bp_{\sig}' \ge - \phi\cdot \ag]\\
    \ge& 1 - 2\exp(-2\phi^2\ag). 
\end{align*}
Similarly, for the $\Lowset$ side we have 
\begin{align*}
    \lp_{\wos}^{\bR}(\stgp_{\ag}') =& \Pr[B(\tc, \sum_{\sig =1}^{\Sig} P_{\sig \wos}\cdot (1-\bp_{\sig}')) - \tc \sum_{\sig =1}^{\Sig} P_{\sig \wos}\cdot (1-\bp_{\sig}') \ge - \hthd^\ag_{\wos\bR} (\stg)\cdot \ag]\\
    \ge &\Pr[B(\tc, \sum_{\sig =1}^{\Sig} P_{\sig \wos}\cdot (1-\bp_{\sig}')) - \tc \sum_{\sig =1}^{\Sig} P_{\sig \wos}\cdot (1-\bp_{\sig}') \ge - \phi\cdot \ag]\\
    \ge & 1 - 2\exp(-2\phi^2\ag). 
\end{align*}
Therefore for all $\ag > \ag_0$, 
\begin{align*}
    \acc(\stgp_{\ag}') = & \sum_{\wos\in \Lowset}P_{\wos}\cdot \lp_{\wos}^{\bR}(\stgp_{\ag}') + \sum_{\wos\in \Highset}P_{\wos} \cdot \lp_{\wos}^{\bA}(\stgp_{\ag}')\\
    \ge &  \sum_{\wos\in \Lowset}P_{\wos}( 1 - 2\exp(-2\phi^2\ag)) + \sum_{\wos\in \Highset}P_{\wos}(1 - 2\exp(-2\phi^2\ag))\\
    = & 1 - 2\exp(-2\phi^2\ag). 
\end{align*}
\end{proof}

\subsection{Theorem~\ref{thm:nb_acc} (Theorem~\ref{thm:acc} for binary setting)} 
\label{apx:nb_acc}
Then we give the theorem of the equivalence between fidelity and strong equilibrium. 

\begin{thmnb}{thm:nb_acc}
 Given an arbitrary sequence of instances and an arbitrary sequence of regular strategy profiles $\{\stgp_{\ag}\}$ , let $\{\acc(\stgp_{\ag})\}_{\ag=1}^{\infty}$ be the sequence of fidelity of $\stgp$.
\begin{itemize}
    \item If $\lim_{\ag\in \infty} \acc(\stgp_{\ag}) = 1$, for every $\ag > \ag_\mu$, $\stgp_{\ag}$ is an $\varepsilon$-strong BNE with $\varepsilon = o(1)$. 
    \item If $\lim_{\ag\in \infty} \acc(\stgp_{\ag}) = 1$ does NOT hold, there exists infinitely many $\ag$ such that $\stgp_\ag$ is not an $\varepsilon$-strong BNE with constant $\varepsilon$
\end{itemize}
\end{thmnb}

\begin{proof}
Let $e(\ag) = (1 - \acc(\stgp_{\ag}))$, and we can directly apply Lemma~\ref{lem:nb_nash}, and get that for every $\ag$, $\stgp_{\ag}$ is an $\varepsilon$-strong BNE where $\varepsilon = \Wos B((\Wos-1)B+1)\cdot (1-\acc(\stgp_{\ag})) $.

\noindent{\bf Case 1: $\lim_{\ag\in \infty} \acc(\stgp_{\ag}) = 1$}. For this case,  Since $\acc(\stgp_{\ag})$ converges to $1$, $\varepsilon$, which is positive proportion to $1 - \acc(\stgp_{\ag})$, will converge to 0 as $\ag\to \infty$. 

\noindent{\bf Case 2:} $\lim_{\ag\in \infty} \acc(\stgp_{\ag}) = 1$ does not hold. Then there exist a constant $\delta > 0$ and a infinite set $\negT\subseteq R$, s.t. for all $\ag\in \negT$, $\acc(\stgp_{\ag}) \le 1 - \delta$. From Theorem~\ref{thm:nb_deviate} we know that we can construct a series of symmetric deviating strategies $\{\stgp_{\ag}'\}$, and find constants $\ag_0 > 0, \phi > 0$, s.t. for all $\ag > \ag_0$, $\acc(\stgp_{\ag}') \le 1 - 2\exp(-2\phi^2\ag)$. Note that in both $\{\stgp_{\ag}\}$ and $\{\stgp_{\ag}'\}$ every friendly (unfriendly, respectively) agent always votes for $\bA$ ($\bR$, respectively). Therefore, the deviating group contains only contingent agents, and we only need to consider the expected utility of contingent agents. Given an arbitrary  $(\ag > \ag_0)\wedge (\ag\in\negT)$, recall the difference of expected utility for an agent between $\stgp_{\ag}$ and $\stgp_{\ag}'$: 
\begin{align*}
     \ut_{\sag}(\stgp_{\ag}') - \ut_{\sag}(\stgp_{\ag})
    =  & \sum_{\wos\in \Lowset} P_{\wos}(\vt_{\sag}(\wos, \bR) - \vt_{\sag}(\wos, \bA))(\lp_{\wos}^{\bR}(\stgp_{\ag}') - \lp_{\wos}^{\bR}(\stgp_{\ag})) \\
    & + \sum_{\wos\in \Highset} P_{\wos}(\vt_{\sag}(\wos, \bA) - \vt_{\sag}(\wos, \bR))(\lp_{\wos}^{\bA}(\stgp_{\ag}') - \lp_{\wos}^{\bA}(\stgp_{\ag})).
\end{align*}
And recall the definition of fidelity: 
\begin{equation*}
    \acc(\stgp) =  \sum_{\wos\in \Lowset}P_{\wos}\cdot \lp_{\wos}^{\bR}(\stgp) + \sum_{\wos\in \Highset}P_{\wos} \cdot \lp_{\wos}^{\bA}(\stgp).
\end{equation*}
For $\stgp_{\ag}'$ we have $\acc(\stgp_{\ag}') \ge 1 - 2\exp(-2\phi^2\ag)$. Therefore, for every $\wos\in\Highset$ we have $\lp_{\wos}^{\bA}(\stgp_{\ag}') \ge 1 - \frac{ 2\exp(-2\phi^2\ag)}{P_{\wos} }$, and for every $\wos\in\Lowset$, $\lp_{\wos}^{\bR}(\stgp_{\ag}') \ge 1-  \frac{ 2\exp(-2\phi^2\ag)}{P_{\wos} }$. For $\stgp_{\ag}$, on the other hand, we have $\acc(\stgp_{\ag}) \le 1 - \delta$ for all $\ag\in\negT$. Therefore, there exists some $\wos^*\in\Wosset$, s.t.
\begin{enumerate}
    \item if $\wos^*\in \Highset$, $\lp_{\wos^*}^{\bA}(\stgp_{\ag}) \le 1-\delta$, or
    \item if $\wos^*\in\Lowset$, $\lp_{\wos^*}^{\bR}(\stgp_{\ag}) \le 1-\delta$. 
\end{enumerate}
W.l.o.g we assume $\wos^*\in \Highset$. The reasoning for $\wos^*\in \Lowset$ will be almost the same. Then in the difference of expected utility, we have 
\begin{align*}
     \ut_{\sag}(\stgp_{\ag}') - \ut_{\sag}(\stgp_{\ag})
    =  & P_{\wos^*}(\vt_{\sag}(\wos^*, \bA) - \vt_{\sag}(\wos^*, \bR))(\lp_{\wos^*}^{\bA}(\stgp_{\ag}') - \lp_{\wos^*}^{\bA}(\stgp_{\ag})) \\
    & + \sum_{\wos\in \Lowset} P_{\wos}(\vt_{\sag}(\wos, \bR) - \vt_{\sag}(\wos, \bA))(\lp_{\wos}^{\bR}(\stgp_{\ag}') - \lp_{\wos}^{\bR}(\stgp_{\ag}))\\
    & + \sum_{\wos\in \Highset\setminus\{n^*\}} P_{\wos}(\vt_{\sag}(\wos, \bA) - \vt_{\sag}(\wos, \bR))(\lp_{\wos}^{\bA}(\stgp_{\ag}') - \lp_{\wos}^{\bA}(\stgp_{\ag}))\\
    \ge & P_{\wos^*}\cdot 1 \cdot (1 - \frac{ 2\exp(-2\phi^2\ag)}{P_{\wos^*} } - (1-\delta))- \sum_{\wos\in\Wosset\setminus\{\wos^*\}} P_{\wos} \cdot B\cdot \frac{ 2\exp(-2\phi^2\ag)}{P_{\wos}}\\
    =& \delta P_{\wos^*} -((\Wos-1)B+1)\cdot 2\exp(-2\phi^2\ag). \\
    \ge&  \delta\cdot  \min_{\wos\in\Wosset}(P_{\wos}) -((\Wos-1)B+1)\cdot 2\exp(-2\phi^2\ag).
\end{align*}
Note that the first term $ \delta\cdot  \min_{\wos\in\Wosset}(P_{\wos})$ is a constant, while the second term $((\Wos-1)B+1)\cdot 2\exp(-2\phi^2\ag)$ converge to $0$ as $\ag\to\infty$. Therefore, there exists a $\ag_\delta\ge \ag_0$ s.t. for all $(\ag >\ag_{\delta}) \wedge( \ag\in\negT)$, $\ut_{\sag}(\stgp_{\ag}') - \ut_{\sag}(\stgp_{\ag}) > \frac{1}{2} \delta\cdot  \min_{\wos\in\Wosset}(P_{\wos})$. Therefore, for every such $\ag$, $\stgp_{\ag}$ is NOT an $\varepsilon$-strong BNE for all  $\varepsilon \le\frac{1}{2} \delta\cdot  \min_{\wos\in\Wosset}(P_{\wos}) $. 
\end{proof}

\subsection{Theorem~\ref{thm:nb_arbitrary} (Theorem~\ref{thm:arbitrary} for binary setting)}
\label{apx:nb_arbitrary}
Then we give the theorem based on the probability analysis of the fidelity and provide a criterion for judging whether a profile sequence is of high fidelity. 

Consider a strategy profile sequence $\{\stgp_{\ag}\}$. For every $\ag$, we define random variable $\xrv_{\sag}^{\ag}$ as "agent $\sag$ votes for $\bA$". That is, in the instance of $\ag$ agents, $\xrv_{\sag}^{\ag} = 1$ if agent $\sag$ votes for $\bA$, $\xrv_{\sag}^{\ag} = 0$ if $\sag$ votes for $\bR$. Then we can write $\lp_{\wos}^{\bA}(\stgp_{\ag})$ and $\lp_{\wos}^{\bR}(\stgp_{\ag})$:

\begin{align*}
    \lp_{\wos}^{\bA}(\stgp_{\ag}) =& \Pr[\sum_{\sag = 1}^{\ag}\xrv_{\sag}^{\ag} \ge \Thd\cdot \ag \mid \wos]\\
    =& \Pr[\sum_{\sag = 1}^{\ag}\xrv_{\sag}^{\ag} - \sum_{\sag = 1}^{\ag}E[\xrv_{\sag}^{\ag}\mid \wos] \ge \Thd\cdot \ag - \sum_{\sag = 1}^{\ag}E[\xrv_{\sag}^{\ag}\mid \wos] \mid \wos]\\
    \lp_{\wos}^{\bR}(\stgp_{\ag}) =& \Pr[\sum_{\sag = 1}^{\ag}(1-\xrv_{\sag}^{\ag}) > (1 - \Thd)\ag \mid\wos]\\
    =&\Pr[\sum_{\sag = 1}^{\ag}(1-\xrv_{\sag}^{\ag}) - \sum_{\sag = 1}^{\ag}E[1-\xrv_{\sag}^{\ag}\mid \wos] >(1 - \Thd)\ag  - \sum_{\sag = 1}^{\ag}E[1-\xrv_{\sag}^{\ag}\mid\wos] \mid \wos]
\end{align*}

Let the \exshare\ be 
\begin{align*}
    \hthd^\ag_{\wos\bA} =&  \frac{1}{\ag}\sum_{\sag = 1}^{\ag}E[\xrv_{\sag}^{\ag}\mid \wos] - \Thd.\\
    \hthd^\ag_{\wos\bR} =& \frac{1}{\ag}\sum_{\sag = 1}^{\ag}E[1-\xrv_{\sag}^{\ag}\mid \wos] - (1-\Thd). 
\end{align*}

Then, for every $\ag$, let $$\hthd^{\ag} = \min\left( \min_{\wos\in\Highset} (\hthd^\ag_{\wos\bA}), \min_{\wos\in\Lowset} (\hthd^\ag_{\wos\bR})\right).  $$

\begin{thm}
\label{thm:nb_arbitrary}
Given an arbitrary sequence of instances and arbitrary sequence of strategy profiles $\{\stgp_{\ag}\}_{\ag=1}^\infty$ let $\hthd^\ag$ is defined for every $\stgp^\ag$.
\begin{itemize}
    \item If $\liminf_{\ag\to\infty} \sqrt{\ag}\cdot \hthd^{\ag} = +\infty$, 
    $\lim_{\ag\to\infty} \acc(\stgp_\ag) = 1$. 
    \item If $\liminf_{\ag\to\infty} \sqrt{\ag}\cdot \hthd^{\ag} < 0$ (including $-\infty$), i.e., there exists a constant $\eta <0$ and an infinite set $\negT$, s.t. for every $\ag\in \negT$, $\sqrt{\ag}\cdot \hthd^{\ag} \le \eta$, then there exists a constant $\ag_{\eta} > 0$ such that for all $(\ag\in \negT) \wedge (\ag > \ag_{\eta})$, $\acc(\stgp_{\ag}) $ has a constant distance with 1. 
    \item If $\liminf_{\ag\to\infty} \sqrt{\ag}\cdot \hthd^{\ag} \ge 0$ (not including $+\infty$), i.e., there exists a constant $\eta \ge 0$ and an infinite set $\negT$, s.t. for every $\ag\in \negT$, $\sqrt{\ag}\cdot \hthd^{\ag} \le \eta$, and there exists a constant $\psi$ s.t. for every $\ag\in\negT$ and every $\wos\in \Wosset$, $Var(\sum_{\sag = 1}^{\ag} \xrv_{\sag}^{\ag} \mid \wos) \ge \psi\cdot \ag$, then there exists a constant $\ag_{\eta} > 0$ such that for all $(\ag\in \negT) \wedge (\ag > \ag_{\eta})$, $\acc(\stgp_{\ag}) $ has a constant distance with 1.
\end{itemize}

\end{thm}

\begin{proof}
\noindent {\bf Case 1: $\liminf_{\ag\to\infty} \sqrt{\ag}\cdot \hthd^{\ag} = +\infty$}. In this case we use the Hoeffding Inequality to give each $\lp_{\wos}^{\bA}(\stgp_{\ag})$ (or $\lp_{\wos}^{\bR}(\stgp_{\ag})$) a lower bound. For all $\wos\in\Highset$, according to the definition of $\hthd^{\ag}$:
\begin{align*}
    \lp_{\wos}^{\bA}(\stgp_{\ag})
    =& \Pr[\sum_{\sag = 1}^{\ag}\xrv_{\sag}^{\ag} - \sum_{\sag = 1}^{\ag}E[\xrv_{\sag}^{\ag}\mid \wos] \ge -\hthd^\ag_{\wos \bA}\cdot \ag \mid \wos]\\
    \ge & \Pr[\sum_{\sag = 1}^{\ag}\xrv_{\sag}^{\ag} - \sum_{\sag = 1}^{\ag}E[\xrv_{\sag}^{\ag}\mid \wos] \ge -\hthd^{\ag}\cdot \ag \mid \wos]. 
\end{align*}
Let $\ag_0$ satisfies that for all $\ag > \ag_0, \sqrt{\ag}\cdot \hthd^\ag > 0$. 
Then for all $\ag > \ag_0$, we can directly apply the Hoeffding Inequality:
\begin{align*}
    \lp_{\wos}^{\bA}(\stgp_{\ag})
    \ge & \Pr[\sum_{\sag = 1}^{\ag}\xrv_{\sag}^{\ag} - \sum_{\sag = 1}^{\ag}E[\xrv_{\sag}^{\ag}\mid \wos] \ge -\hthd^{\ag}\cdot \ag \mid \wos]\\
    \ge & 1 - 2\exp(-2(\hthd^{\ag})^2\ag). 
\end{align*}

Similarly, for all $\wos\in\Lowset$, we have 
\begin{align*}
    \lp_{\wos}^{\bR}(\stgp_{\ag})
    =& \Pr[\sum_{\sag = 1}^{\ag}(1-\xrv_{\sag}^{\ag}) - \sum_{\sag = 1}^{\ag}E[1-\xrv_{\sag}^{\ag}\mid \wos] \ge -\hthd^\ag_{\wos \bR} \cdot \ag\mid \wos]\\
    \ge & \Pr[\sum_{\sag = 1}^{\ag}(1-\xrv_{\sag}^{\ag}) - \sum_{\sag = 1}^{\ag}E[1-\xrv_{\sag}^{\ag}\mid \wos] \ge -\hthd^{\ag} \cdot \ag\mid \wos]\\
    \ge &1 - 2\exp(-2(\hthd^{\ag})^2\ag). 
\end{align*}
 Therefore, the fidelity of $\stgp_{\ag}$ satisfies
 \begin{align}
     \acc(\stgp_\ag) = & \sum_{\wos\in \Lowset}P_{\wos}\cdot \lp_{\wos}^{\bR}(\stgp_\ag) + \sum_{\wos\in \Highset}P_{\wos} \cdot \lp_{\wos}^{\bA}(\stgp_\ag)\\
     \ge & \sum_{\wos\in\Wosset} ( 1 - 2\exp(-2(\hthd^{\ag})^2\ag)\\
      =& 1 - 2\exp(-2(\hthd^{\ag})^2\ag).\label{eq:nb_acc_1} 
 \end{align}

\noindent{\bf Case 2: } $\liminf_{\ag\to\infty} \sqrt{\ag}\cdot \hthd^{\ag} <0$. In this case, there exists a constant $\eta <0$ and an infinite set $\negT$, s.t. for every $\ag\in \negT$, $\sqrt{\ag}\cdot \hthd^{\ag} \le \eta$. For every $\ag$, we define $\wos_{\ag}$ as follows:
$$\wos_{\ag} = \arg\min_\wos\left( \min_{\wos\in\Highset} (\hthd^\ag_{\wos\bA}), \min_{\wos\in\Lowset} (\hthd^\ag_{\wos\bR})\right).  $$
That is, $\wos_{\ag}$ is the world state whose $\hthd^\ag_{\wos_{\ag}, \bA}$ (if $\wos_\in\Highset$) or $\hthd^\ag_{\wos_{\ag}, \bR}$ (if $\wos_{\ag}\in\Lowset$) reach the minimum $\hthd^{\ag}$.  Note that for different $\ag$, $\wos_{\ag}$ may be different, and some of them will be in $\Highset$ while others may be in $\Lowset$. However, the reasoning for different $\wos_{\ag}$ will be the same. Consider the fidelity of $\stgp_{\ag}$ For some $\ag\in \negT$ when world state $\Wosrv = \wos_{\ag}$, and w.l.o.g suppose $\wos_{\ag}\in \Highset$:
\begin{align*}
    \lp_{\wos_{\ag}}^{\bA}(\stgp_{\ag})
    =& \Pr[\sum_{\sag = 1}^{\ag}\xrv_{\sag}^{\ag} - \sum_{\sag = 1}^{\ag}E[\xrv_{\sag}^{\ag}\mid \wos_{\ag}] \ge -\hthd^\ag_{\wos_{\ag} \bA} \cdot \ag\mid \wos_{\ag}]\\
    \le & \Pr[\sum_{\sag = 1}^{\ag}\xrv_{\sag}^{\ag} - \sum_{\sag = 1}^{\ag}E[\xrv_{\sag}^{\ag}\mid \wos_{\ag}] \ge -\eta \sqrt{\ag} \mid \wos_{\ag}]. 
\end{align*}
In this case we have $\eta<0$, thus $-\eta \sqrt{\ag}>0$. Therefore, we can directly apply the Hoeffding Inequality and get 
\begin{align*}
    \lp_{\wos_{\ag}}^{\bA}(\stgp_{\ag})
    \le & \Pr[\sum_{\sag = 1}^{\ag}\xrv_{\sag}^{\ag} - \sum_{\sag = 1}^{\ag}E[\xrv_{\sag}^{\ag}\mid \wos_{\ag}] \ge -\eta \sqrt{\ag} \mid \wos_{\ag}]\\
    \le & \exp(-2\eta^2).
\end{align*}
For every $\ag\in \negT$, we can prove $\lp_{\wos_{\ag}}^{\bA}(\stgp_{\ag}) \le  \exp(-2\eta^2)$ if $\wos_{\ag}\in \Highset$, or $\lp_{\wos_{\ag}}^{\bR}(\stgp_{\ag}) \le  \exp(-2\eta^2)$ if $\wos_{\ag}\in \Lowset$. Therefore, for every $\ag\in \negT$, the fidelity 
\begin{align*}
    \acc(\stgp_{\ag}) \le & 1 - P_{\wos_{\ag}} (1 - \exp(-2\eta^2))\\
    \le & 1 -  (1 - \exp(-2\eta^2)) \min_{\wos\in\Wosset}(P_{\wos}).
\end{align*}

\noindent{\bf Case 3: }$\liminf_{\ag\to\infty} \sqrt{\ag}\cdot \hthd^{\ag} \ge 0$. In this case, we apply the Berry-Esseen Theorem, which bounds the difference between the sum of random variables and normal distribution.
We define $\mathcal{N}$ and $\wos_{\ag}$ similarly as in Case 2. Suppose $\wos_{\ag} \in \Highset$. 
First we rewrite the form of $ \lambda_{\wos_{\ag}}^{\bA}(\stgp_{\ag})$. Let $\yrv_{\sag}^{\ag} = \xrv_{\sag}^{\ag} - E[\xrv_{\sag}^{\ag}\mid \wos_{\ag}]$, and $\vart = \sqrt{\sum_{\sag=1}^{\ag} Var(\yrv_{\sag}^{\ag}\mid \wos_{\ag})} = \sqrt{\sum_{\sag=1}^{\ag} Var(\xrv_{\sag}^{\ag}\mid\wos_{\ag})}$. We have
\begin{align*}
    \lp_{\wos_{\ag}}^{\bA}(\stgp_{\ag})
    \le & \Pr[\sum_{\sag = 1}^{\ag}\xrv_{\sag}^{\ag} - \sum_{\sag = 1}^{\ag}E[\xrv_{\sag}^{\ag}\mid \wos_{\ag}] \ge -\eta\sqrt{\ag} \mid \wos_{\ag}]
\end{align*}
Let $\ag'$ be the number of agents whose $Var(\xrv_{\sag}^{\ag} \mid \wos_{\ag}) > 0$. W.l.o.g, let $Var(\xrv_{\sag}^{\ag} \mid \wos_{\ag}) > 0$ for $\sag =1,2,\cdots, \ag'$, and $Var(\xrv_{\sag}^{\ag} \mid \wos_{\ag}) = 0$ for $\sag = \ag'+1, \ag'+2, \cdots, \ag$. From the assumption $\sum_{\sag = 1}^{\ag} Var(\xrv_{\sag}^{\ag}\mid \wos_{\ag}) \ge \psi\cdot \ag$ we know that $\ag' \ge \psi\cdot \ag$ (because $Var(\xrv_{\sag}^{\ag} \mid \wos_{\ag}) \le 1$). Since $Var(\xrv_{\sag}^{\ag} \mid \wos_{\ag}) = 0$ for $\sag = \ag'+1, \ag'+2, \cdots, \ag$, we have $\xrv_{\sag}^{\ag} = E[\xrv_{\sag}^{\ag}\mid \wos_{\ag}]$ condition on $\wos_{\ag}$ for these $\sag$. Therefore, we have 
\begin{align*}
    \lp_{\wos_{\ag}}^{\bA}(\stgp_{\ag})
    \le & \Pr[\sum_{\sag = 1}^{\ag}\xrv_{\sag}^{\ag} - \sum_{\sag = 1}^{\ag}E[\xrv_{\sag}^{\ag}\mid \wos_{\ag}] \ge -\eta\sqrt{\ag} \mid \wos_{\ag}]\\
    = & \Pr[\sum_{\sag = 1}^{\ag'}\xrv_{\sag}^{\ag} - \sum_{\sag = 1}^{\ag'}E[\xrv_{\sag}^{\ag}\mid \wos_{\ag}] \ge -\eta\sqrt{\ag} \mid \wos_{\ag}]
\end{align*}

Then we rewrite the formula in $\yrv_{\sag}^{\ag}$ and $\vart$: 
\begin{align*}
   \lp_{\wos_{\ag}}^{\bA}(\stgp_{\ag})
    \le & \Pr[\sum_{\sag = 1}^{\ag'}\xrv_{\sag}^{\ag} - \sum_{\sag = 1}^{\ag'}E[\xrv_{\sag}^{\ag}\mid \wos_{\ag}] \ge -\eta\sqrt{\ag} \mid \wos_{\ag}]\\
    = & \Pr\left[\left.\frac{\sum_{\sag = 1}^{\ag'}\yrv_{\sag}^{\ag}}{\vart} \ge -\frac{\eta\sqrt{\ag} }{\vart}  \right| \wos_{\ag}\right].
\end{align*}
Then we apply the Berry-Esseen Theorem and get
\begin{align*}
    \lambda_{\wos_{\ag}}^{\bA}(\stgp_{\ag}) \le 1 - \left(\Phi\left(-\frac{\eta\sqrt{\ag} }{\vart}\right) - C_0\cdot \frac{\sum_{\sag = 1}^{\ag'} E[|\yrv_{\sag}^{\ag}|^3\mid \wos_{\ag}]}{\left(\sum_{\sag = 1}^{\ag'}Var(Y_{\sag}^{\ag}\mid \wos_{\ag})\right)^{3/2}}\right),
\end{align*}
where $\Phi$ is the CDF of standard normal distribution, and $C_0 < 1$ is a constant from the theorem.
We deal with this inequality term by term to show that $\lambda_{\wos_{\ag}}^{\bA}(\stgp_{\ag})$ has a constant difference with 1. First, we show 
\begin{equation*}
    \Phi\left(\frac{\eta\sqrt{\ag} }{\vart}\right) \ge \Phi\left(-\frac{\eta\sqrt{\ag}}{\sqrt{\psi\ag}}\right) = \Phi\left(-\frac{\eta}{\sqrt{\psi}}\right).
\end{equation*}
This is because $\vart^2 = \sum_{\sag = 1}^{\ag} Var(\yrv_{\sag}^{\ag}\mid \wos_{\ag}) \ge \psi\cdot \ag$. The following table reveals  how things works. 
\begin{table}[htbp]
\centering
\begin{tabular}{@{}lll@{}}
\toprule
$\vart$                                                       & $\ge $ & $\sqrt{\psi\cdot \ag}$                                          \\ \midrule
$\frac{\eta\sqrt{\ag}}{\vart}$                   & $\le $ & $\frac{\eta\sqrt{\ag}}{\sqrt{\psi\ag}}$                   \\ \midrule
$-\frac{\eta\sqrt{\ag}}{\vart}$                  & $\ge $ & $-\frac{\eta\sqrt{\ag}}{\sqrt{\psi\ag}}$                  \\ \midrule
$\Phi\left(-\frac{\eta\sqrt{\ag}}{\vart}\right)$ & $\ge $ & $\Phi\left(-\frac{\eta\sqrt{\ag}}{\sqrt{\psi\ag}}\right)$ \\ \bottomrule
\end{tabular}
\caption{Comparison on the left and the right side.\label{tab:phi}}
\end{table}

Secondly, we start to deal with $$\frac{\sum_{\sag = 1}^{\ag'} E[|\yrv_{\sag}^{\ag}|^3\mid\wos_{\ag}]}{\left(\sum_{\sag = 1}^{\ag'}Var(Y_{\sag}^{\ag}\mid \wos_{\ag})\right)^{3/2}}.$$ Firstly, since $-1\le Y_{\sag} \le 1$, we have $E[|\yrv_{\sag}^{\ag}|^3\mid\wos_{\ag}]\le 1$. Therefore, 
\begin{equation*}
    \frac{\sum_{\sag = 1}^{\ag'} E[|\yrv_{\sag}^{\ag}|^3\mid\wos_{\ag}]}{\left(\sum_{\sag = 1}^{\ag'}Var(Y_{\sag}^{\ag}\mid \wos_{\ag})\right)^{3/2}} \le \frac{\ag'}{\left(\sum_{\sag = 1}^{\ag'}Var(Y_{\sag}^{\ag}\mid \wos_{\ag})\right)^{3/2}}.
\end{equation*}
Then, notice that since $\sigma^2_{\sag} = 0$ for $\sag = \ag'+1, \ag'+2, \cdots, \ag$, we have $\sum_{\sag = 1}^{\ag'}Var(Y_{\sag}^{\ag}\mid \wos_{\ag}) = \sum_{\sag = 1}^{\ag}Var(Y_{\sag}^{\ag}\mid \wos_{\ag}) \ge \psi\cdot \ag$. Therefore, 
\begin{equation*}
    \frac{\sum_{\sag = 1}^{\ag'} E[|\yrv_{\sag}^{\ag}|^3\mid\wos_{\ag}]}{\left(\sum_{\sag = 1}^{\ag'}Var(Y_{\sag}^{\ag}\mid \wos_{\ag})\right)^{3/2}} \le \frac{\ag'}{\left(\sum_{\sag = 1}^{\ag'}Var(Y_{\sag}^{\ag}\mid \wos_{\ag})\right)^{3/2}}\le \frac{\ag'}{(\psi\cdot \ag) ^{3/2}}. 
\end{equation*}
Finally, since $\ag'\le \ag$, and $\tc \ge \agc \ag -1$,
\begin{equation*}
    \frac{\sum_{\sag = 1}^{\ag'} E[|\yrv_{\sag}^{\ag}|^3\mid\wos_{\ag}]}{\left(\sum_{\sag = 1}^{\tc'}Var(Y_{\sag}^{\ag}\mid \wos_{\ag})\right)^{3/2}} \le \frac{\ag'}{(\psi\cdot \ag) ^{3/2}}\le \frac{1}{\sqrt{\psi^3\cdot \ag}}
\end{equation*}
Now we are ready to wrap things up: 
\begin{align*}
    \lambda_{\wos_{\ag}}^{\bA}(\stgp_{\ag}) \le& 1 - \left(\Phi\left(-\frac{\eta\sqrt{\ag} }{\vart}\right) - C_0\cdot \frac{\sum_{\sag = 1}^{\ag'} E[|\yrv_{\sag}^{\ag}|^3\mid \wos_{\ag}]}{\left(\sum_{\sag = 1}^{\ag'}Var(Y_{\sag}^{\ag}\mid \wos_{\ag})\right)^{3/2}}\right)\\
    \le & 1 - \left( \Phi\left(-\frac{\eta}{\sqrt{\psi}}\right) - C_0\cdot \frac{1}{\sqrt{\psi^3\cdot \ag}}  \right)\\
    \le & 1 - \left( \Phi\left(-\frac{\eta}{\sqrt{\psi}}\right) - \frac{1}{\sqrt{\psi^3\cdot \ag}}  \right)\\
\end{align*}
The last inequality comes from that $C_0 \le 0.5600$~\cite{Shevtsova10Improvement}.


The same reasoning also works for other $\ag\in\negT$. Therefore, let $\phi = \Phi\left(-\frac{\eta}{\sqrt{\psi}}\right)$, then there must exists a $\ag_\phi$, s.t. for all $(\ag > \ag_\phi) \wedge (\ag\in\negT)$, $\lambda_{\wos_{\ag}}^{\bA}(\stgp_{\ag}) \le 1 - \frac12 \phi$. Therefore, $\acc(\stgp_{\ag}) \le 1 - \frac12P_{\wos_{\ag}}\cdot \phi \le 1 - \frac12 \min_{\wos\in\Wosset} (P_{\wos})\cdot \phi$.
\end{proof}

Similarly to the binary setting, if we apply Theorem~\ref{thm:nb_acc} to each case of Theorem~\ref{thm:nb_arbitrary}, we get a criterion for judging whether a profile sequence is an equilibrium.

\begin{coro}
\label{coro:nb_arbeq}
Given an arbitrary sequence of instances and arbitrary sequence of regular strategy profiles $\{\stgp_{\ag}\}_{\ag=1}^\infty$ let $\hthd^\ag$ is defined for every $\stgp^\ag$.
\begin{itemize}
    \item If $\liminf_{\ag\to\infty} \sqrt{\ag}\cdot \hthd^{\ag} = +\infty$, then
    for all sufficiently large $\ag$, $\stgp_{\ag}$ is an $\varepsilon$-strong BNE with $\varepsilon = o(1)$.
    \item If $\liminf_{\ag\to\infty} \sqrt{\ag}\cdot \hthd^{\ag} < 0$ (including $-\infty$), there exists infinitely many $\ag$ such that $\stgp_\ag$ is not an $\varepsilon$-strong BNE with some constant $\varepsilon$. 
    \item If $\liminf_{\ag\to\infty} \sqrt{\ag}\cdot \hthd^{\ag} \ge 0$ (not including $+\infty$), and there exists a constant $\psi$ s.t. for every $\ag\in\negT$ and every $\wos\in \Wosset$, $Var(\sum_{\sag = 1}^{\ag} \xrv_{\sag}^{\ag} \mid \wos) \ge \psi\cdot \ag$, there exists infinitely many $\ag$ such that $\stgp_\ag$ is not an $\varepsilon$-strong BNE with some constant $\varepsilon$. 
\end{itemize}

\end{coro}

\subsection{Corollary~\ref{coro:nb_sym} (Corollary~\ref{coro:sym} for binary setting)}
\label{apx:nb_sym}
Finally, we can directly get a dichotomy for symmetric strategy profiles from Theorem~\ref{thm:nb_arbitrary}, by showing that every case of symmetric profiles falls into some case. 

Given arbitrary $\ag$ and a symmetric strategy profile induced by $\stg = (\bp_1, \bp_2. \cdots, \bp_\Sig)$ the \exshare{} is
\begin{align*}
    \hthd^{\ag}_{\wos\bA}
    = & \frac{1}{N}\sum_{\sag = 1}^{\ag} (\sum_{\sig =1}^{\Sig} P_{\sig \wos}\cdot \bp_{\sig}) -\Thd = \sum_{\sig =1}^{\Sig} P_{\sig \wos}\cdot \bp_{\sig} -\Thd,\\
    \hthd^{\ag}_{\wos\bR}
    = &\frac{1}{N}\sum_{\sag = 1}^{\ag}(\sum_{\sig =1}^{\Sig} P_{\sig \wos}\cdot (1-\bp_{\sig})) -(1-\Thd) = \sum_{\sig =1}^{\Sig} P_{\sig \wos}\cdot (1-\bp_{\sig}) -(1-\Thd), \\
    \hthd^{\ag} = & \min \left( \min_{\wos \in \Highset} \hthd^{\ag}_{\wos\bA}, \min_{\wos \in \Lowset} \hthd^{\ag}_{\wos\bR}\right). 
\end{align*}

An interesting observation is that the \exshare{} of symmetric profiles is independent of $N$. Therefore, for simplicity, we use $\hthd_{\wos\bA}, \hthd_{\wos\bR}$, and $\hthd$ to denote the \exshare{} for symmetric profiles. 

\begin{coro}
\label{coro:nb_sym}
    For an arbitrary strategy $\stg = (\bp_1, \bp_2. \cdots, \bp_\Sig)$ and an arbitrary sequence of instances, let $\{\stgp_{\ag}\}$ be the sequence of strategy profile $\stgp_{\ag}$ induced by $\stg$, and $\hthd$ be the \exshare{} of the strategy profiles. 
\begin{itemize}
    \item If $\hthd > 0$, then there exists a constant $\ag_0 >0$, s.t. for all $\ag > \ag_0$,  $\acc(\stgp_{\ag}) \ge 1 - 2\exp(-\frac12\hthd^2\ag)$.
    \item If $\hthd \le 0$, then there exist constants $\ag_0>0$ and $\eta' > 0$, s.t. for every $\ag > \ag_0$. $\acc(\stgp_{\ag})\le 1 - \eta'$. 
\end{itemize}
\end{coro}

\begin{proof}
    For $\hthd > 0$, we have $\liminf{\ag\to \infty} \sqrt{\ag} \cdot \hthd = + \infty$. Therefore, we can apply Case 1 of Theorem~\ref{thm:nb_acc} (Inequality~\ref{eq:nb_acc_1}), and have 
    \begin{equation*}
        \acc(\stgp_\ag) \ge 1 - 2\exp{(-2 \hthd^2 \ag)}. 
    \end{equation*}

    For $\hthd < 0$, we have $\liminf{\ag\to \infty} \sqrt{\ag} \cdot \hthd = - \infty$. Therefore, we can apply Case 2 of Theorem~\ref{thm:nb_acc} with any constant $\eta < 0$. Then there exists a $\ag_{\eta}$ such that for all $\ag > \ag_{\eta}$, we have $\acc(\stgp_\ag) \le 1 - (1 - \exp(-2\eta^2))\min_{\wos\in\Wosset}(P_{\wos})$. Therefore, let $\eta' = (1 - \exp(-2\eta^2))\min_{\wos\in\Wosset}(P_{\wos})$, we have $\acc(\stgp_\ag) \le 1 - \eta'$.  

    For $\hthd = 0$, we have $\liminf{\ag\to \infty} \sqrt{\ag} \cdot \hthd =0$. Therefore, we need to consider the variance.  $\sum_{\ag=1}^{\ag} Var(\xrv_{\sag}^{\ag}\mid\wos)$. For a single contingent agent $\sag$, given $\Wosrv = \wos$, we have 
\begin{align*}
    \Pr[\xrv_{\sag}^{\ag} =0\mid\wos] =&\sum_{\sig =1}^{\Sig} P_{\sig \wos}\cdot (1-\bp_{\sig})\\
    \Pr[\xrv_{\sag}^{\ag} = 1\mid\wos] =& \sum_{\sig =1}^{\Sig} P_{\sig \wos}\cdot \bp_{\sig}. 
\end{align*}
Therefore, we have $$Var(\xrv_{\sag}^{\ag} \mid\wos) = \left(\sum_{\sig =1}^{\Sig} P_{\sig \wos}\cdot (1-\bp_{\sig})\right)\left(\sum_{\sig =1}^{\Sig} P_{\sig \wos}\cdot \bp_{\sig}\right)$$
is a constant. Let $\psi = \min_{\wos\in\Wosset} \left(\sum_{\sig =1}^{\Sig} P_{\sig \wos}\cdot (1-\bp_{\sig})\right)\left(\sum_{\sig =1}^{\Sig} P_{\sig \wos}\cdot \bp_{\sig}\right)$, and we get for every $\wos\in\Wosset$, $\sum_{\ag=1}^{\ag} Var(\xrv_{\sag}^{\ag}\mid\wos)\ge \psi\cdot\ag.$ 

We show by contradiction that $\psi = 0$ will not happen. Suppose it is not the case and $\psi = 0$. Let $\wos_{\psi} = \arg\min_{\wos\in\Wosset} \left(\sum_{\sig =1}^{\Sig} P_{\sig \wos}\cdot (1-\bp_{\sig})\right)\left(\sum_{\sig =1}^{\Sig} P_{\sig \wos}\cdot \bp_{\sig}\right)$. Then we have  $$\left(\sum_{\sig =1}^{\Sig} P_{\sig \wos_\psi}\cdot (1-\bp_{\sig})\right)\left(\sum_{\sig =1}^{\Sig} P_{\sig \wos_\psi}\cdot \bp_{\sig}\right) = 0.$$ W.l.o.g, assume $\left(\sum_{\sig =1}^{\Sig} P_{\sig \wos_\psi}\cdot \bp_{\sig}\right) = 0$. Then we consider $\hthd_{\wos_\psi \bA}$, and have 
\begin{equation*}
        \hthd_{\wos_\psi \bA} = \agc\sum_{\sig =1}^{\Sig} P_{\sig \wos_\psi}\cdot \bp_{\sig} - \Thd = - \Thd < 0,
    \end{equation*}
which contradict with $\hthd =0$. Therefore, $\psi=0$ will not happen. 

Since we have guaranteed that $\psi > 0$, we can apply Theorem~\ref{thm:nb_arbitrary}. Let $\eta>0$ be any positive constant, We can apply Case 3 of Theorem~\ref{thm:nb_arbitrary} and get that there exists a $\ag_{\eta}$, s.t. for all $ \ag > \ag_{\eta}$, $\acc(\stgp_{\ag}) \le 1 - \frac12\min_{\wos\in\Wosset}(P_{\wos})^2 \cdot \Phi\left(-\frac{\eta}{\sqrt{\psi}}\right)$. Therefore, let $\eta' = 1 - \frac12\min_{\wos\in\Wosset}(P_{\wos})^2 \cdot \Phi\left(-\frac{\eta}{\sqrt{\psi}}\right)$, we have $\acc(\stgp_{\ag}) \le \eta'$ for all $\ag > \ag_{\eta}$. 
\end{proof}